\documentclass[a4paper,UKenglish,cleveref,
autoref, thm-restate]{article}
\usepackage{fullpage}
\usepackage[T1]{fontenc}

\usepackage{amsmath, amssymb,mathrsfs, amsthm}
\usepackage{mathtools}
\usepackage{graphicx}
\usepackage{tikz}
\usetikzlibrary{arrows.meta}
\usepackage{float}
\usepackage{etoolbox}
\usepackage{orcidlink}
\usepackage{todonotes}
\usepackage{thmtools} 
\newtheorem{theorem}{Theorem}
\newtheorem{lemma}{Lemma}
\newtheorem{proposition}{Proposition}
\newtheorem{corollary}{Corollary}
\newtheorem{definition}{Definition}

\usepackage{thm-restate}
\usepackage{amssymb}
\usepackage{hyperref}
\usepackage{xcolor} % For own colors
\usepackage{enumerate}
\usepackage{subcaption}
\usepackage{authblk}
\usetikzlibrary{snakes, fit}
\usetikzlibrary{shapes.geometric}
\usepackage{ifthen}
\usepackage{multirow}
\usepackage{fancybox,framed}

\newboolean{mainpart}
\setboolean{mainpart}{true} 

\hypersetup{colorlinks,
pdfstartview = FitH,
bookmarksopen = true,
bookmarksnumbered = true,
linkcolor = blue,
plainpages = false,
hypertexnames = false,
urlcolor = blue,
citecolor = blue}
\usepackage{cleveref}
\usepackage{comment}

\usetikzlibrary{decorations.pathreplacing}
\usetikzlibrary{shapes.misc}

\newcommand{\ffn}{\text{ffn}}
\newcommand{\mff}{\textsc{FireFighting}}
\newcommand{\mfft}{\textsc{FireFightingInTime}}
\newcommand{\mhun}{\textsc{Hunting}}

\newcommand{\hn}{\text{hn}}
\newcommand{\tw}{\text{tw}}
\newcommand{\vcn}{\text{vcn}}
\newcommand{\diam}{\text{diam}}

\newcommand{\pw}{\text{pw}}

\newcommand{\N}{\mathbb{N}}
\newcommand{\Z}{\mathbb{Z}}

\newcommand{\B}{\mathcal{B}}

\newcommand{\ZZ}{\mathrm{Z\kern-.3em\raise-0.5ex\hbox{Z}}}
\newcommand{\dist}{\text{dist}}

\newcommand{\hw}{\text{hw}}
\newcommand{\flips}{\text{flips}}

\definecolor{indiagreen}{rgb}{0.07, 0.53, 0.03}
\definecolor{dkteal}{RGB}{0,72,101}
\newcommand{\pathcolor}{blue}

\newif\ifshortversion
\shortversiontrue

\tikzset{cross/.style={cross out, draw=black, minimum size=2*(#1-\pgflinewidth), inner sep=0pt, outer sep=0pt},
cross/.default={1pt}}

\title{Complexity of Firefighting on Graphs}

\begin{document}
\date{}
\author[1]{Julius Althoetmar\orcidlink{0009-0001-0152-7408}}
\author[1]{Jamico Schade\orcidlink{0009-0005-2727-1163}}
\author[2]{Torben Schürenberg\orcidlink{0009-0006-5947-0172}}
\affil[1]{Technical University of Munich, Germany, \textit {\{julius.althoetmar,jamico.schade\}@tum.de}}
\affil[2]{University of Bremen, Germany, \textit{torsch@uni-bremen.de}}

\maketitle

\begin{abstract}
We consider a pursuit-evasion game that describes the process of extinguishing a fire burning on the nodes of an undirected graph.
We denote the minimum number of firefighters required by $\ffn(G)$ and provide almost sharp bounds to this graph parameter for complete binary trees.
We show that deciding whether $\ffn(G) \leq m$ for given $G$ and $m$ is \texttt{NP-hard}.
Furthermore, we show that shortest strategies can have superpolynomial length, leaving open whether the problem is in \texttt{NP}. We provide a construction that allows for transferring these results to a well-established Cops and Robbers variant called the ''Hunter and Rabbit game``.

% \keywords{Complexity \and Cops and Robbers \and Pursuit-Evasion}
\end{abstract}
% \clearpage
\section{Introduction}
We consider a game played on a simple undirected graph $G = (V, E)$.
At the start of the game, we imagine all nodes of the graph to be on fire.
A fixed number of firefighters are trying to extinguish the fire. Each round, every
firefighter can extinguish one freely chosen node (without any restrictions like moving along edges), but must then leave to gather more water. In their absence, the fire spreads: Each node with a burning neighbour catches fire again. In particular, this can include nodes that have just been extinguished.
We are interested in the smallest number of firefighters for which it is possible to extinguish the fire entirely, and call this number the \emph{firefighter number $\ffn(G)$}.

This problem was studied by Bernshteyn and Lee~\cite{Bernshteyn2021SearchingFA} in 2022.\footnote{They call it the ``inspection number'' of a graph.} 
They introduced a method of proving a lower bound on the firefighter number, which we improve, see Lemma \ref{lem:lowerbounds}.1.
Furthermore, they showed that complete binary trees can have an arbitrarily high firefighter number.
We extend this result by giving almost (up to an additive logarithmic term) tight bounds on the firefighter number of complete binary trees:

\setcounter{theorem}{6} 
\begin{theorem}[Bounds on the Firefighter Number of Binary Trees]
\label{thm:begin:binTree}
Let $\B_d$ be the complete binary tree of depth $d$.
We have $\ffn(\B_0)=1$, $\ffn(\B_{1})=\ffn(\B_{2})=2$, $\ffn(\B_3) = \ffn(\B_4) = 3$ and $\ffn(\B_5), \ffn(\B_6) \in \{3, 4\}$.
For all $d\in\N_{\geq 7}$ it holds that
\[
    \left\lfloor \frac{d - 1}{2}\right\rfloor - \frac12 \log\left( \left\lfloor \frac{d - 5}{2} \right\rfloor\right) - 2 < \ffn(\B_d) \leq \left\lceil \frac{d}{2} \right\rceil + 1.
\]
\end{theorem}

\noindent
The game can also be interpreted as a pursuit-evasion game, specifically as a cops and robbers game with helicopter cops and an entirely invisible, omniscient robber with bounded speed.
At any given time, the set of burning nodes corresponds to the set of still possible locations of the robber.
The cops are not forced to move along edges, which was the case in the initial variant that was introduced by To{\v{s}}i{\'c}~\cite{tovsic1985vertex}.
Pursuit-evasion games have been widely studied over the past few decades due to their broad range of applications.
For example, the fire fighting version could be used to model the propagation of diseases or pests throughout certain ecosystems, i.e. how bark beetles spread between forests, and provide strategies to counter this spread even with limited resources.

A good overview of results regarding pursuit-evasion games can be found in the surveys by Alspach~\cite{Alspach2006SEARCHINGAS,ALSPACH2008158}, Bonato and Nowakowski~\cite{bonatoBook}, Bonato and Yang~\cite{Bonato2013}, Fomin and Thilikos~\cite{Fomin2008AnAB}, and Hahn~\cite{hahnSurvey}.
For many variants, certain graph parameters, such as pathwidth or treewidth, can yield upper or lower bounds.
In Table~\ref{tab:basicProps}, we provide an overview of the relation between the firefighter number and some graph parameters.   

A closely related variant of the problem discussed in this paper is the Hunter and Rabbit game, where the fugitive is forced to move to a neighbouring node in each time step (inspired by a rabbit moving when startled by a gunshot) while the hunters have no restrictions on which nodes to shoot at in every round.
If the fugitive is also allowed to not move, this problem is equivalent to the firefighter variant.
The Hunter and Rabbit game was studied by different researchers, including Abramovskaya, Fomin, Golovach and Pilipczuk~\cite{Abramovskaya2015HowTH},
Bolkema and Groothuis~\cite{Bolkema2017HuntingRO},
Britnell and Wildon~\cite{Britnell2012FindingAP}, 
Gruslys and M'eroueh~\cite{Gruslys2015CatchingAM} and
Haslegrave~\cite{HASLEGRAVE20141}.
Among other results, they characterize the set of graphs with hunter number equal to one, and find the hunter number for certain graph classes.
Similarly, we give a characterization of graphs $G$ with $\ffn(G) = 1$ and $\ffn(G) = 2$ and determine the firefighter number for some graph classes in Section \ref{ch:basicProps}.

Apart from pursuit-evasion games, researchers are also studying some other problems related to a fire spreading on a graph:
The graph burning problem (see e.g.~\cite{graphBurning} for a survey) concerns itself with the quickest way to burn an entire graph, where the player is allowed to set a new node on fire each turn, after which the fire spreads.
In the firefighter problem (see e.g.~\cite{finbow} for a survey), a node of a graph is on fire and the fire spreads through the graph.
The firefighters can (permanently) save a number of not yet burned nodes each turn, with the goal of containing the fire and saving the most nodes from being burned.

Although the concept of a fire spreading on a graph in these problems is similar to the problem we are concerned with, the underlying rules and objectives are fundamentally different, so to the best of our knowledge, there are no transferable results.
In particular, if a node catches on fire, it is considered permanently burned in both of these problems, which gives them a monotonous nature.
By contrast, a node may catch on fire and be extinguished an arbitrary amount of times in our problem.

The decision variants of many pursuit-evasion games are \texttt{NP-hard}, see \cite{FOMIN20101167} and \cite{MAMINO201348}.
In this paper, we analyze the following two decision problems.

\begin{framed}
\begin{tabular}{l}
\textbf{(\mff): fire fighting} \\
\hspace{0.5cm}\textbf{Input:} A graph $G$ and $m\in \N_{>0}$.\\
\hspace{0.5cm}\textbf{Output:} Is $\ffn(G)\leq m$?
\end{tabular}
\end{framed}

\begin{framed}
\begin{tabular}{l}
\textbf{(\mfft): fire fighting within a given time horizon} \\
\hspace{0.5cm}\textbf{Input:} A graph $G$, $m\in \N_{>0}$ and $T\in \N_{>0}$.\\
\hspace{0.5cm}\textbf{Output:} Is $\ffn(G)\leq m$ when restricted to at most $T$ turns?
\end{tabular}
\end{framed}

\noindent In 2025, Ben-Ameur and Maddaloni~\cite{https://doi.org/10.1002/net.22284} proved that $\mfft$ is \texttt{NP-hard} via a reduction from the \textsc{3-Partition} problem.
Moreover, they proved that the Hunter and Rabbit game is \texttt{NP-hard} on digraphs.
In this variant, the movement of the rabbit is constrained to the direction of the edges.
Together with Gahlawat, they extended their proof for undirected graphs~ \cite{benameur2025huntingrabbithard}.
We prove that \mff ~is \texttt{NP-hard}, which also implies that \mff ~on digraphs is \texttt{NP-hard}, since each undirected graph can be interpreted as a digraph where each edge has a reverse counterpart.

\setcounter{theorem}{10} 
\begin{theorem}[Hardness of \mff]
    $\mff$ is \texttt{NP-hard}.
\end{theorem}

\noindent We provide a construction to show that Theorem \ref{thm:mffHard} implies \texttt{NP-hardness} of the Hunter and Rabbit game, thus giving an alternative proof to the result by Ben-Ameur, Gahlawat and Maddaloni in \cite{benameur2025huntingrabbithard}.
Additionally, we show that $\mfft$ is \texttt{NP-hard}, even on some heavily restricted graph classes.

\setcounter{theorem}{12} 
\begin{theorem}[Hardness of \mfft]
    The problem $\mfft$ is \texttt{NP-hard} even on trees. In particular, it is \texttt{NP-hard} even on trees with diameter at most $4$ and on spiders (trees where at most one node has a degree greater than $2$). 
\end{theorem}

\noindent Furthermore, we provide a class of graphs $G_m$ whose shortest strategies can have superpolynomial length $T(G_m)$, leaving open whether \mff~is in \texttt{NP}.

\setcounter{theorem}{17} 
\begin{theorem}[Shortest Strategies can have Superpolynomial Length]
    There is an infinite class of graphs $(G_m)_{m\in \N_{\geq 2}}$ where the shortest strategies with $\ffn(G_m)$ many firefighters have a length of at least $(m - 1)!$, which is superpolynomial in $\text{size}(G_m) = \mathcal{O}(m^6)$.
\end{theorem}

\noindent We can extend Theorem \ref{thm:length} to the Hunter and Rabbit game, answering one of the open questions in \cite{benameur2025huntingrabbithard} regarding a bound on the length of shortest hunter strategies. 

For all statements, the full rigorous proofs omitted in the main paper can be found in Appendix \ref{app}.
\section{Model}
\setcounter{theorem}{0} 
Any graphs mentioned in this paper are assumed to be simple and undirected.
For a graph $G = (V, E)$ and a set $W \subseteq V$, we define the neighbourhood $N(W)$ as the set of nodes in $V \setminus W$ that are adjacent to at least one node in $W$.
To keep the notation concise, we will sometimes refer to the node set of a graph $G$ as $G$ when it is clear from the context.
For a number $n \in \Z$, we set $[n] = \{1, \dots, n\}$ and $[n]_0 = \{0,1, \dots, n\}$.

\begin{figure}[h]
	\renewcommand{\arraystretch}{0.7}
	\centering
	\begin{tabular}{|c|c|}
    \hline
		\begin{tikzpicture}[scale=0.3, every node/.style={font=\small}]
			
			% Add extra invisible node to increase height and width of image
			\node[circle, fill=white!20] (v0) at (1,0.5) {};
                \node[circle, fill=white!20] (v0) at (-5.5,0) {};
                \node[circle, fill=white!20] (v0) at (7.5,0) {};
			% Define styled nodes
			\node[draw, circle, fill=red!40] (v1) at (1,0) {};
			
			\node[draw, circle, fill=white] (v2) at (-3,-1.5) {};
			\node[draw, circle, fill=red!40] (v3) at (5,-1.5) {};
			
			\node[draw, circle, fill=white] (v4) at (-5,-3) {};
			\node[draw, circle, fill=white] (v5) at (-1,-3) {};
			
			\node[draw, circle, fill=red!40] (v6) at (3,-3) {};
			\node[draw, circle, fill=white] (v7) at (7,-3) {};
			% Connect nodes with styled edges
			\draw[thick] (v1) -- (v2);
			\draw[thick] (v1) -- (v3);
			
			\draw[thick] (v2) -- (v4);
			\draw[thick] (v2) -- (v5);
			
			\draw[thick] (v3) -- (v6);
			\draw[thick] (v3) -- (v7);

		\end{tikzpicture}&\raisebox{0.15cm}{ % Adjust this value as needed
		\begin{tikzpicture}[scale=0.4, every node/.style={font=\normalsize}]
			% Add a legend
			\node[text width=6.5cm, right] at (0,0) { 
				\tikz[baseline=-0.8ex]\draw[fill=red!40] (0,0) circle (1.2ex); Burning nodes $B_t$ \\
				\tikz[baseline=-0.8ex]\draw[fill=white] (0,0) circle (1.2ex); Extinguished nodes $E_t$ \\
			};
		\end{tikzpicture}
		}\\
		\hline
		\multicolumn{2}{|c|}{\textdownarrow \textcolor{white}{|}Extinguishing \textdownarrow\hspace{6.87cm}\textcolor{white}{a}} \\		
		\hline
		\begin{tikzpicture}[scale=0.3, every node/.style={font=\footnotesize}]
			% Add extra invisible node to increase height of image
			\node[circle, fill=white!20] (v0) at (1,0.5) {};
			% Define styled nodes
			\node[draw, circle, fill=red!40] (v1) at (1,0) {};
			
			\node[draw, circle, fill=white] (v2) at (-3,-1.5) {};
            
			\node[draw=dkteal, thick, circle, fill=white, line width=1mm] (v3) at (5,-1.5) {};
			
			\node[draw, circle, fill=white] (v4) at (-5,-3) {};
			\node[draw, circle, fill=white] (v5) at (-1,-3) {};
			
			\node[draw=dkteal, thick, circle, fill=white, line width=1mm] (v6) at (3,-3) {};
			\node[draw, circle, fill=white] (v7) at (7,-3) {};
			
			% Connect nodes with styled edges
			\draw[thick] (v1) -- (v2);
			\draw[thick] (v1) -- (v3);
			
			\draw[thick] (v2) -- (v4);
			\draw[thick] (v2) -- (v5);
			
			\draw[thick] (v3) -- (v6);
			\draw[thick] (v3) -- (v7);
			\end{tikzpicture}&\raisebox{0.01cm}{ % Adjust this value as needed
			\begin{tikzpicture}[scale=0.5, every node/.style={font=\normalsize}]
				% Add a legend
				\node[text width=6.5cm, right] at (0,-2) { % Adjust position as needed
					%			\textbf{Legend:} \\
					\tikz[baseline=-0.6ex]\draw[fill=red!40] (0,0) circle (1.2ex); Provisionally burning nodes $\Tilde{B}_{t + 1}$ \\
					\tikz[baseline=-0.6ex]\draw[fill=white] (0,0) circle (1.2ex); Provisionally extinguished nodes $\Tilde{E}_{t + 1}$ \\
					\hspace{0.19mm}\tikz[baseline=-0.6ex]\draw[dkteal, thick, fill=white!40, line width=1mm] (0,0) circle (1ex); Firefighter set $F_{t + 1}$ \\
				};
			\end{tikzpicture}
		}\\
		\hline 
		\multicolumn{2}{|c|}{\textdownarrow \textcolor{white}{|}Propagation \textdownarrow\hspace{6.84cm}\textcolor{white}{a}} \\		
		\hline
		\begin{tikzpicture}[scale=0.3, every node/.style={font=\tiny}]
			% Add extra invisible node to increase height of image
			\node[circle, fill=white!20] (v0) at (1,0.5) {};
                \node[circle, fill=white!20] (v0) at (-5.5,0) {};
                \node[circle, fill=white!20] (v0) at (7.5,0) {};
			% Define styled nodes
			\node[draw, circle, fill=red!40] (v1) at (1,0) {};
			
			\node[draw, circle, fill=red!40] (v2) at (-3,-1.5) {};
			\node[draw, circle, fill=red!40] (v3) at (5,-1.5) {};
			
			\node[draw, circle, fill=white] (v4) at (-5,-3) {};
			\node[draw, circle, fill=white] (v5) at (-1,-3) {};
			
			\node[draw, circle, fill=white] (v6) at (3,-3) {};
			\node[draw, circle, fill=white] (v7) at (7,-3) {};
			
			% Connect nodes with styled edges
			\draw[thick] (v1) -- (v2);
			\draw[thick] (v1) -- (v3);
			
			\draw[thick] (v2) -- (v4);
			\draw[thick] (v2) -- (v5);
			
			\draw[thick] (v3) -- (v6);
			\draw[thick] (v3) -- (v7);

		\end{tikzpicture}&\raisebox{0.15cm}{ % Adjust this value as needed
			\begin{tikzpicture}[scale=0.5, every node/.style={font=\normalsize}]
				% Add a legend
				\node[text width=6.5cm, right] at (10,-2) { % Adjust position as needed
					%\textbf{Legend:} \\
					\tikz[baseline=-0.8ex]\draw[fill=red!40] (0,0) circle (1.2ex); Burning nodes $B_{t+1}$ \\
					\tikz[baseline=-0.8ex]\draw[fill=white] (0,0) circle (1.2ex); Extinguished nodes $E_{t+1}$ \\
				};
			\end{tikzpicture}
		}\\
		\hline
	\end{tabular}
	\caption[]{Two firefighters try to extinguish a partially burning tree. Visualization of the extinguishing and propagation process of the fire.}
	\label{fig:SimpleExample}
\end{figure}
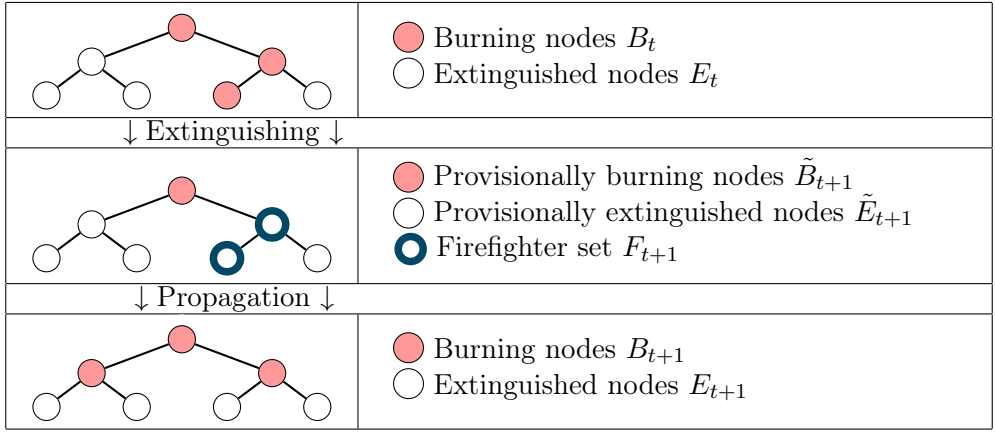

Let us now introduce some basic notation for this game.
A vector $S=(F_1,\ldots,F_T)$ with $F_i\subseteq V$ and $|F_i|\leq m$ for all $i\in [T]$ is called an \emph{$m$-strategy} for $G$ of length $T$. $F_i$ is called the \emph{firefighter set} at time $i$.
Given an $m$-strategy $S=(F_1,\ldots,F_T)$, we define the set of \emph{burning nodes $B_t$} at time $t$ iteratively by setting
$B_0\coloneqq V$ and $B_{t}\coloneqq (B_{t-1} \setminus F_{t}) \cup N(B_{t-1} \setminus F_{t})$ for $t \geq 1$, where $F_t=\emptyset$ for $t > T$.
Furthermore, let the set of \emph{extinguished nodes $E_t$} at time $t \in \N_{\geq 0}$ be defined as $E_t\coloneqq V \setminus B_t$.
For convenience, we write $\Tilde{B}_t$ to denote $B_{t-1} \setminus F_t$, the \emph{provisionally burning nodes}, as well as $\tilde{E}_t$ to denote $V \setminus \Tilde{B}_t = E_t \cup F_t$, the \emph{provisionally extinguished nodes}.
A simple example of this extinguishing and fire propagation process is shown in Figure \ref{fig:SimpleExample}.
An $m$-strategy for $G$ is called a \emph{$T$-winning $m$-strategy} if $B_T= \emptyset$ (or simply a \emph{winning $m$-strategy} if $T$ does not need to be specified).
If there exists such a $T$-winning $m$-strategy for $G$, the graph $G$ is called \emph{$m$-winning} or, more precisely, \emph{$m$-winning in time $T$}.

We denote the shortest possible length of a winning $m$-strategy for $G$ by $T_m(G)$.
If $\ffn(G) > m$, we set $T_m(G) = \infty$.
We set $T(G) \coloneqq T_{\ffn(G)}(G)$ to be the length of the shortest possible winning strategy when using the smallest possible number of firefighters.

\section{Basic Properties and Bounds}\label{ch:basicProps}
We characterize the classes of graphs with firefighter number one and two, and give some results for specific graph classes.
Furthermore, Table \ref{tab:basicProps} gives a short overview of some common graph parameters that can be used to find bounds on the firefighter number.

\begin{restatable}[Characterization: $\ffn(G)=1$ and $\ffn(G)=2$]{corollary}{lemmaChar}\label{prop:char}
\noindent For a graph $G = (V, E)$, it holds that $\ffn(G)=1$ if and only if $|E|=0$ and $\ffn(G)=2$ if and only if $|E|>0$ and any connected component of $G$ is a caterpillar graph, i.e., a tree in which all the nodes are within distance 1 of a central path.
\end{restatable}

\begin{restatable}[Firefighter Number of $K_n$, $C_n$ and $K_{n, m}$]{corollary}{corExactFfn}
\label{prop:graphs}
   $\ffn(K_n)=n$, $\ffn(K_{n,m})=\min\{n,m\}+1$ and $\ffn(C_n)=3$ for any $n,m\in\N_{>0}$, where $K_n$ is a complete graph, $K_{n,m}$ is a complete bipartite graph and $C_n$ is a circular graph.
\end{restatable}
\begin{restatable}[$d$-regular Graphs]{corollary}{corDRegular}
\label{prop:dRegular}
    Let $d\in \N_{>0}$. Every $d$-regular graph $G$ fulfills $\ffn(G)\geq d+1$. This bound is tight. For $d\in \{1,2\}$ we have $\ffn(G)=d+1$. For any $d\geq 3$, the firefighter number can reach arbitrarily high values.
\end{restatable}

\begin{restatable}[Order of a Forest]{corollary}{forst}
\label{prop:forst}
     For every forest $F = (V, E)$, we have $\ffn(F)\leq \log_3(2|V| + 1)+2$.
\end{restatable}

\noindent These results are direct consequences of the following useful lemmata providing bounds on the firefighter number.
In particular, Lemma \ref{lem:lowerbounds}.1 will be an essential tool for proving theorems in later sections. It is a strictly stronger version of a criterion given in \cite{Bernshteyn2021SearchingFA}.
In their version, it is necessary to prove that for some $i \in [|V|]$, any $W \subseteq V$ with $i - m + 1 < |W| < i$ fulfills $|N(W)| \geq m - 1$, in order to prove $\ffn(G)\geq m$.
In our version, one only has to consider subsets $W$ with $|W| = i$ instead of subsets $W$ with an entire range of sizes, which makes applying the lemma significantly easier.
As this Lemma is one of the only known techniques to show lower bounds, this is a valuable improvement.
Note that Lemma \ref{lem:lowerbounds}.1 can give arbitrarily bad lower bounds, as shown in \ref{app:basicProps}.

\begin{restatable}[Lower Bounds]{lemma}{lemmaLowerBounds}
\label{lem:lowerbounds}
    Let $G=(V,E)$ be a graph and $m\in\N_{>0}$. Then any of the following conditions imply that $\ffn(G)\geq m$:
    \begin{enumerate}
        \item There exists an $i \in [|V| - m + 1]$ such that any $W \subseteq V$ with $|W| = i$ fulfills $|N(W)| \geq m - 1$.
        \item $\delta_{\min}(G) \geq m-1$, i.e., each node in $G$ has at least $m-1$ neighbours.
        \item There exists a $G'\subseteq G$ with $\ffn(G')\geq m$.
        \item $|E| \geq m \cdot (|V| - \frac{m+1}{2})$.
    \end{enumerate}
\end{restatable}
\noindent\begin{proof}[Proof of 1]
    Assume that there is a winning $(m-1)$-strategy $(F_1, \dots, F_T)$ for $G$, while an $i$ as in the statement of the lemma exists.
    Then, there has to be a smallest $t$ such that $|B_t| < i + (m-1)$.
    Note that $|B_0| = |V| \geq i + (m-1)$, hence $t \geq 1$.
    By the definition of $t$, we have $|B_{t - 1}| \geq i + (m-1)$.
    Since $|F_t| \leq m-1$, it follows that $|\tilde{B}_t| \geq i$.
    In particular, this implies the existence of a set $\tilde{W} \subseteq \tilde{B}_t$ with $|\tilde{W}| = i$.
    Let $W \coloneqq \tilde{W} \cup N(\tilde{W})$, i.e., the set of burning nodes after the fire spreads from the provisionally burning nodes $\tilde{W}$.
    Since $|N(\tilde{W})| \geq m-1$ by the definition of $i$, we have $|W| \geq |\tilde{W}| + (m-1) = i + (m-1)$. $\tilde{W}$ being a subset of $\tilde{B}_t$ implies $W \subseteq B_t$ and hence, we have $|B_t| \geq i + (m-1)$, a contradiction.
\end{proof}

\begin{restatable}[Upper Bounds]{lemma}{lemmaUpperBounds}
\label{lem:upperbounds}
    Let $G=(V,E)$ be a graph and $m\in\N_{>0}$. Then any of the following conditions imply that $\ffn(G)\leq m$:
    \begin{enumerate}
        \item $\pw(G) +1\leq m$, where $\pw(G)$ denotes the pathwidth of $G$.\ifthenelse{\boolean{mainpart}}{\footnote{Included just for the sake of completeness, has already been proven in \cite{Bernshteyn2021SearchingFA}. }}{} 
        \item $G$ is a tree with $\diam(G) \leq 2m-2$.
        \item There exists a graph $G_k=(V_k,E_k)$ for $k\in \N_{\geq 0}$ with $\ffn(G_k)\leq m-k$, $G_k\subseteq G$ and $|V_k|\geq |V|-k$.
    \end{enumerate}
\end{restatable}
\begin{table}[htb]
    \centering
    \begin{tabular}{c||c|c}
        Graph parameter &  \;Upper bound on $\ffn$ \;&  \;Lower bound on $\ffn$  \;\\
        \hline \hline
        \;Minimum degree \;$\delta_{\min}$& - & $\delta_{\min} + 1\leq \ffn(G)$
        \\
        \hline
         \;Maximum degree  \;$\delta_{\max}$ & - & - \\
        \hline
        Order $|V|$ & $\ffn(G) \leq |V|$ & - \\
        \hline
        Depth $d$ of a tree $T$ & $\ffn(T) \leq d + 1$ & - \\
        \hline
        Beta index $|E|/|V|$ & - & \;      
        $|E| \leq (\ffn(G) - 1) \cdot (|V| - \ffn(G)/2)$  \;\\
        \hline
        Pathwidth $\pw(G)$ & $\ffn(G) \leq \pw(G) + 1$ & - \\
        \hline
        Treewidth $\tw(G)$ & - & ? \\
        \hline
        Vertex cover number& \multirow{2}{*}{$\ffn(G) \leq \vcn(G) + 1$} & \multirow{2}{*}{-} \\
        $\vcn(G)$& & \\
    \end{tabular}
    \vspace{10pt}
    \caption{Tight upper and lower bounds on the firefighter number in terms of some graph parameters.
    For every missing entry except for the question mark, we prove that a bound based on this parameter is not possible.
    Proofs based on Lemma \ref{lem:lowerbounds}, Lemma \ref{lem:upperbounds}, and~\cite{Bernshteyn2021SearchingFA} for the entries are given in \ref{app:basicProps}.
    }
    \label{tab:basicProps}
\end{table}

\section{Firefighting on Complete Binary Trees}\label{ch:binaryTrees}

A complete binary tree is a rooted tree in which the distance of any leaf node to the root is the same and every non-leaf node has exactly two children.
In~\cite{Bernshteyn2021SearchingFA} it was shown that for every $k$ there exists a complete binary tree $T=(V,E)$ such that $\ffn(T)>k$.
We extend this result by giving almost tight bounds on the firefighter number of complete binary trees.\footnote{For a complete binary tree $\B = (V, E)$, the difference between our upper and lower bound is in $\mathcal{O}(\log(\log(|V|)))$.}

\begin{restatable}[Bounds on the Firefighter Number of Binary Trees]{theorem}{thmBinTree}
\label{thm:binTree}
Let $\B_d$ be the complete binary tree of depth $d$.
We have $\ffn(\B_0)=1$, $\ffn(\B_{1})=\ffn(\B_{2})=2$, $\ffn(\B_3) = \ffn(\B_4) = 3$ and $\ffn(\B_5), \ffn(\B_6) \in \{3, 4\}$.
For all $d\in\N_{\geq 7}$ it holds that
\[
    \left\lfloor \frac{d - 1}{2}\right\rfloor - \frac12 \log\left( \left\lfloor \frac{d - 5}{2} \right\rfloor\right) - 2 < \ffn(\B_d) \leq \left\lceil \frac{d}{2} \right\rceil + 1.
\]
\end{restatable}
\begin{proof}[Proof sketch]
To prove the upper bound, we provide a construction to extend a given winning $k$-strategy for $\B_{2d}$ to a winning $(k + 1)$-strategy for $\B_{2d + 1}$ that extinguishes the root node of $\B_{2d + 1}$ in each step.
Based on this specific strategy for $\B_{2d + 1}$, we then construct a winning $(k + 1)$-strategy for $\B_{2d + 2}$.

We prove the lower bound by showing that we can apply Lemma \ref{lem:lowerbounds}.1 for some adequately chosen $i$, i.e., that any subset $W\subseteq \B_d$ with $|W| = i$ has at least $\left\lfloor \frac{d-1}{2} \right\rfloor -   \frac12 \log_{2}\left(\left\lfloor\frac{d-5}{2}\right\rfloor \right) - 2$ neighbours.
To that end, we observe that any subset $W$ of nodes in $\B_d$ can be constructed by iteratively adding or removing complete binary trees of decreasing size, as shown in Figure~\ref{fig:treeDecomposition}.

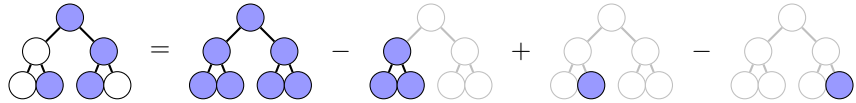
\begin{figure}[htb]
    \centering
    \begin{tikzpicture}[scale=0.3, every node/.style={font=\small}]
			\node[draw, circle, fill=blue!40] (v1) at (0,0) {};
			
			\node[draw, circle, fill=white] (v2) at (-1.5,-1.5) {};
			\node[draw, circle, fill=blue!40] (v3) at (1.5,-1.5) {};
			
			\node[draw, circle, fill=white] (v4) at (-2.1,-3) {};
			\node[draw, circle, fill=blue!40] (v5) at (-0.9,-3) {};
			\node[draw, circle, fill=blue!40] (v6) at (0.9,-3) {};
			\node[draw, circle, fill=white] (v7) at (2.1,-3) {};
			% Connect nodes with styled edges
			\draw[thick] (v1) -- (v2);
			\draw[thick] (v1) -- (v3);
			
			\draw[thick] (v2) -- (v4);
			\draw[thick] (v2) -- (v5);
			
			\draw[thick] (v3) -- (v6);
			\draw[thick] (v3) -- (v7);

                \node[] (eq) at (4,-1.5) {$=$};

                % Tree 2
                \node[draw, circle, fill=blue!40] (vv1) at (8,0) {};
			
			\node[draw, circle, fill=blue!40] (vv2) at (6.5,-1.5) {};
			\node[draw, circle, fill=blue!40] (vv3) at (9.5,-1.5) {};
			
			\node[draw, circle, fill=blue!40] (vv4) at (5.9,-3) {};
			\node[draw, circle, fill=blue!40] (vv5) at (7.1,-3) {};
			\node[draw, circle, fill=blue!40] (vv6) at (8.9,-3) {};
			\node[draw, circle, fill=blue!40] (vv7) at (10.1,-3) {};
                \draw[thick] (vv1) -- (vv2);
			\draw[thick] (vv1) -- (vv3);
			\draw[thick] (vv2) -- (vv4);
			\draw[thick] (vv2) -- (vv5);
			\draw[thick] (vv3) -- (vv6);
			\draw[thick] (vv3) -- (vv7);

                \node[] (eq) at (12,-1.5) {$-$};

                % Tree 3
                \node[draw, gray!50, circle, fill=white] (vv1) at (16,0) {};
			
			\node[draw, circle, fill=blue!40] (vv2) at (14.5,-1.5) {};
			\node[draw, gray!50, circle, fill=white] (vv3) at (17.5,-1.5) {};
			
			\node[draw, circle, fill=blue!40] (vv4) at (13.9,-3) {};
			\node[draw, circle, fill=blue!40] (vv5) at (15.1,-3) {};
			\node[draw, gray!50, circle, fill=white] (vv6) at (16.9,-3) {};
			\node[draw, gray!50, circle, fill=white] (vv7) at (18.1,-3) {};
                \draw[thick, gray!50] (vv1) -- (vv2);
			\draw[thick, gray!50] (vv1) -- (vv3);
			\draw[thick] (vv2) -- (vv4);
			\draw[thick] (vv2) -- (vv5);
			\draw[thick, gray!50] (vv3) -- (vv6);
			\draw[thick, gray!50] (vv3) -- (vv7);

                \node[] (eq) at (20,-1.5) {$+$};

                % Tree 4
                \node[draw, gray!50, circle, fill=white] (vv1) at (24,0) {};
			
			\node[draw, gray!50, circle, fill=white] (vv2) at (22.5,-1.5) {};
			\node[draw, gray!50, circle, fill=white] (vv3) at (25.5,-1.5) {};
			
			\node[draw, gray!50, circle, fill=white] (vv4) at (21.9,-3) {};
			\node[draw, circle, fill=blue!40] (vv5) at (23.1,-3) {};
			\node[draw, gray!50, circle, fill=white] (vv6) at (24.9,-3) {};
			\node[draw, gray!50, circle, fill=white] (vv7) at (26.1,-3) {};
                \draw[thick, gray!50] (vv1) -- (vv2);
			\draw[thick, gray!50] (vv1) -- (vv3);
			\draw[thick, gray!50] (vv2) -- (vv4);
			\draw[thick, gray!50] (vv2) -- (vv5);
			\draw[thick, gray!50] (vv3) -- (vv6);
			\draw[thick, gray!50] (vv3) -- (vv7);

                \node[] (eq) at (28,-1.5) {$-$};

                % Tree 5
                \node[draw, gray!50, circle, fill=white] (vv1) at (32,0) {};
			
			\node[draw, gray!50, circle, fill=white] (vv2) at (30.5,-1.5) {};
			\node[draw, gray!50, circle, fill=white] (vv3) at (33.5,-1.5) {};
			
			\node[draw, gray!50, circle, fill=white] (vv4) at (29.9,-3) {};
			\node[draw, gray!50,  circle, fill=white] (vv5) at (31.1,-3) {};
			\node[draw, gray!50, circle, fill=white] (vv6) at (32.9,-3) {};
			\node[draw, circle, fill=blue!40] (vv7) at (34.1,-3) {};
                \draw[gray!50, thick] (vv1) -- (vv2);
			\draw[gray!50, thick] (vv1) -- (vv3);
			\draw[gray!50, thick] (vv2) -- (vv4);
			\draw[gray!50, thick] (vv2) -- (vv5);
			\draw[gray!50, thick] (vv3) -- (vv6);
			\draw[gray!50, thick] (vv3) -- (vv7);
    \end{tikzpicture}		
    \caption[]{Decomposing a subset of a complete binary tree into multiple complete binary trees of varying depths.}
    \label{fig:treeDecomposition}
\end{figure}

The number of binary trees that is used in such a decomposition corresponds to the number of edges between $W$ and the rest of the graph, which can be bounded in terms of the number of neighbours of $W$.
Since the number of nodes in a complete binary tree equals $2^x - 1$ for some $x \in \N_{>0}$, we can use this decomposition to express the size of $W$ as a sum of powers of two and their negatives, where the total number of summands is again bounded in terms of the number of neighbours of $W$.
By applying some concepts from information theory (in particular the Hamming weight of the binary representation of a number), we show that representing $i$ as such a sum requires a certain number of summands, which then implies that any $W \subseteq \B_d$ with $|W| = i$ has to have at least $\left\lfloor \frac{d-1}{2} \right\rfloor -   \frac12 \log_{2}\left(\left\lfloor\frac{d-5}{2}\right\rfloor \right) - 2$ neighbours, finishing the proof.
\end{proof}
    
\section{NP-Hardness}\label{ch:npHardness}
%It was recently shown in \cite{digraphHardness} that it is \texttt{NP-hard} to determine the firefighter number for digraphs.
%Their construction heavily relies on directed edges, and we are not aware of a way to use their techniques for undirected graphs.
%Nevertheless, in this section, we use a novel construction to show that both $\mff$ and $\mfft$ are \texttt{NP-hard} on undirected graphs as well.
Ben-Ameur and Maddaloni recently proved that \mfft ~is \texttt{NP-hard} (via a reduction from the \textsc{3-Partition} problem) in~\cite{https://doi.org/10.1002/net.22284}.
Moreover, they proved that a directed variant of the Hunter and Rabbit game where the rabbit moves on a digraph is also \texttt{NP-hard}.
However, their proof cannot be generalized to the undirected setting.
Nevertheless, in this section, we use a novel construction to show that both $\mff$ and $\mfft$ are \texttt{NP-hard} on undirected graphs as well.
Proving that \mff ~is \texttt{NP-hard} also implies that \mff ~on digraphs is \texttt{NP-hard}, since each undirected graph can be interpreted as a digraph where each edge has a reverse counterpart. We also improve on their results by showing that \mfft ~is \texttt{NP-hard} even on trees with diameter at most 4 and on spiders (i.e. trees where at most one node has a degree greater than 2).

In order to show \texttt{NP-hardness} of $\mff$, we will build a gadget $H(G,T)$ for an arbitrary graph $G$ such that there is a $T$-winning $m$-strategy for $G$ if and only if $\ffn(H(G,T))\leq 4m$.\footnote{Note that this construction cannot serve as a polynomial reduction from $\mfft$ to $\mff$, since the size of $H(G, T)$ will be quadratic in $T$ and therefore not polynomial in the encoding size of the $\mfft$-instance $(G, m, T)$.}
The main idea of this gadget is to attach the $2$-blowup $\mathbb{G}$ of $G$ (i.e. the graph that arises by replacing every node of $G$ with a $2$-clique and adding edges between all nodes of two cliques if the original two nodes of $G$ were connected) to a circular structure consisting of a lower and an upper part, which serve as a timed fuse and as an interface between $\mathbb{G}$ and the fuse.
As a result, the only way to extinguish the whole graph with $4m$ firefighters is to first extinguish the entire lower part, then move along the circle, extinguish $\mathbb{G}$ as fast as possible, and finally catch the fire that spreads in the lower part, right in time before it spreads to the upper part again.
This is exactly possible if there is a $T$-winning $m$-strategy for $G$.
To utilize this gadget we can use the \texttt{strongly NP-complete} \textsc{BinPacking} problem~\cite{10.5555/574848} to build a graph such that there is a $T$-winning $m$-strategy if and only if the \textsc{BinPacking} instance is a yes-instance, which proves that \textsc{BinPacking} can be reduced to $\mff$. Thus, we can conclude that \mff ~is \texttt{NP-hard}.

\begin{definition}[Gadget $H(G,T)$]
    Let $G$ be an arbitrary graph, $m\in \N_{>0}$ and $T\in \N_{\geq 2}$.
    The graph $H(G,T)$ is defined by a block structure, where the blocks $X,Y,Z$ are $2m$-cliques, $\mathbb{G}$ is the $2$-blowup of $G$, and all other blocks are $m$-cliques.
    Finally, we add an edge between every pair of nodes from two different blocks from the following list of block combinations: $(\mathbb{G}, Y), (A,X), (X,Y), (Y,Z), (Z,B)$, $(A,P_i^1)$ for all $i\in [2T+2]$, $(B,P_i^{T+1})$ for all $i\in [2T+2]$ and $(P_i^{j},P_i^{j+1})$ for all $i\in [2T+2]$ and $j\in [T]$.
\end{definition}

\noindent See Figure \ref{fig:TimeGadget} for a visualization of this gadget.
We further define the $i$-th path $P_i \coloneqq \bigcup_{j \in [T + 1]} P_i^{j}$ and the set of all paths as $\mathcal{P}\coloneqq \bigcup_{i\in[2T+2]} P_i$.

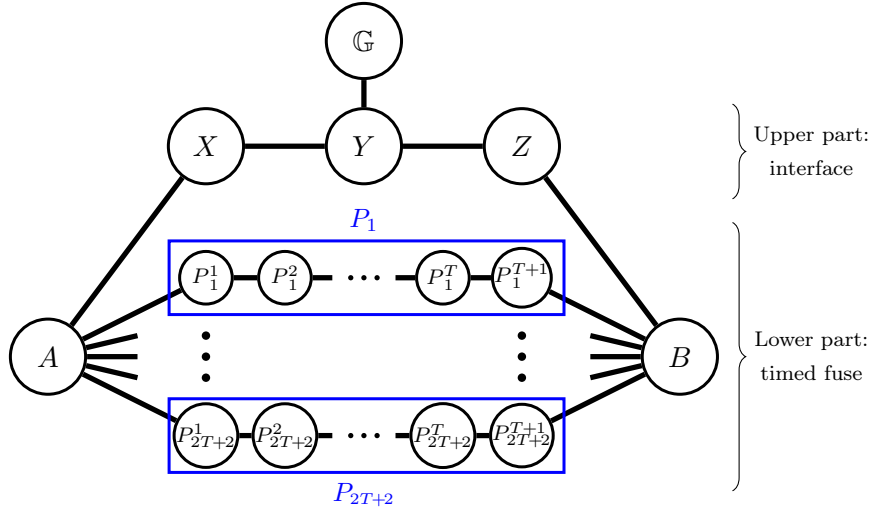
\begin{figure}[]
    \centering
    \begin{tikzpicture}[scale=0.7, every node/.style={font=\scriptsize}]
        %Upper part
        \node[draw, line width=0.4mm, circle, minimum size=1.0cm, inner sep=0pt] (A) at (-6,0) {\normalsize $A$};
        \node[draw, line width=0.4mm, circle, minimum size=1.0cm, inner sep=0pt] (B) at (6,0) {\normalsize$B$};
        \node[draw, line width=0.4mm, circle, minimum size=1.0cm, inner sep=0pt] (X) at (-3,4) {\normalsize$X$};
        \node[draw, line width=0.4mm, circle, minimum size=1.0cm, inner sep=0pt] (Y) at (0,4) {\normalsize$Y$};
        \node[draw, line width=0.4mm, circle, minimum size=1.0cm, inner sep=0pt] (Z) at (3,4) {\normalsize$Z$};
        \node[draw, line width=0.4mm, circle, minimum size=1.0cm, inner sep=0pt] (G) at (0,6) {\normalsize$\mathbb{G}$};
        \draw[line width=0.7mm, shorten >=-0.0cm, shorten <=-0.0cm] (A) -- (X);
        \draw[line width=0.7mm, shorten >=-0.0cm, shorten <=-0.0cm] (X) -- (Y);
        \draw[line width=0.7mm, shorten >=-0.0cm, shorten <=-0.0cm] (Y) -- (Z);
        \draw[line width=0.7mm, shorten >=-0.0cm, shorten <=-0.0cm] (Z) -- (B);
        \draw[line width=0.7mm, shorten >=-0.0cm, shorten <=-0.0cm] (Y) -- (G);

        % Paths P_1 - P_m
        % Path P_1
        \node[draw, line width=0.4mm, circle, minimum size=0.7cm, inner sep=0pt] (P11) at (-3,1.5) {$P\mathrlap{}^{1}_{1}$};
        \node[draw, line width=0.4mm, circle, minimum size=0.7cm, inner sep=0pt] (P12) at (-1.5,1.5) {$P\mathrlap{}^{2}_{1}$};
        
        \node[draw, line width=0.4mm, circle, minimum size=0.7cm, inner sep=0pt] (P1t) at (1.5,1.5) {$P\mathrlap{}^{T}_{1}$};
        \node[draw, line width=0.4mm, circle, minimum size=0.7cm, inner sep=0pt] (P1t+1) at (3,1.5) {$P\mathrlap{}^{\text{\tiny \smash{$T\mkern-5mu+\mkern-5mu1$}}}_{1}$};
        \draw[line width=0.7mm, shorten >=-0.0cm, shorten <=-0.0cm] (A) -- (P11);
        \draw[line width=0.7mm, shorten >=-0.0cm, shorten <=-0.0cm] (P11) -- (P12);
        \draw[line width=0.7mm, shorten >=-0.0cm, shorten <=-0.0cm] (P12) -- (-0.6,1.5);
        \fill (-0.25,1.5) circle (0.05cm);
        \fill (0,1.5) circle (0.05cm);
        \fill (0.25,1.5) circle (0.05cm);
        \draw[line width=0.7mm, shorten >=-0.0cm, shorten <=-0.0cm] (0.6,1.5) -- (P1t);
        \draw[line width=0.7mm, shorten >=-0.0cm, shorten <=-0.0cm] (P1t) -- (P1t+1);
        \draw[line width=0.7mm, shorten >=-0.0cm, shorten <=-0.0cm] (P1t+1) -- (B);
        \draw[draw=\pathcolor, line width=0.4mm] (-3.7,2.2) rectangle (3.8,0.8);

        \node at (0, 2.6) {\textcolor{\pathcolor}{\small$P_1$}};
        %Dots in between Paths
        \fill (-3,0.4) circle (0.08cm);
        \fill (-3,0) circle (0.08cm);
        \fill (-3,-0.4) circle (0.08cm);
        \draw[line width=0.7mm, shorten >=-0.0cm, shorten <=-0.0cm] (A) -- (-4.3,0.4);
        \draw[line width=0.7mm, shorten >=-0.0cm, shorten <=-0.0cm] (A) -- (-4.3,0.0);
        \draw[line width=0.7mm, shorten >=-0.0cm, shorten <=-0.0cm] (A) -- (-4.3,-0.4);
        
        \fill (3,0.4) circle (0.08cm);
        \fill (3,0) circle (0.08cm);
        \fill (3,-0.4) circle (0.08cm);

        \draw[line width=0.7mm, shorten >=-0.0cm, shorten <=-0.0cm] (B) -- (4.3,0.4);
        \draw[line width=0.7mm, shorten >=-0.0cm, shorten <=-0.0cm] (B) -- (4.3,0.0);
        \draw[line width=0.7mm, shorten >=-0.0cm, shorten <=-0.0cm] (B) -- (4.3,-0.4);

        %Path P_m
        \node[draw, line width=0.4mm, circle, minimum size=0.7cm, inner sep=0pt] (Pm1) at (-3,-1.5) {$P^{1}_{\text{\tiny \smash{$2T\mkern-6mu+\mkern-6mu2$}}}$};
        \node[draw, line width=0.4mm, circle, minimum size=0.7cm, inner sep=0pt] (Pm2) at (-1.5,-1.5) {$P^{2}_{\text{\tiny \smash{$2T\mkern-6mu+\mkern-6mu2$}}}$};
        
        \node[draw, line width=0.4mm, circle, minimum size=0.7cm, inner sep=0pt] (Pmt) at (1.5,-1.5) {$P^{T}_{\text{\tiny \smash{$2T\mkern-6mu+\mkern-6mu2$}}}$};
        \node[draw, line width=0.4mm, circle, minimum size=0.7cm, inner sep=0pt] (Pmt+1) at (3,-1.5) {$P^{\text{\tiny \smash{$T\mkern-5mu+\mkern-5mu1$}}}_{\text{\tiny \smash{$2T\mkern-6mu+\mkern-6mu2$}}}$};
        \draw[line width=0.7mm, shorten >=-0.0cm, shorten <=-0.0cm] (A) -- (Pm1);
        \draw[line width=0.7mm, shorten >=-0.0cm, shorten <=-0.0cm] (Pm1) -- (Pm2);
        \draw[line width=0.7mm, shorten >=-0.0cm, shorten <=-0.0cm] (Pm2) -- (-0.6,-1.5);
        \fill (-0.25,-1.5) circle (0.05cm);
        \fill (0,-1.5) circle (0.05cm);
        \fill (0.25,-1.5) circle (0.05cm);
        \draw[line width=0.7mm, shorten >=-0.0cm, shorten <=-0.0cm] (0.6,-1.5) -- (Pmt);
        \draw[line width=0.7mm, shorten >=-0.0cm, shorten <=-0.0cm] (Pmt) -- (Pmt+1);
        \draw[line width=0.7mm, shorten >=-0.0cm, shorten <=-0.0cm] (Pmt+1) -- (B);
        \draw[draw=\pathcolor, line width=0.4mm] (-3.7,-2.2) rectangle (3.8,-0.8);

        \node at (0, -2.6) {\textcolor{\pathcolor}{\small$P_{\text{\tiny \smash{$2T\mkern-6mu+\mkern-6mu2$}}}$}};

        %Upper part
        \draw [decorate,decoration={brace,amplitude=5pt}] (7,4.8) -- (7,3);
        \node at (8.5,4.2) {Upper part:};
        \node at (8.5,3.6) {interface};

        %Lower part
        \draw [decorate,decoration={brace,amplitude=5pt}] (7,2.55) -- (7,-2.55);
        \node at (8.5,0.3) {Lower part:};
        \node at (8.5,-0.3) {timed fuse};
        
    \end{tikzpicture}		
    \caption[]{Gadget graph $H(G,T)$ with the property that there is a $T$-winning $m$-strategy for $G$ if and only if $\ffn(H(G,T))\leq 4m$.
    The blocks $X,Y,Z$ are $2m$-cliques and every $P_i^j$ and the blocks $A,B$ are $m$-cliques.
    An edge between two blocks $\mathcal{B}_1$ and $\mathcal{B}_2$ in this image corresponds to connecting every node in $\mathcal{B}_1$ to every node in $\mathcal{B}_2$.}
    \label{fig:TimeGadget}
\end{figure}

\noindent Due to the following helpful lemma, we can restrict our analysis to strategies with firefighter sets such that each firefighter set either contains all nodes of a certain clique or none.

\begin{restatable}[Cliques]{lemma}{lemCliques}
\label{lem:cliques}
    Let $G$ be an arbitrary graph containing a clique $K$ with $N(v) \setminus K = N(w) \setminus K$ for all $v,w\in K$, and let $S$ be a $T$-winning $m$-strategy.
    Then there exists a $T$-winning $m$-strategy $S' = (F'_1, \dots, F'_T)$ such that for all $i\in [T]$ and $v\in K$, we have $v \in F'_i$ if and only if $ K\subseteq F'_i$.
\end{restatable}

\begin{restatable}[Time Gadget]{proposition}{propTimeGadget}
\label{prop:timegadget}
    Let $G$ be a graph. There is a $T$-winning $m$-strategy for $G$ if and only if $\ffn(H(G,T))\leq 4m$. 
\end{restatable}
\begin{proof}[Proof sketch]
    If there is a $T$-winning $m$-strategy for $G$, we can explicitly state a winning $4m$-strategy for $H(G, T)$ (see \ref{app:npHardness}), which implies $\ffn(H(G,T))\leq 4m$.

    Next, we assume that there is no $T$-winning $m$-strategy for $G$.
    Using Lemma \ref{lem:cliques}, it is easy to see that this is equivalent to there not being a $T$-winning $2m$-strategy for $\mathbb{G}$.
    Before giving the technical details of the actual rigorous proof, let us first give some intuition on why this means that no $4m$-winning strategy for $H(G, T)$ exists.
    Every node in $X$, $Y$, $Z$, $A$ or $B$ has degree at least $4m$, so the first nodes to be extinguished have to be either in $\mathbb{G}$ or in $\mathcal{P}$.

    If we start by extinguishing nodes in $\mathbb{G}$, we either do not extinguish all nodes in $Y$ afterwards, which lets all of the nodes in $\mathbb{G}$ reignite, or we keep extinguishing all nodes in $Y$.
    In the second case, we use $2m$ firefighters on $Y$, leaving only $2m$ firefighters for the rest of the graph.
    As each node in the graph outside of $\mathbb{G}$ has degree greater than $2m$, it is not possible to make any further progress.

    Let us instead start by extinguishing nodes in $\mathcal{P}$.
    We can initially reach the point where $\mathcal{P}$ and $A$ are completely extinguished.
    Trying to extinguish $B$ next only mirrors the position, as we would lose $A$ in turn.
    Extinguishing $X$ requires all $4m$ firefighters (positioned in $X$ and $Y$) and hence the last block of every path $P_i$ catches on fire again.
    We can now stop the fire from spreading back to $X$ and start working on extinguishing $\mathbb{G}$ by placing $2m$ firefighters in $Y$.
    If we try to extinguish $\mathbb{G}$ with the remaining $2m$ firefighters, it will take us at least $T + 1$ steps.
    By Lemma \ref{lem:cliques}, we can assume that we either use no firefighters or at least $m$ firefighters in the path $P_i$ for any $i \in [2T + 2]$.
    Hence, with $2m$ firefighters in $T$ steps, we can influence at most $2T$ of the $2T + 2$ paths, meaning that at least $2$ of the paths are fully burning again after $T$ steps.
    Therefore, when we finish extinguishing $\mathbb{G}$, the block $A$ will already be reignited.
    If we do not want the entire block $\mathbb{G}$ to reignite, we have to extinguish $Y$ again.
    The only possible progress with the remaining $2m$ firefighters at this stage is to position them at $Z$, so that the nodes in $Y$ actually stay extinguished for the first time.
    However, then only $Y$ and $\mathbb{G}$ are extinguished, and one can see that it is not possible to extinguish any additional block without reverting to a (possibly mirrored) previous state of the game.
    As we have exhausted all reasonable solution approaches, there can be no $4m$-winning strategy for $H(G, T)$.

    Let us now give a brief overview of the technicalities of our proof.
    In order to avoid nested case distinctions, we instead analyze the subsets between which the set of burning nodes in the graph can transition under a $4m$-strategy.
    In particular, we consider the following node subsets:
    $\Omega_1 = H(G,T) \setminus ( A \cup \mathcal{P} )$, \;
    $\Omega_2 = H(G,T) \setminus (\mathbb{G} \cup \mathcal{P})$, \;
    $\Omega_3 = \mathbb{G} \cup B \cup Y \cup Z \cup \bigcup_{i \in [2T + 2]} P_i^{T+1}$, \;
    $\Omega_4 = B \cup Y \cup Z \cup P_k \cup P_\ell \cup \{v\}$, \;
    $\Omega_5 = A \cup B \cup Y \cup Z \cup P_k \cup P_\ell$, \; $\Omega_6 = A \cup B \cup X \cup Z \cup P_k \cup P_\ell$, and $\Omega_7 = A \cup B \cup Y \cup Z \cup P_k \cup (P_\ell \setminus P_\ell^{1}) \cup \{v\}$ where $v$ is any node from $ \mathbb{G}$ and $k,\ell\in [2T+2]$ with $k\neq \ell$.
    We call a subset of burning nodes \emph{$\Omega_\ell$-blocked}, if it contains $\Omega_\ell$ or one of its symmetric variants, regarding the following symmetries:
    Switching $A$ and $B$, $X$ and $Z$ as well as $P_i^j$ with $P_i^{T + 2 - j}$ for all $i \in [2T + 2], j \in [T + 1]$ (i.e., mirroring the graph as shown in Figure \ref{fig:TimeGadget} horizontally), switching the complete paths $\{P_1, \ldots, P_{2T + 2}\}$ according to any permutation, or replacing $v$ by any other node in $\mathbb{G}$.

    In \ref{app:npHardness} we prove that for any $4m$-strategy and any subset of burning nodes that is $\Omega_n$-blocked for some $n \in [7]$, after finitely many steps, the subset of burning nodes will be $\Omega_{n'}$-blocked for some $n' \in [7]$.
    Since the initial state of a fully burning graph is $\Omega_1$-blocked, this means that there is no winning $4m$-strategy for $H(G,T)$, as the empty set is not $\Omega_{n}$-blocked for any $n \in [7]$.
\end{proof}

\begin{restatable}[Hardness of \mff]{theorem}{mffHard}
\label{thm:mffHard}
    $\mff$ is \texttt{NP-hard}.
\end{restatable}
\begin{proof}
Let an instance of the following decision problem be given.
\begin{framed}
\begin{tabular}{l}
\textbf{(\textsc{BinPacking}):} \\
\hspace{0.5cm}\textbf{Input:} Items $i_1, \ldots, i_n$, sizes $s_1, \ldots, s_n \in \N_{>0}$ and $b \in \N_{>0}$ bins with capacity $c \in \N_{>0}$\\ 
\hspace{0.5cm} each, such that $\sum_{k \in [n]} s_k \leq b \cdot c$.\\
\hspace{0.5cm}\textbf{Output:} Is there a packing of all items into the bins fulfilling the capacity constraints?
\end{tabular}
\end{framed}
    \noindent We construct a graph $G$ as $\bigcup_{k \in [n]} K_{s_k}$, i.e., the graph $G$ consists of $n$ connected components, which are complete graphs of varying sizes.
    This graph $G$ can be extinguished by $c$ firefighters in at most $b$ steps if and only if the \textsc{BinPacking} instance is a yes-instance, since we can assume that only full cliques are extinguished in each step due to Lemma \ref{lem:cliques}. By using the construction from Proposition \ref{prop:timegadget} with this $G$ and $T = b$, we reduce this instance of \textsc{BinPacking} to determining if $\ffn(H(G,T)) \leq 4c$. Note that the size of this instance of $\mff$ is in $\mathcal{O}(c^4 \cdot b^2)$ and therefore polynomial in the unary encoding size of the \textsc{BinPacking} problem.
    Since \textsc{BinPacking} is \texttt{strongly NP-complete}~\cite{10.5555/574848}, it follows that in general, determining whether a given graph can be extinguished by $m$ firefighters is \texttt{NP-hard}.
\end{proof}

\noindent Note that this Theorem implies that $\mff$ is even \texttt{strongly NP-hard}, since for any non-trivial instance, $m$ is bounded by $|V|$, so the input size of the problem remains in $\mathcal{O}(\text{size}(G))$ even if $m$ is encoded in unary.

While the proof strategy for Theorem \ref{thm:mffHard} is based on first showing that an \texttt{NP-hard} problem can be reduced to solving an instance of $\mfft$ and this specific instance is afterwards shown to be equivalent to an instance of $\mff$, we cannot generally reduce $\mfft$ to $\mff$ with our construction:
Since the number of vertices of $H(G, T)$ is greater than $T^2$, the size of the graph $H(G, T)$ is in general not polynomial in the encoding size of $G$, $T$ and $m$.
However, a straightforward polynomial reduction for the converse direction exists:

\begin{proposition}[Reducing $\mff$ to \mfft]\label{prop:Reduction}
    \mff ~can be polynomially reduced to \mfft.
\end{proposition}

\begin{proof}
    Let $G$ be the given graph and $m$ the given number of firefighters. $G$ can be solved by $m$ firefighters if and only if it can be solved by $m$ firefighters in at most $2^{|V(G)|}$ steps, since any shortest winning $m$-strategy will not reach any burning set more than once. As the encoding length of $2^{|V(G)|}$ is in $\mathcal{O}(|V(G)|)$, this reduction is polynomial in the encoding size of the initial question.
\end{proof}

\noindent In particular, this immediately implies that $\mfft$ is \texttt{NP-hard} as well.
However, we can go even further:

\begin{restatable}[Hardness of \mfft]{theorem}{thmMffinTIMEHard}
\label{thm:mffinTIMEHard}
    The problem $\mfft$ is \texttt{NP-hard} even on trees.
    Moreover, it is \texttt{NP-hard} even on trees with diameter at most $4$ and on spiders (trees where at most one node has a degree greater than $2$). 
\end{restatable}
\begin{proof}[Proof sketch]
    We prove this via a reduction from \textsc{3-Partition}, which is \texttt{strongly NP-hard}, see \cite{10.5555/574848}.
    For some $k \in \N_{>0}$, let $a_1, \dots, a_{3k} \in \N_{>0}$ be the positive integer numbers in a given instance of \textsc{3-Partition}, and set $s = \sum_{i = 1}^{3k} a_i$.
    Without loss of generality, we may assume $\tfrac{s}{k} \in \N_{>0}$. Otherwise, a 3-partition of the numbers trivially cannot exist.
    
    We now construct a graph $G$ which we claim is $(\tfrac{s}{k} + 3s + 1)$-winning in time $k$ if and only if there exists a 3-partition of $a_1, \dots, a_{3k}$ (this claim is proven in \ref{app:npHardness}).
    Let $T_i$ be an arbitrary tree with $a_i + s$ nodes for each $i \in [3k]$.
    Then the graph $G$ arises by adding a new node $c$ and, for each $i \in [3k]$, adding an edge between $c$ and an arbitrary node from $T_i$ as visualized in Figure \ref{fig:TreeConstruction2}.
    Note that $|G| = s + 3sk + 1 = k \cdot (\tfrac{s}{k} + 3s + 1) - (k - 1)$, which implies that at most $k - 1$ nodes can be reignited during a successful strategy in time $k$ with $\tfrac{s}{k} + 3s + 1$ firefighters.
    By choosing a star graph (resp. a path graph) for each $T_i$ and attaching $c$ to the internal node (resp. to an end of the path), we get the result for trees with diameter $\leq 4$ (resp. for spiders).
\end{proof}
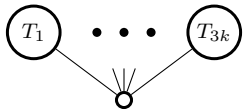
\begin{figure}[H]
    \centering
\begin{tikzpicture}[scale=0.6, every node/.style={font=\scriptsize}]

        %First node, X and some edges
        \node[draw, line width=0.4mm, circle, minimum size=0.2cm, inner sep=0pt, fill=white] (v1) at (0,-0.5) {};

        \draw[] (v1) -- (-0.25,0.2);
        \draw[] (v1) -- (0,0.2);
        \draw[] (v1) -- (0.25,0.2);

        %Complete Graphs
        \node[draw, line width=0.4mm, circle, minimum size=0.7cm, inner sep=0pt] (K1) at (-2,1) {$T_1$};
        \node[draw, line width=0.4mm, circle, minimum size=0.7cm, inner sep=0pt] (K4) at (2,1) {$T_{3k}$};
        \fill (-0.6,1) circle (0.1cm);
        \fill (0.0,1) circle (0.1cm);
        \fill (0.6,1) circle (0.1cm);
        \draw[] (v1) -- (K1);
        \draw[] (v1) -- (K4);

\end{tikzpicture}
    \caption{Construction of $G$. Every $T_i$ is an arbitrary tree with $a_i + s$ nodes.} 
    \label{fig:TreeConstruction2}
\end{figure}

\section{Graphs with Long Shortest Strategies}\label{sec:shortStrat}

After proving that the problems $\mff$ and $\mfft$ are \texttt{NP-hard}, a naturally arising question is whether these problems are in \texttt{NP}, i.e., whether there exists a polynomial certificate for yes-instances.
A natural candidate for such a certificate would be a winning $m$-strategy (resp. winning $m$-strategy in time $T$), since a strategy can be verified in polytime with respect to its size and the size of the graph.
In this section, we will show that such a straightforward approach to attain a polynomial certificate does not work, by giving a class of graphs where the shortest possible winning $m$-strategy takes superpolynomially (in the size of the graph) many steps.

To this end, we first define the auxiliary graph $H_m$ (see Figure \ref{fig:LongStrategies} (a)), which has the useful property that any winning $m$-strategy on $H_m$ has to use at least $m - 1$ firefighters for a certain number of consecutive steps.
Here, $\alpha$ and $\beta$ (which will be used in the upcoming definitions) denote some fixed values in $\N_{>0}$ that fulfill $2\beta + 2 \geq \alpha \geq \beta + 3$, e.g., $\beta = 1, \alpha = 4$.
In particular, they do not depend on $m$ or $X$.

\begin{definition}[Auxiliary Graph $H_m$]\label{def:aux}
    For $m \geq 2$, we set $H_m = (V, E)$ with $V = \{v_1, \dots, v_{m - 1}\} \cup \{w_1, \dots, w_\alpha\}$ and $E = \{\{v_i, v_j\} : i, j \in [m - 1], i \neq j\} \cup \{\{v_i, w_j\}: i \in [m - 1], j \in [\alpha]\}$.
\end{definition}

\noindent
Next, we construct a graph $G(m, X)$ (see Figure \ref{fig:LongStrategies} (b)) with $\ffn(G(m, X)) = m$ that contains a subgraph $X$ with $\ffn(X) = m - 1$.
We will prove that any $m$-winning strategy for $G(m, X)$ needs to fully extinguish the graph $X$ with $m - 1$ firefighters at least $m - 1$ times.

\begin{definition}[$G(m, X)$]\label{def:GkX}
    For any $m \in \N_{\geq 2}$ and graph $X$ with $\ffn(X) = m - 1$, the graph $G(m, X)$ arises in the following way:
    First, add an additional node $c$ to $X$, which shares an edge with every node in $X$.
    Next, add $m$ paths with $\beta$ nodes $(v^i_1, \dots, v^i_\beta)$ for $i \in [m]$, and for each $i \in [m]$, connect $c$ to $v^i_1$.
    Finally, add $m$ auxiliary graphs $H_m$ as defined previously, say $H_m^1, \dots, H_m^m$, and for each $i \in [m]$, connect $v^i_\beta$ to one arbitrary node $u_i$ of the $(m - 1)$-clique contained in $H^i_m$.
\end{definition}

\begin{figure}[H]
    \centering
    \subfloat[$H_m$]{
\begin{tikzpicture}[scale=0.6, every node/.style={font=\scriptsize}]
        %Complete Graphs
        \node[draw, line width=0.4mm, circle, minimum size=1.0cm, inner sep=0pt] (K1) at (-3,3) {$K_{m-1}$};
    
        %nodes on top of complete Graphs
        \node[draw, circle, minimum size=0.2cm, inner sep=0pt, fill=white] (v2) at (-3.8,4.7) {};
        \node[draw, circle, minimum size=0.2cm, inner sep=0pt, fill=white] (v3) at (-2.2,4.7) {};
        \draw[line width=1mm, shorten >=-0.0cm, shorten <=-0.0cm] (K1) -- (v2);
        \draw[line width=1mm, shorten >=-0.0cm, shorten <=-0.0cm] (K1) -- (v3);
        
        \draw [decorate,decoration={brace,amplitude=5pt}] (-3.9,5) -- (-2.1,5);
        \fill (-2.7,4.7) circle (0.05cm);
        \fill (-3.0,4.7) circle (0.05cm);
        \fill (-3.3,4.7) circle (0.05cm);
        \node at (-3, 5.5) {$\alpha$-times};
        \node at (-2.2, 0.75) {\;};	
    \end{tikzpicture}
}
\hspace{3cm}
\subfloat[$G(m,X)$]{
\begin{tikzpicture}[scale=0.5, every node/.style={font=\scriptsize}]
        \tikzset{decoration={snake,amplitude=.4mm,segment length=2mm,
                       post length=0mm,pre length=0mm}}
        %First node, X and some edges
        \node[draw, circle, minimum size=0.2cm, inner sep=0pt, fill=white] (v1) at (-0.5,0) {};
        \node[draw, line width=0.4mm, circle, minimum size=1.0cm, inner sep=0pt] (G1) at (-2.5,0) {$X$};
        \draw[line width=1mm, shorten >=-0.0cm, shorten <=-0.0cm] (v1) -- (G1);

        \draw[] (v1) -- (0.2,-0.5);
        \draw[] (v1) -- (0.2,0);
        \draw[] (v1) -- (0.2,0.5);
 
        %Complete Graphs
        \node[draw, line width=0.4mm, circle, minimum size=1cm, inner sep=0pt] (K1) at (6.5,-2) {$H_m$};
        \node[draw, line width=0.4mm, circle, minimum size=1cm, inner sep=0pt] (K4) at (6.5,2) {$H_m$};
        \fill (6.5,-0.6) circle (0.1cm);
        \fill (6.5,0.0) circle (0.1cm);
        \fill (6.5,0.6) circle (0.1cm);

        %Paths to Complete Graphs
        \node[draw, circle, minimum size=0.2cm, inner sep=0pt] (p11) at (0.2,-2) {};
        \node[draw, circle, minimum size=0.2cm, inner sep=0pt] (unten1) at (1.2,-2) {};
        \node[draw, circle, minimum size=0.2cm, inner sep=0pt] (unten2) at (3.8,-2) {};
        \node[draw, circle, minimum size=0.2cm, inner sep=0pt] (p12) at (4.8,-2) {};
        % \draw[decorate] (p11) -- (p12);

        \draw[] (v1) -- (p11);
        \draw[] (p11) -- (unten1);
        \draw[] (unten1) -- (1.7,-2);
        \fill (2.2,-2) circle (0.05cm);
        \fill (2.5,-2) circle (0.05cm);
        \fill (2.8,-2) circle (0.05cm);
        \draw[] (3.3,-2) -- (unten2);
        \draw[] (unten2) -- (p12);

        \draw [decorate,decoration={brace,amplitude=5pt}] (4.9,-2.5) -- (0.1,-2.5);
        \node at (2.5,-3.2) {path of length $\beta$};

        \node[draw, circle, minimum size=0.2cm, inner sep=0pt] (p41) at (0.2,2) {};
        \node[draw, circle, minimum size=0.2cm, inner sep=0pt] (oben1) at (1.2,2) {};
        \node[draw, circle, minimum size=0.2cm, inner sep=0pt] (oben2) at (3.8,2) {};
        \node[draw,circle, minimum size=0.2cm, inner sep=0pt] (p42) at (4.8,2) {};
        % \draw[decorate] (p41) -- (p42);

        \draw[] (v1) -- (p41);
        \draw[] (p41) -- (oben1);
        \draw[] (oben1) -- (1.7,2);
        \fill (2.2,2) circle (0.05cm);
        \fill (2.5,2) circle (0.05cm);
        \fill (2.8,2) circle (0.05cm);
        \draw[] (3.3,2) -- (oben2);
        \draw[] (oben2) -- (p42);

        \draw[] (p12) -- (K1);
        \draw[] (p42) -- (K4);

        %Bracket on top of graph and label $G_k$
        \draw [decorate,decoration={brace,amplitude=5pt}] (7.6,3) -- (7.6,-3);
        \node at (9.2,0) {$m$-times};

        \draw [decorate,decoration={brace,amplitude=5pt}] (0.1,2.5) -- (4.9,2.5);
        \node at (2.5,3.2) {path of length $\beta$};
\end{tikzpicture}
}
    \caption{When a thick edge connects two subgraphs A and B, then every node in A is connected to every node in B.
    The rightmost node of each path of length $\beta$ is connected to exactly one (arbitrary) node of $K_{m - 1}$ of the corresponding $H_m$.} 
    \label{fig:LongStrategies}
\end{figure}
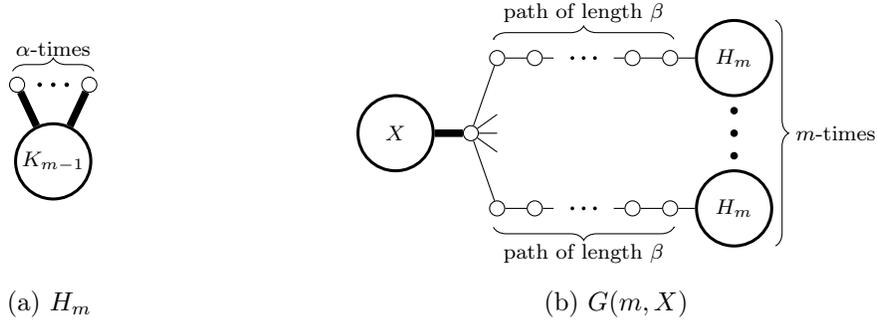

\begin{restatable}[Firefighter Number of $G(m, X)$]{lemma}{lemFfnofGkX}
\label{lem:ffnofGkX}
    $\ffn(G(m, X)) = m$.
\end{restatable}

\begin{proof}[Proof sketch]
    By construction, the graph $G(m, X)$ contains the subgraph $H_m$, which in turn contains a $m$-clique.
    This shows $\ffn(G(m, X)) \geq m$.

    To give some intuition for why we have $\ffn(G(m, X)) \leq m$, let us now sketch an outline of a winning $m$-strategy.
    Start by extinguishing $H^1_m \setminus \{u_1\}$, which takes $\alpha$ steps.
    Next, extinguish the path from $u_1$ to $c$ and position one firefighter in $c$ while extinguishing $X$ with the remaining $m - 1$ firefighters.
    Afterwards, it is possible to extinguish all other paths up to the $u_i$ nodes without letting the fire spread back, so that only $H^2_m, \dots, H^m_m$ still contain burning nodes.
    Then, continue by extinguishing the next auxiliary subgraph $H^2_m$.
    During this process, the fire barely does not reach $H^1_m$, which allows us to again extinguish the entire rest of the graph apart from $H^3_m, \dots, H^m_m$.
    By repeating this process $m - 2$ more times, the entire graph is extinguished.
\end{proof}

\noindent After determining that $\ffn(G(m, X)) = m$, we shall now find a lower bound to the length of a winning $m$-strategy for $G(m, X)$ by showing that such a strategy needs to be similar to the strategy described in the proof sketch of the previous lemma, and therefore needs to repeatedly extinguish $X$ with $m - 1$ firefighters.

\begin{restatable}[Lower Bound on $T(G(m, X))$]{lemma}{lemGkXtakeslong}
\label{lem:GkXtakeslong}
    $T(G(m, X)) \geq (m - 1) \cdot T(X)$.
\end{restatable}

\noindent By recursive applications of the construction $G(m, X)$, we can give a class of graphs such that the length of a shortest extinguishing strategy with the smallest possible number of firefighters is superpolynomial in the size of the graph.

\begin{restatable}[Shortest Strategies can have Superpolynomial Length]{theorem}{thmLength}
\label{thm:length}
	Let $G_2 = G(2, (\{v\}, \emptyset))$ and $G_m = G(m, G_{m - 1})$ for any $m \in \N_{\geq 3}$.
    For any $m \in \N_{\geq 2}$, we have $\ffn(G_m) = m$ and $T(G_m) = T_m(G_m) \geq (m - 1)!$, which is superpolynomial in $\text{size}(G_m) = \mathcal{O}(m^6)$.
\end{restatable}
\begin{proof}
    By Lemmata \ref{lem:ffnofGkX} and \ref{lem:GkXtakeslong}, we have $\ffn(G_2) = 2$ and $T(G_2) \geq 1$.
    Using induction and the same two lemmata, we get $\ffn(G_m) = m$ and $T(G_m) \geq (m - 1) \cdot T(G_{m - 1}) \geq (m - 1)!$ for any $m \geq 2$.
    Moreover, $G(m, X) \setminus X$ contains $1 + m \cdot (m - 1 + \alpha + \beta)$ nodes, which is in $\mathcal{O}(m^2)$.
    Since $G_m$ is composed of $m - 1$ graphs that do not have a greater number of nodes than $G(m, X) \setminus X$, it follows that $G_m$ has $\mathcal{O}(m^3)$ nodes and therefore has size in $\mathcal{O}(m^6)$.
    As $(m - 1)!$ is superpolynomial in $m^6$, this finishes the proof.
\end{proof}

\noindent While the above theorem underlines the possibility that $\mfft$ as well as $\mff$ are not in \texttt{NP}, we can at least give an upper bound to their space complexity.

\begin{proposition}[$\mff$ is in \texttt{PSPACE}]
    The problem $\mfft$ (and therefore also $\mff$) is in \texttt{PSPACE}.
\end{proposition}
\begin{proof}
    Given an instance $(G = (V, E), m, T)$ of $\mfft$, consider the following algorithm:
    Set $B = V$.
    Repeat the following steps $T$ times:
    Pick a random $F \subseteq V$ with $|F| \leq m$.
    Set $\tilde{B} = B \setminus F$.
    Set $B = \tilde{B} \cup N(\tilde{B})$.
    If $B = \emptyset$ after $T$ repetitions, return ``yes'', otherwise return ``no''.
    This non-deterministic algorithm has a probability strictly greater than $0$ to return ``yes'' if $(G, m, T)$ is a yes-instance of $\mfft$, and will always return ``no'' otherwise.
    Furthermore, the required space is in $\mathcal{O}(\text{size}(G))$, so $\mfft$ is in \texttt{NPSPACE} and, by the Theorem of Savitch~\cite{SAVITCH1970177}, in \texttt{PSPACE}.
    Proposition \ref{prop:Reduction} extends this result to $\mff$.
\end{proof}

\section{Implications on the Rabbit and Hunter Game}\label{app:hunter}
In this section, we show that the hardness and the existence of long shortest strategies transfer to the Rabbit and Hunter game.
The only difference between $\mff$ and the Rabbit and Hunter game is that the fugitive is not forced to move in each turn, i.e., we need to redefine $B_t\coloneqq N(B_{t-1} \setminus F_t)$ to use our notation.
To distinguish between settings, we will use $\hn(G)$ instead of $\ffn(G)$, $m$-hunter strategy instead of $m$-(firefighter-)strategy and $T^h(G)$ for the shortest $m$-hunter strategy instead of $T(G)$ for the shortest $m$-firefighter strategy.

\begin{framed}
\begin{tabular}{l}
\textbf{(\mhun): Rabbit and Hunter game} \\
\hspace{0.5cm}\textbf{Input:} A graph $G$ and $m\in \N_{>0}$.\\
\hspace{0.5cm}\textbf{Output:} Is $\hn(G)\leq m$?
\end{tabular}
\end{framed}

\noindent Let $G'=(V',E')$ be the graph that arises by replacing each edge $e\in E$ from a graph $G=(V,E)$ with $|V|+1$ many disjoint paths, each with one intermediate node as visualized in Figure \ref{fig:hunter}. Note that $G'$ is bipartite with partition sets $V$ and $V'\setminus V$. 

\begin{figure}[h]
    \centering
    \subfloat[$G=(V,E)$]{
\begin{tikzpicture}[scale=0.3, every node/.style={font=\scriptsize}]
            \node[draw, circle] (v1) at (0,0) {};
			\node[draw, circle] (v2) at (5,0) {};
			\node[draw, circle] (v3) at (10,5) {};
			\node[draw, circle] (v4) at (10,-5) {};
			% Connect nodes with styled edges
			\draw[thick] (v1) -- (v2);
			\draw[thick] (v2) -- (v3);
			
			\draw[thick] (v2) -- (v4);
			\draw[thick] (v3) -- (v4);
    \end{tikzpicture}
}
\hspace{1cm}
\subfloat[$G'=(V',E')$]{
\begin{tikzpicture}[scale=0.3, every node/.style={font=\scriptsize}]
            \node[draw, circle] (w1) at (15,0) {};
			\node[draw, circle] (w2) at (20,0) {};
			\node[draw, circle] (w3) at (28,5) {};
			\node[draw, circle] (w4) at (28,-5) {};

            \node[draw, circle, minimum size=0.2cm, inner sep=0pt, fill=blue!40] (w121) at (17.5,2) {};
            \node[draw, circle, minimum size=0.2cm, inner sep=0pt, fill=blue!40] (w122) at (17.5,1) {};
            \node[draw, circle, minimum size=0.2cm, inner sep=0pt, fill=blue!40] (w123) at (17.5,0) {};
            \node[draw, circle, minimum size=0.2cm, inner sep=0pt, fill=blue!40] (w124) at (17.5,-1) {};
            \node[draw, circle, minimum size=0.2cm, inner sep=0pt, fill=blue!40] (w125) at (17.5,-2) {};

            \draw[thick] (w1) -- (w121);
            \draw[thick] (w1) -- (w122);
            \draw[thick] (w1) -- (w123);
            \draw[thick] (w1) -- (w124);
            \draw[thick] (w1) -- (w125);

            \draw[thick] (w2) -- (w121);
            \draw[thick] (w2) -- (w122);
            \draw[thick] (w2) -- (w123);
            \draw[thick] (w2) -- (w124);
            \draw[thick] (w2) -- (w125);

            \node[draw, circle, minimum size=0.2cm, inner sep=0pt, fill=blue!40] (w231) at (23.5,4.5) {};
            \node[draw, circle, minimum size=0.2cm, inner sep=0pt, fill=blue!40] (w232) at (24,3.5) {};
            \node[draw, circle, minimum size=0.2cm, inner sep=0pt, fill=blue!40] (w233) at (24.5,2.5) {};
            \node[draw, circle, minimum size=0.2cm, inner sep=0pt, fill=blue!40] (w234) at (25,1.5) {};
            \node[draw, circle, minimum size=0.2cm, inner sep=0pt, fill=blue!40] (w235) at (23,5.5) {};

            \draw[thick] (w2) -- (w231);
            \draw[thick] (w2) -- (w232);
            \draw[thick] (w2) -- (w233);
            \draw[thick] (w2) -- (w234);
            \draw[thick] (w2) -- (w235);

            \draw[thick] (w3) -- (w231);
            \draw[thick] (w3) -- (w232);
            \draw[thick] (w3) -- (w233);
            \draw[thick] (w3) -- (w234);
            \draw[thick] (w3) -- (w235);

            \node[draw, circle, minimum size=0.2cm, inner sep=0pt, fill=blue!40] (w241) at (23.5,-4.5) {};
            \node[draw, circle, minimum size=0.2cm, inner sep=0pt, fill=blue!40] (w242) at (24,-3.5) {};
            \node[draw, circle, minimum size=0.2cm, inner sep=0pt, fill=blue!40] (w243) at (24.5,-2.5) {};
            \node[draw, circle, minimum size=0.2cm, inner sep=0pt, fill=blue!40] (w244) at (25,-1.5) {};
            \node[draw, circle, minimum size=0.2cm, inner sep=0pt, fill=blue!40] (w245) at (23,-5.5) {};

            \draw[thick] (w4) -- (w241);
            \draw[thick] (w4) -- (w242);
            \draw[thick] (w4) -- (w243);
            \draw[thick] (w4) -- (w244);
            \draw[thick] (w4) -- (w245);

            \draw[thick] (w2) -- (w241);
            \draw[thick] (w2) -- (w242);
            \draw[thick] (w2) -- (w243);
            \draw[thick] (w2) -- (w244);
            \draw[thick] (w2) -- (w245);

            \node[draw, circle, minimum size=0.2cm, inner sep=0pt, fill=blue!40] (w341) at (26,0) {};
            \node[draw, circle, minimum size=0.2cm, inner sep=0pt, fill=blue!40] (w342) at (27,0) {};
            \node[draw, circle, minimum size=0.2cm, inner sep=0pt, fill=blue!40] (w343) at (28,0) {};
            \node[draw, circle, minimum size=0.2cm, inner sep=0pt, fill=blue!40] (w344) at (29,0) {};
            \node[draw, circle, minimum size=0.2cm, inner sep=0pt, fill=blue!40] (w345) at (30,0) {};

            \draw[thick] (w4) -- (w341);
            \draw[thick] (w4) -- (w342);
            \draw[thick] (w4) -- (w343);
            \draw[thick] (w4) -- (w344);
            \draw[thick] (w4) -- (w345);

            \draw[thick] (w3) -- (w341);
            \draw[thick] (w3) -- (w342);
            \draw[thick] (w3) -- (w343);
            \draw[thick] (w3) -- (w344);
            \draw[thick] (w3) -- (w345);
\end{tikzpicture}
}
    \caption{$G'=(V',E')$ arises from $G=(V, E)$ by replacing each edge $e\in E$ with $|V|+1$ many disjoint paths, each with one intermediate node. Note that $G'$ is bipartite.} 
    \label{fig:hunter}
\end{figure}
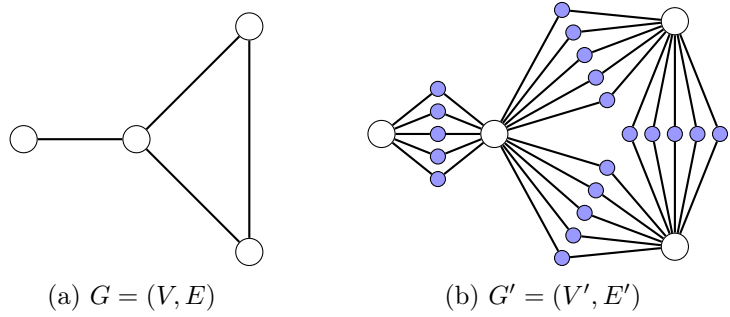

\begin{lemma}
   $\ffn(G)=\hn(G')$ and $2 T(G) = T^{h}(G')$.
\end{lemma}

\begin{proof}
    Since $G'$ is bipartite, the rabbit's position alternates between nodes in $V$ and $V'\setminus V$.
    Let $F = (F_1,F_2,\ldots, F_n)$ be a winning $\ffn(G)$-firefighter strategy on $G$.
    If the rabbit starts in a node $v \in V$, it will be in $\{v\} \cup \{w\in V : \; \exists \{v,w\} \in E\}$ after two time steps.
    Thus, the strategy $(F_1, \emptyset, F_2,\emptyset, \ldots, F_{n-1},\emptyset, F_n)$ catches the rabbit if it starts in $V$, and the strategy $(\emptyset, F_1, \emptyset, F_2,\emptyset, \ldots,$ $F_{n-1},\emptyset, F_n)$ catches the rabbit if it starts in $V'\setminus V$.
    Hence, $F'=(F_1,F_1, F_2, F_2,\ldots, F_n,F_n)$ is a winning $\ffn(G)$-hunter strategy.
    This implies $\hn(G')\leq \ffn(G)$ and $2T(G) \geq T^h(G')$.

    Suppose $\hn(G')< \ffn(G)$.
    Let $F'=(F_1',\ldots, F_{n'}')$ be a winning $\hn(G')$-hunter strategy on $G'$.
    If $u \in \tilde{B}_t$ or $v \in \tilde{B}_t$ for some $\{u, v\} \in E$ and $t \in \N_{> 0}$, then all intermediate nodes between $u$ and $v$ in $G'$ are in $B_t$.
    Since there are $|V|+1 > |V|\geq \ffn(G)>\hn(G')$ many of these nodes, at least one of them is in $\tilde{B}_{t+1}$.
    Therefore, we have $u,v\in B_{t+1}$, independent of the hunter strategy.
    This yields that the pruned hunter strategy $F''= (F_1'\cap V,\ldots, F_{n'}'\cap V)$ is also a winning $\hn(G')$-hunter strategy on $G'$. 
    The even and odd substrategies $(F''_{2i})_{i \in [\lfloor n'/2 \rfloor]}$ and $(F''_{2i - 1})_{i \in [\lceil n'/2 \rceil]}$ are both winning $\hn(G')$-firefighter strategies on $G$, since $B_{2t}\cap V$ in $G'$ under hunter strategy $F''$ coincides with $B_t$ in $G$ under firefighter strategy $(F''_{2i})_{i \in [\lfloor n'/2 \rfloor]}$, and $B_{2t-1}\cap V$ in $G'$ under hunter strategy $F''$ coincides with $B_{t}$ in $G$ under firefighter strategy $(F''_{2i - 1})_{i \in [\lceil n'/2 \rceil]}$.
    This implies that both $(F''_{2i})_{i \in [\lfloor n'/2 \rfloor]}$ and $(F''_{2i - 1})_{i \in [\lceil n'/2 \rceil]}$ are $hn(G')$-winning firefighter strategies for $G$. 
    Hence, we have $\ffn(G) \leq \hn(G')$, which is a contradiction.
    Additionally, if $F'$ was already a shortest winning $\hn(G')$-hunter strategy, this implies that $2T(G)\leq T^h(G')$.
    We conclude $\ffn(G)=\hn(G')$ and $2 T(G) = T^{h}(G')$.
\end{proof}

\noindent Observe that $\text{size}(G')\in \mathcal{O}(\text{size}(G)^2)$.
Together with Theorem \ref{thm:mffHard} and Theorem \ref{thm:length}, this implies the following.

\begin{corollary}[Hardness of $\mhun$ and Long Shortest Strategies]\label{prop:longHunt}
    $\mhun$ is \texttt{NP-hard} even on bipartite graphs and there exists an infinite family of graphs for which shortest hunter-strategies have superpolynomial length in their respective sizes.
\end{corollary}

\noindent Corollary \ref{prop:longHunt} answers one of the open questions in \cite{benameur2025huntingrabbithard} regarding a bound on the length of shortest hunter strategies.
Note that the \texttt{NP-hardness} of $\mhun$ on bipartite graphs has also been shown in \cite{benameur2025huntingrabbithard}, using a completely different approach.

\section{Open Problems}
In Corollary \ref{prop:char}, we gave a characterization of graphs with firefighter number one or two.
A naturally arising question is whether it is possible to find a somewhat compact characterization for graphs with firefighter number three, or other bigger fixed numbers.

Table \ref{tab:basicProps} covers some relations between the firefighter number and commonly used graph parameters. It is an interesting question whether one can relate the firefighter number to other graph parameters.
In particular, we think that it would be worthwhile to investigate the existence of a lower bound based on the treewidth.

While we have proven that \mff~ is \texttt{NP-hard} and in \texttt{PSPACE}, one could investigate the placement of the problem within the polynomial hierarchy in more detail.

If applying Lemma \ref{lem:lowerbounds}.1 to all subgraphs suffices to characterize the firefighter number for some class of graphs, we can show that $\mff$ restricted to this class is in \texttt{co-NP}, if either the number of firefighters $m$ or the treewidth of the graphs is bounded by a constant.
This would imply that these restricted variants are not \texttt{NP-hard}, unless \texttt{NP} equals \texttt{co-NP}.
However, it is an open problem which graph classes satisfy this property, see Open Problem 70 in \cite{Bernshteyn2021SearchingFA}.

Finally, we are interested in algorithms to solve \mff~ or \mfft~ for medium-sized instances.
The canonical IP formulation requires $\Omega(T\cdot |V|)$ integer variables for \mfft~and $\Omega(2^{|V|}\cdot |V|)$ integer variables for \mff, which is already problematic for very small graphs.

\clearpage

\bibliography{references}

@article{Abramovskaya2015HowTH,
title = {How to hunt an invisible rabbit on a graph},
journal = {European Journal of Combinatorics},
volume = {52},
pages = {12-26},
year = {2016},
issn = {0195-6698},
doi = {10.1016/j.ejc.2015.08.002},
author = {Tatjana V. Abramovskaya and Fedor V. Fomin and Petr A. Golovach and Michał Pilipczuk}
}

@article{Gruslys2015CatchingAM,
  title={Catching a mouse on a tree},
  author={Vytautas Gruslys and Ares M'eroueh},
  journal={arXiv: Combinatorics},
  year={2015},
  url={https://api.semanticscholar.org/CorpusID:117260205}
}

@article{Bolkema2017HuntingRO,
  title={Hunting rabbits on the hypercube},
  author={Jessalyn Bolkema and Corbin Groothuis},
  journal={Discret. Math.},
  year={2017},
  volume={342},
  pages={360-372},
  doi={10.1016/j.disc.2018.10.011}
}

@article{Britnell2012FindingAP,
author = {Britnell, John and Wildon, Mark},
year = {2012},
month = {04},
pages = {},
title = {Finding a Princess in a Palace: a Pursuit-Evasion Problem},
volume = {20},
journal = {The Electronic Journal of Combinatorics},
doi = {10.37236/2296}
}

@book{10.5555/574848,
author = {Garey, Michael R. and Johnson, David S.},
title = {Computers and Intractability; A Guide to the Theory of NP-Completeness},
year = {1990},
isbn = {0716710455},
publisher = {W. H. Freeman \& Co.},
address = {USA}
}

@article{Bernshteyn2021SearchingFA,
  title={Searching for an Intruder on Graphs and Their Subdivisions},
  author={Anton Bernshteyn and Eugene Eu-Juin Lee},
  journal={Electron. J. Comb.},
  year={2021},
  volume={29},
  doi={10.37236/10577}
}

@article{https://doi.org/10.1002/net.22284,
author = {Ben-Ameur, Walid and Maddaloni, Alessandro},
title = {Complexity Results for a Cops and Robber Game on Directed Graphs},
journal = {Networks},
volume = {86},
number = {2},
pages = {144-156},
doi = {https://doi.org/10.1002/net.22284},
year = {2025}
}

@incollection{SCHEFFLER1992287,
title = {Optimal Embedding of a Tree into an Interval Graph in Linear Time},
editor = {Jaroslav Neŝetril and Miroslav Fiedler},
series = {Annals of Discrete Mathematics},
publisher = {Elsevier},
volume = {51},
pages = {287-291},
year = {1992},
booktitle = {Fourth Czechoslovakian Symposium on Combinatorics, Graphs and Complexity},
issn = {0167-5060},
doi = {10.1016/S0167-5060(08)70644-7},
author = {Petra Scheffler}
}

@article{Alspach2006SEARCHINGAS,
  title={SEARCHING AND SWEEPING GRAPHS: A BRIEF SURVEY},
  author={Brian Alspach},
  journal={Le Matematiche},
  year={2006},
  volume={59},
  pages={5-37},
  url={https://api.semanticscholar.org/CorpusID:15575799}
}

@article{Fomin2008AnAB,
title = {An annotated bibliography on guaranteed graph searching},
journal = {Theoretical Computer Science},
volume = {399},
number = {3},
pages = {236-245},
year = {2008},
note = {Graph Searching},
issn = {0304-3975},
doi = {10.1016/j.tcs.2008.02.040},
author = {Fedor V. Fomin and Dimitrios M. Thilikos}
}

@book{bonatoBook,
author = {Bonato, Anthony and Nowakowski, Richard},
year = {2011},
month = {09},
pages = {},
title = {The Game of Cops and Robbers on Graphs},
isbn = {9780821853474},
doi = {10.1090/stml/061}
}

@article{hahnSurvey,
author = {Hahn, Gena},
year = {2007},
month = {01},
pages = {},
title = {Cops, robbers and graphs},
volume = {36},
journal = {Tatra Mountains Mathematical Publications}
}

@Inbook{Bonato2013,
author={Bonato, Anthony
and Yang, Boting},
title={Graph Searching and Related Problems},
bookTitle={Handbook of Combinatorial Optimization},
year={2013},
publisher={Springer New York},
address={New York, NY},
pages={1511--1558},
isbn={978-1-4419-7997-1},
doi={10.1007/978-1-4419-7997-1_76}
}

@inproceedings{tovsic1985vertex,
  title={Vertex-to-vertex search in a graph},
  author={To{\v{s}}i{\'c}, Ratko},
  booktitle={Proceedings of the Sixth Yugoslav Seminar on Graph Theory},
  pages={233--237},
  year={1985},
  organization={University of Novi Sad}
}

@article{HASLEGRAVE20141,
title = {An evasion game on a graph},
journal = {Discrete Mathematics},
volume = {314},
pages = {1-5},
year = {2014},
issn = {0012-365X},
doi = {10.1016/j.disc.2013.09.004},
author = {John Haslegrave}
}

@article{SAVITCH1970177,
title = {Relationships between nondeterministic and deterministic tape complexities},
journal = {Journal of Computer and System Sciences},
volume = {4},
number = {2},
pages = {177-192},
year = {1970},
issn = {0022-0000},
doi = {10.1016/S0022-0000(70)80006-X},
author = {Walter J. Savitch}
}

@article{Erds1963TheMR,
  title={The Minimal Regular Graph Containing a Given Graph},
  author={Paul Erd{\"o}s and Paul Joseph Kelly},
  journal={American Mathematical Monthly},
  year={1963},
  volume={70},
  pages={1074},
  url={https://api.semanticscholar.org/CorpusID:124944769}
}

@article{MAMINO201348,
title = {On the computational complexity of a game of cops and robbers},
journal = {Theoretical Computer Science},
volume = {477},
pages = {48-56},
year = {2013},
issn = {0304-3975},
doi = {10.1016/j.tcs.2012.11.041},
author = {Marcello Mamino}
}

@article{FOMIN20101167,
title = {Pursuing a fast robber on a graph},
journal = {Theoretical Computer Science},
volume = {411},
number = {7},
pages = {1167-1181},
year = {2010},
issn = {0304-3975},
doi = {10.1016/j.tcs.2009.12.010},
author = {Fedor V. Fomin and Petr A. Golovach and Jan Kratochvíl and Nicolas Nisse and Karol Suchan}
}

@article{ALSPACH2008158,
title = {Time constrained graph searching},
journal = {Theoretical Computer Science},
volume = {399},
number = {3},
pages = {158-168},
year = {2008},
note = {Graph Searching},
issn = {0304-3975},
doi = {10.1016/j.tcs.2008.02.035},
author = {Brian Alspach and Danny Dyer and Denis Hanson and Boting Yang}
}

@article{finbow,
  author       = {Stephen Finbow and
                  Gary MacGillivray},
  title        = {The Firefighter Problem: a survey of results, directions and questions},
  journal      = {Australas. {J} Comb.},
  volume       = {43},
  pages        = {57--78},
  year         = {2009},
  url          = {http://ajc.maths.uq.edu.au/pdf/43/ajc\_v43\_p057.pdf},
  bibsource    = {dblp computer science bibliography, https://dblp.org}
}

@article{graphBurning,
author = {Bonato, Anthony},
journal = {Contributions to Discrete Mathematics},
year = {2021},
volume = {16},
month = {03},
pages = {},
title = {A survey of graph burning},
doi = {https://doi.org/10.11575/cdm.v16i1.71194}
}

@misc{benameur2025huntingrabbithard,
      title={Hunting a rabbit is hard}, 
      author={Walid Ben-Ameur and Harmender Gahlawat and Alessandro Maddaloni},
      year={2025},
      eprint={2502.15982},
      archivePrefix={arXiv},
      primaryClass={math.CO}
}

\clearpage
\setboolean{mainpart}{false} 
\appendix
\section{Appendix}\label{app}
\subsection{Omitted Proofs of Section \ref{ch:basicProps}: Basic Properties and Bounds}\label{app:basicProps}
We start by giving the missing proofs of the upper and lower bounds on the firefighter number.
\lemmaLowerBounds*
\begin{proof}[Proof of 2] Follows directly by Lemma \ref{lem:lowerbounds}.1 for $i=1$. 
    An intuitive proof is the following: If $B_t = V$ and we use at most $m-1$ firefighters, then each node in $\tilde{E}_t = B_t \setminus F_t$ has at least one provisionally burning neighbour left, since at most $m - 2$ of its at least $m-1$ neighbours can also be in $F_t$. Hence, $B_{t + 1} = B_t = V$ and it is impossible to make any progress with only $m-1$ firefighters.
\end{proof}

\begin{proof}[Proof of 3] Let  $G'\subseteq G$ with $\ffn(G')\geq m$.
Suppose $\ffn(G)<m$ and let $S = (F_0, \dots, F_T)$ be a winning $\ffn(G)$-strategy for $G$.
Then, $S' = (F_0', \dots, F_T')$ with $F_i' := F_i \cap V'$ for $i \in [T]$ is a winning $\ffn(G)$-strategy for $G'$, which contradicts $\ffn(G')\geq m$.
\end{proof}

\begin{proof}[Proof of 4]
We define $f_m(n) \coloneqq m \cdot (n - \frac{m+1}{2})$ and perform a proof by contraposition, i.e., we show that $\ffn(G) < m$ implies $|E| < f_m(|V|)$.

    For $|V| < m$, we have $f_m(|V|) > m \cdot (m - \frac{m + 1}{2}) = \binom{|V|}{2}$.
    Since $\binom{|V|}{2}$ is the the number of edges in $K_{|V|}$, we have $|E| \leq \binom{|V|}{2} < \binom{m}{2} < f_m(|V|)$ for any graph $G = (V, E)$.
    
    For $|V| \geq m$, we prove the statement by induction over $|V|$.
    If $|V| = m$ and $\ffn(G) < m$, then $G = (V, E)$ cannot be the complete graph due to Lemma~\ref{lem:lowerbounds}.2., which implies $|E| < \binom{|V|}{2} = f_m(|V|)$.   
    Now assume that the statement holds for graphs with $|V| = n$ for some $n \in \N_{>0}$.
    Let $G$ be a graph with $n + 1$ nodes and $\ffn(G) < m$.
    By Lemma \ref{lem:lowerbounds}.2, $G$ contains a node $v$ with less than $m - 1$ neighbours.
    Removing this node and its adjacent edges yields a graph $G'$ with $\ffn(G') < m$ and $n$ nodes.
    By the induction hypothesis, $G'$ has less than $f_m(n)$ edges.
    Therefore, $G$ has less than $f_m(n) + m - 1 < f_m(n + 1)$ edges.
\end{proof}

\medskip
\noindent Note that Lemma \ref{lem:lowerbounds}.1 can give arbitrarily bad lower bounds for some graphs as shown in the following lemma. 
\begin{lemma}\label{lem:badInstances}
    For each $m \in \N_{\geq 1}$ there exists a graph $G_m=(V_m,E_m)$ with $\ffn(G)=m$ but for all $i\in [|V| -1]$ there exists a $W_i\subseteq V_m$ with $|N(W_i)|=1$.
\end{lemma}
\begin{proof}
     Let $G_m$ be the graph consisting of a $m$-clique $K_m=(V_{K_m},E_{K_m})$ and a path $v_1,\ldots,v_m$ where any node of $K_m$ is connected to $v_m$.
     As visualized in Figure \ref{fig:pathClique}, for any $i\in [m]$, the set $W_i=\{v_1,\ldots v_i\}$ has only one adjacent node, and for $i\in [|V|-1] \setminus [m]$ the set $W_i=V_{K_m} \cup \{v_m,\ldots , v_{m-(i-m)-1}\}$ has only one adjacent node, proving the claim.
    \begin{figure}[b]
            \centering
    \subfloat[Example for $i=k-1<k+1$.]{
\begin{tikzpicture}[scale=0.4, every node/.style={font=\footnotesize}]
				
				\node[draw, line width=0.4mm, circle, minimum size=1.2cm, inner sep=0pt] (K1) at (-7,0) {$K_k$};
				
				% Define the positions of the nodes
				\node[draw, circle, minimum size=0.7cm, inner sep=0pt, fill=white] (vk) at (-4,0) {$v_k$};
				\node[draw, circle, minimum size=0.7cm, inner sep=0pt, fill=blue!40] (vk1) at (-2,0) {$v_{k-1}$};
				\node[draw, circle, minimum size=0.7cm, inner sep=0pt, fill=blue!40] (vk2) at (0,0) {$v_{k-2}$};
				\node[draw, circle, minimum size=0.7cm, inner sep=0pt, fill=blue!40] (v1) at (4,0) {$v_1$};
                
				\fill (1.8,0) circle (0.05cm);
                \fill (2,0) circle (0.05cm);
                \fill (2.2,0) circle (0.05cm);
                % \node at (-2.2, 0.75) {\;};	
				
				\draw[thick] (K1) -- (vk);
				\draw[thick] (vk) -- (vk1);
				\draw[thick] (vk1) -- (vk2);
				\draw[thick]  (vk2) -- (1.5,0);
                \draw[thick] (2.5,0) -- (v1);
				
			\end{tikzpicture}
}
\hspace{0.5cm}
\subfloat[Example for $i=k+2>k$.]{
\begin{tikzpicture}[scale=0.4, every node/.style={font=\footnotesize}]
				
				\node[draw, line width=0.4mm, circle, minimum size=1.2cm, inner sep=0pt,fill=blue!40] (K1) at (-7,0) {$K_k$};
				
				% Define the positions of the nodes
				\node[draw, circle, minimum size=0.7cm, inner sep=0pt, fill=blue!40] (vk) at (-4,0) {$v_k$};
				\node[draw, circle, minimum size=0.7cm, inner sep=0pt, fill=blue!40] (vk1) at (-2,0) {$v_{k-1}$};
				\node[draw, circle, minimum size=0.7cm, inner sep=0pt, fill=white] (vk2) at (0,0) {$v_{k-2}$};
				\node[draw, circle, minimum size=0.7cm, inner sep=0pt, fill=white] (v1) at (4,0) {$v_1$};
                
				\fill (1.8,0) circle (0.05cm);
                \fill (2,0) circle (0.05cm);
                \fill (2.2,0) circle (0.05cm);
                % \node at (-2.2, 0.75) {\;};	
				
				\draw[thick] (K1) -- (vk);
				\draw[thick] (vk) -- (vk1);
				\draw[thick] (vk1) -- (vk2);
				\draw[thick]  (vk2) -- (1.5,0);
                \draw[thick] (2.5,0) -- (v1);
				
			\end{tikzpicture}
}
        \caption{Lemma \ref{lem:badInstances}: $G_m$ is constructed from a $m$-clique and a path with $m$ nodes
        The node set $W_i$ (colored in blue) has only one adjacent node.}
        \label{fig:pathClique}
    \end{figure}
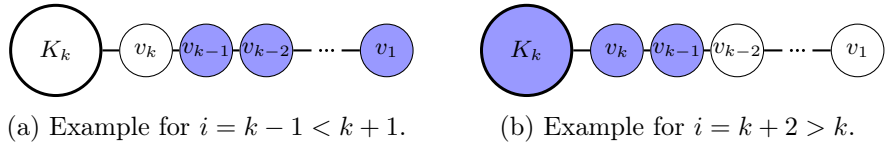
\end{proof}

    \lemmaUpperBounds*
    \noindent \begin{proof}[Proof of 2] We prove this by induction.
    For $m = 1$, $G$ can only consist of a single node, so $\ffn(G) \leq 1$ is obvious.
    Next, we assume that the statement holds for some $m \in \N_{>0}$, i.e., any tree $T$ with $\diam(T) \leq 2m - 2$ satisfies $\ffn(T) \leq m$. 
    Let $G$ be a tree with $\diam(G) \leq 2(m+1)-2 = 2m$.
    First, we choose any longest path in $G$ and set $r$ to be its middle node (or one of the two middle nodes if the longest path has odd length).
    Note that there cannot be a node $v$ with $\dist(v, r) > m$, since that would imply the existence of a longer path.
    Removing $r$ splits $G$ into disjoint trees $T_1, \dots, T_n$. Note that for each $i \in [n]$, $r$ has exactly one neighbour $r_i \in T_i$.

    For any two nodes $v$ and $w$ in such a tree $T_i$, the paths from $v$ to $r$ and from $w$ to $r$ both have to contain $r_i$ and have length less or equal to $m - 1$. Therefore, we have $\dist(v, w) \leq 2m - 2$, since we can construct a path from $v$ to $w$ by connecting the paths from $v$ to $r$ and from $r$ to $w$ and removing the unnecessary step from $r_i$ to $r$ and back.
   Therefore, each $T_i$ is a tree with diameter at most $2m - 2$.
   By the assumption of the induction, there is a winning $m$-strategy for each $T_i$.
   Concatenating these strategies and additionally including $r$ in each of the resulting firefighter sets yields a winning $(m + 1)$-strategy for $G$, which finishes the proof.
\end{proof}
\begin{proof}[Proof of 3] Let $S=(F_1,\ldots, F_n)$ be a winning $(m-k)$-strategy for $G_k$.
Since $|V\setminus V_K|\leq k$, the strategy $S'=(F_1\cup (V\setminus V_K), \ldots, F_n\cup (V\setminus V_K))$ is a winning $m$-strategy for $G$.
\end{proof}

\noindent We can now use these bounds to characterize all graphs with firefighter numbers 1 and 2, and give bounds for $d$-regular graphs.
\lemmaChar*
\begin{proof}
If $G$ has an edge $\{u,v\}$, the nodes $u$ and $v$ form a 2-clique.
With Lemma \ref{lem:lowerbounds}.3 and $\ffn(K_2)=2$, it follows that $\ffn(G) \geq 2$.
If $G$ does not have any edges and nodes $V = \{v_1, \dots, v_n\}$, then $S=(\{v_1\}, \dots, \{v_n\})$ is a winning $1$-strategy and we have $\ffn(G) = 1$.

It is well known that caterpillar graphs have pathwidth $1$, hence Lemma \ref{lem:upperbounds}.1 implies that all caterpillar graphs have firefighter number at most $2$.
If $G$ is not a caterpillar graph, it has to contain a cycle $C_\ell$ of length $\ell \geq 3$, or the tree $T$ which is constructed by attaching three paths of length two to a single node. It is easy to verify that all subsets of $T$ containing three nodes have at least $2$ neighbours in $T$, and any arbitrary node in $C_\ell$ has at least $2$ neighbours in $C_\ell$. Hence, by Lemma \ref{lem:lowerbounds}.1, $G$ contains a subgraph $G'$ with $\ffn(G') > 2$, which implies $\ffn(G) > 2$.
\end{proof}

\corExactFfn*
\begin{proof}
    Since we have $\min_{v\in V(C_n)} \delta(v)=2$, $\min_{v\in V(K_n)} \delta(v) = n-1$ and $\min_{v\in V(K_{n,m})} \delta(v)$ $= \min\{n,m\}$, Lemma \ref{lem:lowerbounds}.2 yields $\ffn(C_n)> 2$, $\ffn(K_n)> n-1$ and $\ffn(K_{n,m})>\min\{n,m\}$.
    Extinguishing all $n$ nodes simultaneously is a winning $n$-strategy for $K_n$.
    Removing the smaller of the two node sets of the bipartition of $K_{n,m}$ leaves the graph unconnected and thus, according to Lemma \ref{lem:upperbounds}.3, there is a winning $(\min\{n,m\}+1)$-strategy on $K_{n,m}$.
    Similarly, removing a node from $C_n$ leaves the graph as a caterpillar graph, and hence by Lemma \ref{lem:upperbounds}.3 and Corollary \ref{prop:char} we have $\ffn(C_n)\leq 3$.
    Note that $\ffn(K_n)=n$ was already proved in~\cite{Bernshteyn2021SearchingFA} and the proof is only given here for the sake of completeness.
\end{proof}

\corDRegular*
\begin{proof}
    For $d=1$ the graph $G$ is a union of caterpillar graphs with $|E|>0$, which implies $\ffn(G) = 2$ due to Corollary \ref{prop:char}.
    For $d=2$ the graph $G$ is a union of cycles, hence $\ffn(G) = 3$ due to Corollary \ref{prop:graphs}.
    Let $d>3$. Every complete binary tree $T=(V,E)$ has $\max_{v \in V} \delta(v) \leq 3 \leq d$ and can therefore be extended to a $d$-regular $G'$ for every $d\in \N_{\geq 3}$, see~\cite{Erds1963TheMR}. 
    According to Theorem \ref{thm:binTree}, complete binary trees can have arbitrarily high firefighter values.
    Lemma \ref{lem:lowerbounds}.3 implies that the same holds for $G'$.
    Furthermore, a $d$-regular graph $G$ has $\min_{v \in V} \delta(v)=d$ and Lemma \ref{lem:lowerbounds}.2 implies $\ffn(G)\geq d+1$. We can see that this lower bound is tight by noting that $K_{d+1}$ is $d$-regular and $\ffn(K_{d+1})=d+1$.
\end{proof}

\forst*
\begin{proof}
    Due to \cite{SCHEFFLER1992287}, for every forest the pathwidth of any forest of order $n$ is upper bounded by $\log_3(2n + 1)+1$. Together with Lemma \ref{lem:upperbounds}.1, this implies the claim.
\end{proof}

\noindent As a short overview, the following table lists for some common graph parameters whether they can, in general, be used to find an upper or lower bound on the firefighter number. We split the statements into several lemmata and prove them separately.
\begin{table}[H]
    \centering
    \begin{tabular}{c||c|c}
        Graph parameter &  \;Upper bound on $\ffn$ \;&  \;Lower bound on $\ffn$  \;\\
        \hline \hline
        \;Minimum degree \;$\delta_{\min}$& - & $\delta_{\min} + 1\leq \ffn(G)$
        \\
        \hline
         \;Maximum degree  \;$\delta_{\max}$ & - & - \\
        \hline
        Order $|V|$ & $\ffn(G) \leq |V|$ & - \\
        \hline
        Depth $d$ of a tree $T$ & $\ffn(T) \leq d + 1$ & - \\
        \hline
        Beta index $|E|/|V|$ & - & \;      
        $|E| \leq (\ffn(G) - 1) \cdot (|V| - \ffn(G)/2)$  \;\\
        \hline
        Pathwidth $\pw(G)$ & $\ffn(G) \leq \pw(G) + 1$ & - \\
        \hline
        Treewidth $\tw(G)$ & - & ? \\
        \hline
        Vertex cover number& \multirow{2}{*}{$\ffn(G) \leq \vcn(G) + 1$} & \multirow{2}{*}{-} \\
        $\vcn(G)$& & \\
    \end{tabular}
    \vspace{10pt}
    \caption{Tight upper and lower bounds on the firefighter number in terms of some graph parameters.
    For every missing entry except for the question mark, a bound based on this parameter is not possible.}
\end{table}

\begin{lemma}[Minimum Degree]
    For any graph $G$, we have $\ffn(G) \geq \delta_{\min}(G) + 1$.
    There is a class of graphs with bounded minimum degree and arbitrarily large firefighter number.
\end{lemma}
\begin{proof}
    $\ffn(G) \geq \delta_{\min}(G) + 1$ follows directly from Lemma \ref{lem:lowerbounds}.2, and is a tight bound since $\ffn(K_n) = n = \delta_{\min}(K_n) + 1$.
    As proven in Theorem \ref{thm:binTree} and in~\cite{Bernshteyn2021SearchingFA}, binary trees can reach arbitrarily high firefighter number, while their minimum degree is $1$.
\end{proof}

\begin{lemma}[Maximum Degree]
    There is a class of graphs with bounded firefighter number and arbitrarily large maximum degree, and there is a class of graphs with bounded maximum degree and arbitrarily large firefighter number.
\end{lemma}

\begin{proof}
    Any star graph has firefighter number less or equal to $2$ since it is a caterpillar graph, but star graphs can reach arbitrarily high maximum degrees.
    As proven in Theorem \ref{thm:binTree} and in~\cite{Bernshteyn2021SearchingFA}, binary trees can reach an arbitrarily high firefighter number, while their maximum degree is bounded by $3$.
\end{proof}

\begin{lemma}[Order]
    There is a class of graphs with bounded firefighter number and arbitrarily large order, i.e., number of nodes.
    The firefighter number of a graph is less or equal to its order. This bound is tight.  
\end{lemma}

\begin{proof}
    Caterpillar graphs can have arbitrarily high order, but have firefighter number less or equal to $2$.
    Every graph of order $n$ is a subgraph of $K_n$ and thus has firefighter number at most $n$, since $\ffn(K_n)=n$.
\end{proof}

\begin{lemma}[Depth of a tree]
    For any tree $T$ with depth $d \in \N_{\geq 0}$, we have $\ffn(T) \leq d + 1$, which is a tight bound since there exists a tree $T_d$ with depth $d$ and $\ffn(T) = d + 1$ for each such $d$.
    There is a class of trees with arbitrary depth and bounded firefighter number.
\end{lemma}
\begin{proof}
    Caterpillar graphs can have arbitrary depth, but their firefighter number is bounded by $2$.    
    The bound $\ffn(T) \leq d + 1$ follows immediately from Lemma~\ref{lem:upperbounds}.2., since $\diam(G) \leq 2d$.    
    By Corollary \ref{prop:char}, any tree of depth $0$ (resp. $1$) has firefighter number $1$ (resp. $2$).
    
    For $d \geq 3$, consider the complete tree $T_d$ of depth $d$ with branching degree $d^2 - 1$.
    Suppose there exists a winning $d$-strategy for $T_d$.
    Then there exists an earliest point in time $t$ when the root node is in $E_t$.
    Therefore, at time $t-1$, at most $d-1$ of its children are burning, hence at least $d^2-d$ child nodes are extinguished at time $t-1$.
    
    We denote these child nodes as the \emph{relevant child nodes}.
    Let $t'$ denote the latest time before $t$ when the root node was not covered by a firefighter.
    Since the root node is burning at the end of each turn before $t$, all children of the root node were set on fire in turn $t'$.
    By the definition of $t'$, the root node is covered by a firefighter each turn between $t'$ and $t$.
    Since at most $d - 1$ nodes can be extinguished in each turn, the $d^2 - d$ relevant child nodes are extinguished at time $t - 1$ and no child nodes are extinguished at time $t'$, it follows that $t - t' - 1 \geq (d^2 - d) / (d - 1) = d$.
    We denote the subtrees whose root nodes are the relevant children of the original root node as the \textit{relevant subtrees}. 
    Since the relevant subtrees have depth $d - 1$, all of their nodes are burning at time $t' + d$, if none of their nodes were covered between $t'$ and $t' + d$.
    There are only $d - 1$ firefighters available between $t' + 1$ and $t$, hence at most $(d-1)^2$ nodes can be covered between $t'$ and $t' + d$.
    Since there are $d^2 - d$ subtrees whose root nodes are the relevant child nodes, it follows that at some point $t''$ between $t'$ and $t$, there is a relevant subtree that is completely on fire.
    
    We therefore have to be able to extinguish the root node of a complete tree with depth $d - 1$ and branching degree $d^2 - 1$ with $d - 1$ firefighters, since the original root node has to be covered every single turn between $t'$ and $t$.
    Thus, we can recursively repeat the same argument until we have to extinguish the root node of a complete tree of depth $1$ and branching degree $d^2 -1$ with only one firefighter, which is clearly not possible.
    By contradiction, it follows that $\ffn(T_d) > d$.
\end{proof}

\begin{lemma}[Beta Index]
    There exists a class of graphs with bounded beta index (which is defined as the ratio between the number of edges and the number of nodes) and an arbitrarily large firefighter number.
    If a class of graphs has bounded firefighter number, their beta index is also bounded.
    In particular, we have $|E| \leq (\ffn(G) - 1) \cdot(|V| - \ffn(G)/2)$, which implies $\frac{|E|}{|V|} \leq \ffn(G) - 1$.
\end{lemma}

\begin{proof}
    By Theorem \ref{thm:binTree} and~\cite{Bernshteyn2021SearchingFA}, binary trees can have arbitrarily large firefighter number. However, they have the smallest possible beta index for connected graphs, i.e., $\frac{|E|}{|V|} = \frac{|V| - 1}{|V|} < 1$.
    From Lemma \ref{lem:lowerbounds}.4 the rest of the statement follows. 
\end{proof}

\begin{lemma}[Pathwidth]\label{lem:pw}
    For any graph $G$, we have $\ffn(G) \leq \pw(G) + 1$.
    There exists a class of graphs with bounded firefighter number and arbitrarily large pathwidth.
\end{lemma}

\begin{proof}
    The bound is given in Lemma \ref{lem:upperbounds}.1.
    Furthermore, the authors of \cite{Bernshteyn2021SearchingFA} showed that each tree can be subdivided such that it is $3$-winning.
    This allows us to construct a class of graphs as described in the lemma, since trees can have arbitrarily large pathwidth and subdividing edges of a graph does not reduce the pathwidth of the graph. 
\end{proof}

\begin{lemma}[Treewidth]
    There is a class of graphs with bounded treewidth and arbitrarily large firefighter number, and a class of graphs that satisfy $\ffn(G) \leq \lceil (\tw(G) + 1)/3 \rceil + 1$.
\end{lemma}

\begin{proof}
    Trees have a treewidth of $1$ and can have arbitrarily large firefighter number, see Theorem \ref{thm:binTree} or~\cite{Bernshteyn2021SearchingFA}.
    It is well known that $K_n$ has treewidth $n - 1$ and subdividing edges does not decrease the treewidth.
    In \cite{Bernshteyn2021SearchingFA}, it was shown that there exists a subdivision of any $K_n$ with $n\geq 4$ that has firefighter number less or equal to $\lceil n/3 \rceil + 1$, which shows that $\lceil (\tw(G) + 1)/3 \rceil + 1$ is the best possible lower bound on $\ffn(G)$ in terms of $\tw(G)$ that one can hope for. 
\end{proof}

\begin{lemma}[Vertex Cover Number]
    For any graph $G$, we have $\ffn(G) \leq \vcn(G) + 1$.
    There exists a class of graphs with bounded firefighter number and arbitrarily large vertex cover number.
\end{lemma}
\begin{proof}
    Note that the pathwidth of a graph is upper bounded by the vertex cover number of the graph. Hence, it holds that $\ffn(G)\leq vcn(G) +1$ due to Lemma~\ref{lem:pw}.
    This is tight for a star and arbitrarily bad for spider graphs $S_i$ with $i$ legs of length 2, which have $\ffn(S_i)=3$ for all $i\geq 3$. 
\end{proof}

\subsection{Omitted Proofs of Section \ref{ch:binaryTrees}: Firefighting on Complete Binary Trees}\label{app:binaryTrees}

\thmBinTree*
\begin{proof}
    According to Corollary \ref{prop:char} we have $\ffn(\B_{0})=1$, $\ffn(\B_{1})=\ffn(\B_{2})=2$ and $\ffn(\B_3),$ $\ffn(\B_4) \geq 3$ since $\B_0$ has no edges, $\B_1$ and $\B_2$ are caterpillar graphs and $\B_3$ and $\B_4$ are not caterpillar graphs.
    Since $\B_5$ and $\B_6$ both contain a complete binary tree of depth $4$, this also implies $\ffn(\B_5), \ffn(\B_6) \geq \ffn(\B_4) \geq 3$.

    \textit{Proof of the upper bound for $d \geq 3$:}
    We use induction. The base case is already given above.
    Assume that a binary tree of even depth $d$, denoted $\B_d$, can be extinguished by $k \geq 1$ firefighters.
    To extinguish $\B_{d+1}$ with $k+1$ firefighters, we can use the additional firefighter to extinguish the root node in each step.
    For a binary tree $\B_{d + 2}$, let $r$ denote its root node and $r', r''$ its two children.
    A winning $k + 1$-strategy for $\B_{d + 2}$ is then given by the following procedure: First, use the previously given strategy to extinguish the subtree with root $r'$ (which is not affected by $r$ burning, since $r'$ is extinguished in each step of this strategy).
    Then extinguish $\{r', r\}$ and afterwards $\{r, r''\}$.
    Finally, extinguish the subtree with root node $r''$ according to the previously given strategy for $\B_{d + 1}$, which does not let the fire spread back to $r$, since $r''$ is extinguished in each step.
    
    \textit{Proof of the lower bound for $d \geq 5$:}
    First of all, note that for odd $d$, the lower bound remains the same for $d + 1$ and $\B_d \subseteq \B_{d + 1}$, which implies $\ffn(\B_{d + 1}) \geq \ffn(\B_d)$.
    Hence, it suffices to prove the bound for odd $d$.
    
    Let $\B_d = (V, E)$ be a complete binary tree of depth $d \in \N_{\geq 3}$ with root node $r$, and let the \emph{level} of a node $v \in V$ in this binary tree be denoted by $l(v) \coloneqq d - \dist(v, r)$, meaning that $l(r) = d$ and $l(v) = 0$ for any leaf $v$ of $\B_d$.
    For any $v \in V$, let $\B(v) \subseteq V$ be the set of nodes containing $v$ and its descendants in $\B_d$.
    Let $S \subseteq V$ be a subset of nodes in $\B_d$.
    Consider the following way of constructing $S$:
    
    Set $S' = V$.
    Then apply the following step for all $v \in V$, ordered in descending order by $l(v)$:
    If $v \in S'$ but $v \notin S$, remove $\B(v)$ from $S'$.
    If $v \notin S'$ but $v \in S$, add $\B(v)$ to $S'$.
    After this step has been applied for all $v \in V$, we have $S = S'$.
    An illustration of this process can be found in Figure \ref{fig:binTree}.

\begin{figure}[ht]
    \centering
    \begin{tikzpicture}[scale=0.3, every node/.style={font=\small}]
			\node[draw, circle, fill=blue!40] (v1) at (0,0) {};
			
			\node[draw, circle, fill=white] (v2) at (-1.5,-1.5) {};
			\node[draw, circle, fill=blue!40] (v3) at (1.5,-1.5) {};
			
			\node[draw, circle, fill=white] (v4) at (-2.1,-3) {};
			\node[draw, circle, fill=blue!40] (v5) at (-0.9,-3) {};
			\node[draw, circle, fill=blue!40] (v6) at (0.9,-3) {};
			\node[draw, circle, fill=white] (v7) at (2.1,-3) {};
			% Connect nodes with styled edges
			\draw[thick] (v1) -- (v2);
			\draw[thick] (v1) -- (v3);
			
			\draw[thick] (v2) -- (v4);
			\draw[thick] (v2) -- (v5);
			
			\draw[thick] (v3) -- (v6);
			\draw[thick] (v3) -- (v7);

                \node[] (eq) at (4,-1.5) {$=$};

                % Tree 2
                \node[draw, circle, fill=blue!40] (vv1) at (8,0) {};
			
			\node[draw, circle, fill=blue!40] (vv2) at (6.5,-1.5) {};
			\node[draw, circle, fill=blue!40] (vv3) at (9.5,-1.5) {};
			
			\node[draw, circle, fill=blue!40] (vv4) at (5.9,-3) {};
			\node[draw, circle, fill=blue!40] (vv5) at (7.1,-3) {};
			\node[draw, circle, fill=blue!40] (vv6) at (8.9,-3) {};
			\node[draw, circle, fill=blue!40] (vv7) at (10.1,-3) {};
                \draw[thick] (vv1) -- (vv2);
			\draw[thick] (vv1) -- (vv3);
			\draw[thick] (vv2) -- (vv4);
			\draw[thick] (vv2) -- (vv5);
			\draw[thick] (vv3) -- (vv6);
			\draw[thick] (vv3) -- (vv7);

                \node[] (eq) at (12,-1.5) {$-$};

                % Tree 3
                \node[draw, gray!50, circle, fill=white] (vv1) at (16,0) {};
			
			\node[draw, circle, fill=blue!40] (vv2) at (14.5,-1.5) {};
			\node[draw, gray!50, circle, fill=white] (vv3) at (17.5,-1.5) {};
			
			\node[draw, circle, fill=blue!40] (vv4) at (13.9,-3) {};
			\node[draw, circle, fill=blue!40] (vv5) at (15.1,-3) {};
			\node[draw, gray!50, circle, fill=white] (vv6) at (16.9,-3) {};
			\node[draw, gray!50, circle, fill=white] (vv7) at (18.1,-3) {};
                \draw[thick, gray!50] (vv1) -- (vv2);
			\draw[thick, gray!50] (vv1) -- (vv3);
			\draw[thick] (vv2) -- (vv4);
			\draw[thick] (vv2) -- (vv5);
			\draw[thick, gray!50] (vv3) -- (vv6);
			\draw[thick, gray!50] (vv3) -- (vv7);

                \node[] (eq) at (20,-1.5) {$+$};

                % Tree 4
                \node[draw, gray!50, circle, fill=white] (vv1) at (24,0) {};
			
			\node[draw, gray!50, circle, fill=white] (vv2) at (22.5,-1.5) {};
			\node[draw, gray!50, circle, fill=white] (vv3) at (25.5,-1.5) {};
			
			\node[draw, gray!50, circle, fill=white] (vv4) at (21.9,-3) {};
			\node[draw, circle, fill=blue!40] (vv5) at (23.1,-3) {};
			\node[draw, gray!50, circle, fill=white] (vv6) at (24.9,-3) {};
			\node[draw, gray!50, circle, fill=white] (vv7) at (26.1,-3) {};
                \draw[thick, gray!50] (vv1) -- (vv2);
			\draw[thick, gray!50] (vv1) -- (vv3);
			\draw[thick, gray!50] (vv2) -- (vv4);
			\draw[thick, gray!50] (vv2) -- (vv5);
			\draw[thick, gray!50] (vv3) -- (vv6);
			\draw[thick, gray!50] (vv3) -- (vv7);

                \node[] (eq) at (28,-1.5) {$-$};

                % Tree 5
                \node[draw, gray!50, circle, fill=white] (vv1) at (32,0) {};
			
			\node[draw, gray!50, circle, fill=white] (vv2) at (30.5,-1.5) {};
			\node[draw, gray!50, circle, fill=white] (vv3) at (33.5,-1.5) {};
			
			\node[draw, gray!50, circle, fill=white] (vv4) at (29.9,-3) {};
			\node[draw, gray!50,  circle, fill=white] (vv5) at (31.1,-3) {};
			\node[draw, gray!50, circle, fill=white] (vv6) at (32.9,-3) {};
			\node[draw, circle, fill=blue!40] (vv7) at (34.1,-3) {};
                \draw[gray!50, thick] (vv1) -- (vv2);
			\draw[gray!50, thick] (vv1) -- (vv3);
			\draw[gray!50, thick] (vv2) -- (vv4);
			\draw[gray!50, thick] (vv2) -- (vv5);
			\draw[gray!50, thick] (vv3) -- (vv6);
			\draw[gray!50, thick] (vv3) -- (vv7);
    \end{tikzpicture}		
    \caption{Decomposing a subset of a complete binary tree into multiple complete binary trees of varying depths.}
    \label{fig:binTree}
\end{figure}
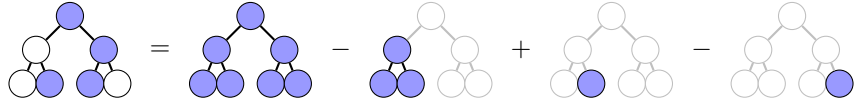

    This construction allows us to express the size of $S$ as a sum of positives and negatives of the sizes of complete binary trees of varying depths.
    Note that during this process, $\B(v)$ is added to $S'$ if and only if $v \in S$ and $p(v) \notin S$, where $p(v)$ denotes the parent of $v$ in $\B_d$.
    Similarly, $\B(v)$ is removed from $S'$ if and only if $v \notin S$ and $p(v) \in S$, or $v = r \notin S$.
    Therefore, it suffices to evaluate the root node and nodes that are incident to an edge connecting $S$ and $N(S)$.
    To that end, we partition $N(S) = C_1 \cup C_2 \cup C_3 \cup C_4 \cup C_5$ with $C_1 \coloneqq \{v \in N(S): p(v), c(v), c'(v) \in S\}$, \;$C_2 \coloneqq \{v \in N(S): p(v) \notin S, c(v), c'(v) \in S\}$, \;$C_3 \coloneqq \{v \in N(S): p(v), c(v) \in S, c'(v) \notin S\}$, \;$C_4 \coloneqq \{v \in N(S): p(v), c(v) \notin S, c'(v) \in S\}$ and \;$C_5 \coloneqq \{v \in N(S): p(v) \in S, c(v), c'(v) \notin S\}$,
    where $c(v)$ and $c'(v)$ denote the children of $v$ (in no particular order).
    If $p(v)$ or $c(v), c'(v)$ do not exist (i.e., if $v$ is the root or a leaf node), we set $p(v) \notin S$ or $c(v), c'(v) \notin S$.
    See Figure \ref{fig:C1toC5} for an illustration of the nodes in these sets.

    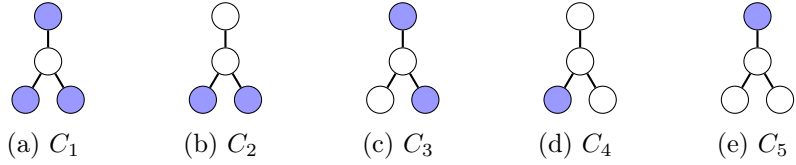
\begin{figure}[H]
		\centering
        \subfloat[$C_1$]{
         \begin{tikzpicture}[scale=0.3, every node/.style={font=\small}]
			\node[draw, circle, fill=blue!40] (v1) at (-1.5,0) {};
			
			\node[draw, circle, fill=white] (v2) at (-1.5,-2) {};
			
			\node[draw, circle, fill=blue!40] (v4) at (-2.5,-3.7) {};
			\node[draw, circle, fill=blue!40] (v5) at (-0.5,-3.7) {};
			% Connect nodes with styled edges
			\draw[thick] (v1) -- (v2);
			
			\draw[thick] (v2) -- (v4);
			\draw[thick] (v2) -- (v5);

    \end{tikzpicture}	
    }
    \hspace{1cm}
    \subfloat[$C_2$]{
         \begin{tikzpicture}[scale=0.3, every node/.style={font=\small}]
			\node[draw, circle, fill=white] (v1) at (-1.5,0) {};
			
			\node[draw, circle, fill=white] (v2) at (-1.5,-2) {};
			
			\node[draw, circle, fill=blue!40] (v4) at (-2.5,-3.7) {};
			\node[draw, circle, fill=blue!40] (v5) at (-0.5,-3.7) {};
			% Connect nodes with styled edges
			\draw[thick] (v1) -- (v2);
			
			\draw[thick] (v2) -- (v4);
			\draw[thick] (v2) -- (v5);

    \end{tikzpicture}	
    }
    \hspace{1cm}
    \subfloat[$C_3$]{
         \begin{tikzpicture}[scale=0.3, every node/.style={font=\small}]
			\node[draw, circle, fill=blue!40] (v1) at (-1.5,0) {};
			
			\node[draw, circle, fill=white] (v2) at (-1.5,-2) {};
			
			\node[draw, circle, fill=white] (v4) at (-2.5,-3.7) {};
			\node[draw, circle, fill=blue!40] (v5) at (-0.5,-3.7) {};
			% Connect nodes with styled edges
			\draw[thick] (v1) -- (v2);
			
			\draw[thick] (v2) -- (v4);
			\draw[thick] (v2) -- (v5);

    \end{tikzpicture}	
    }
    \hspace{1cm}
    \subfloat[$C_4$]{
         \begin{tikzpicture}[scale=0.3, every node/.style={font=\small}]
			\node[draw, circle, fill=white] (v1) at (-1.5,0) {};
			
			\node[draw, circle, fill=white] (v2) at (-1.5,-2) {};
			
			\node[draw, circle, fill=blue!40] (v4) at (-2.5,-3.7) {};
			\node[draw, circle, fill=white] (v5) at (-0.5,-3.7) {};
			% Connect nodes with styled edges
			\draw[thick] (v1) -- (v2);
			
			\draw[thick] (v2) -- (v4);
			\draw[thick] (v2) -- (v5);

    \end{tikzpicture}	
    }
    \hspace{1cm}
    \subfloat[$C_5$]{
         \begin{tikzpicture}[scale=0.3, every node/.style={font=\small}]
			\node[draw, circle, fill=blue!40] (v1) at (-1.5,0) {};
			
			\node[draw, circle, fill=white] (v2) at (-1.5,-2) {};
			
			\node[draw, circle, fill=white] (v4) at (-2.5,-3.7) {};
			\node[draw, circle, fill=white] (v5) at (-0.5,-3.7) {};
			% Connect nodes with styled edges
			\draw[thick] (v1) -- (v2);
			
			\draw[thick] (v2) -- (v4);
			\draw[thick] (v2) -- (v5);

    \end{tikzpicture}	
    }

		\caption{Illustration of nodes in the set $C_1,\ldots, C_5$. The middle node is in the corresponding set $C_i$, nodes in $S$ are colored blue.}
		\label{fig:C1toC5}
    \end{figure}

    For each such set, we can determine the contribution of a node in this set to the size of $S$ as calculated before.
    For example, a node $v \in C_3$ removes a complete binary tree of depth $l(v)$ from $S$, but adds a complete binary tree of depth $l(v) - 1$.
    Hence, its contribution to $|S|$ is $-(2^{l(v) + 1} - 1) + 2^{l(v)} - 1 = -2^{l(v)}$.
    Applying the same logic to the other sets, we calculate
    \begin{align*}
        |S| = 2^{d + 1} - 1 &- \sum_{v \in C_1} 1 + \sum_{v \in C_2} (2^{l(v) + 1} - 2) \\
        &- \sum_{v \in C_3} 2^{l(v)} 
        + \sum_{v \in C_4} (2^{l(v)} - 1)
        - \sum_{v \in C_5} (2^{l(v) + 1} - 1),
    \end{align*}
    which we can rearrange to
    \[
        |S| + 1 + c =
        \sum_{v \in C_2} 2^{l(v) + 1}
        + \sum_{v \in C_4} 2^{l(v)}
        - \sum_{v \in C_3} 2^{l(v)}
        - \sum_{v \in C_5} 2^{l(v) + 1}
    \]
    with $c \coloneqq |C_1| + 2 |C_2| + |C_4| - |C_5|$.
    Using this equation, we will now prove the following claim, which suffices to finish the proof due to Lemma \ref{lem:lowerbounds}.1.

    \medskip
    \noindent \textit{Claim:}
    Let $d \in \N_{\geq 3}$ be odd.
    Then there is no set $S \subseteq V$ with
	\begin{equation*}
		|S| = -1 + \sum_{j \in [ (d - 1)/2]_0} 2^{2j+1}
	\end{equation*}
	and
	\begin{equation*}
		|N(S)| < \frac{d - 1}{2} - \frac12 \log \left(\frac{d - 3}{2} \right) - 1.
    \end{equation*}
    \medskip

    \noindent Let $S$ be a set with cardinality as given in the claim.
    Inserting this in the last equation yields
    \begin{equation}
    \label{eq:flipeq}
        \sum_{j \in [ (d - 1)/2]_0} 2^{2j+1} + c =
        \sum_{v \in C_2} 2^{l(v) + 1}
        + \sum_{v \in C_4} 2^{l(v)}
        - \sum_{v \in C_3} 2^{l(v)}
        - \sum_{v \in C_5} 2^{l(v) + 1}.
    \end{equation}
    We want to deduce a lower bound to $|N(S)|$ from this equation.
    To this end, we use the following two concepts from information theory:
    The \emph{Hamming weight} of a number $x \in \N_{\geq 0}$, denoted as $\hw(x)$, is the number of times that the digit one appears in the binary representation of $x$.
    Moreover, we define $\flips(x)$ as the summed number of occurrences of the sequences $01$ and $10$ in the binary representation of $x$, i.e., the number of bit transitions in the binary representation of $x$.

    If Equation (\ref{eq:flipeq}) holds true, then the number of flips of the left-hand side must equal the number of flips of the right-hand side.
    Using the well-known inequalities $\hw(x + y) \leq \hw(x) + \hw(y)$ and $\flips(x - y) \leq 2 \cdot (\hw(x) + \hw(y))$ for $x, y \in \N_{\geq 0}$, we calculate the number of flips for the right-hand side as
    \begin{align*}
        \flips &\left(\sum_{v \in C_2} 2^{l(v) + 1}
        + \sum_{v \in C_4} 2^{l(v)}
        - \sum_{v \in C_3} 2^{l(v)}
        - \sum_{v \in C_5} 2^{l(v) + 1}\right) \\
        &\leq 2 \cdot (|C_2| + |C_3| + |C_4| + |C_5|) \leq 2 |N(S)|.
    \end{align*}

    \noindent To analyze the left-hand side, we first write $c = c_+ - c_-$ with $c_+ \coloneqq |C_1| + 2|C_2| + |C_4|$ and $c_- \coloneqq |C_5|$.
    Since we have $c_+, c_- \in [2 |N(S)|]_0$, it follows that $|c| \leq 2 |N(S)|$.
    In order to proceed, we use the following rather specific and technical lemma, which we will prove later.

    \begin{restatable}{lemma}{lemTechnicalLemma}
\label{lem:alternatingflips}
    For odd $d \in \N_{\geq 0}$, we have $\flips\left(\sum_{j \in [(d - 1)/2]_0} 2^{2j + 1} \right) = d$ and for any $x \in \Z_{\neq 0}$ we have $\flips\left(\sum_{j \in [(d - 1)/2]_0} 2^{2j + 1} + x\right) \geq d - \lfloor \log(|x|) \rfloor - 4$.
\end{restatable}

    \noindent Applied to the left-hand side of Equation (\ref{eq:flipeq}), this yields
    \[
        \flips\left(\sum_{j \in [ (d - 1)/2]_0} 2^{2j+1} + c\right) \geq d - \lfloor \log(2 |N(S)|) \rfloor - 4 \geq d - \log(|N(S)|) - 5.
    \]
    Combined with the bound on the number of flips of the right-hand side of the equation, this implies $2 |N(S)| + \log(|N(S)|) \geq d - 5$.
    Therefore, we have
    \[
        |N(S)| \geq \frac{d - 1}{2} - \frac12 \log \left( \frac{d - 5}{2} \right) - 2
    \]
    for each set $S \subseteq V$ with $|S| = -1 + \sum_{j \in [(d - 1)/2]_0} 2^{2j + 1}$.
\end{proof}
\begin{proof}[Proof of Lemma \ref{lem:alternatingflips}]
    In binary notation, we have $\sum_{j \in [(d - 1)/2]_0} 2^{2j + 1} = \underbrace{1010 \dots 10}_{d + 1 \text{ digits}}$, which has $d$ flips.
    We will show that adding a number with small absolute value cannot decrease the number of flips by much, since most digits will remain alternating ones and zeros.
    For $x \neq 0$, $|x|$ has $y \coloneqq \lfloor \log(|x|) \rfloor + 1$ digits.
    If $y \geq d - 3$, we have $d - \lfloor \log(|x|)\rfloor - 4 \leq 0$, so the second inequality of the lemma holds trivially.
    Therefore, we may now assume $y < d - 3$.
    
    First, consider the case that $y$ is even and $x \geq 0$.
    The smallest number greater than $\underbrace{1010 \dots 10}_{d + 1 \text{ digits}}$ with a different digit in the $(3 + y)$-th place is $\underbrace{1010 \dots 1}_{d - y - 2 \text{ digits}} \underbrace{100 \dots 0}_{y + 3 \text{ digits}}$.
    Thus, in order to change any of the first $d - y - 1$ digits, we would need
    \[
        |x| \geq \underbrace{1010 \dots 1}_{d - y - 2 \text{ digits}} \underbrace{100 \dots 0}_{y + 3 \text{ digits}} - \underbrace{1010 \dots 10}_{d + 1 \text{ digits}} = \underbrace{1010 \dots 10 110}_{y + 1 \text{ digits}},
    \]
    which contradicts the fact that $|x|$ has $y$ digits.
    Hence, the first $(d - y - 1)$ digits have to remain alternating ones and zeros.
    This also implies that there has to appear a one in the last $2 + y$ digits, since there is no way to remove it without altering the $(3 + y)$-th digit.
    Therefore we have at least $d - y - 1 = d - \lfloor \log(|x|) \rfloor - 2$ flips.
    
    Next, consider the case that $y$ is even and $x < 0$.
    The greatest number smaller than $\underbrace{1010 \dots 10}_{d + 1 \text{ digits}}$ with a different digit in the $(3 + y)$-th place is $\underbrace{1010 \dots 10}_{d - y - 3 \text{ digits}} \underbrace{011 \dots 1}_{y + 4 \text{ digits}}$.
    Thus, in order to change any of the first $d - y - 1$ digits, we would need
    \[
        |x| \geq \underbrace{1010 \dots 10}_{d + 1 \text{ digits}} - \underbrace{1010 \dots 10}_{d - y - 3 \text{ digits}} \underbrace{011 \dots 1}_{y + 4 \text{ digits}}
        = \underbrace{10 \dots 1011}_{y + 2 \text{ digits}},
    \]
    which contradicts the fact that $|x|$ has $y$ digits.
    Hence, the first $(d - y - 1)$ digits have to remain alternating ones and zeros.
    Moreover, the last $y + 2$ digits cannot all be zero, since that would imply $x = - \underbrace{10 \dots 10}_{y + 2 \text{ digits}}$, again contradicting the number of digits of $|x|$.
    Therefore we have at least $d - y - 1 = d - \lfloor \log(|x|) \rfloor - 2$ flips.
    
    Now consider the case that $y$ is odd and $x \geq 0$.
    The smallest number greater than $\underbrace{1010 \dots 10}_{d + 1 \text{ digits}}$ with a different digit in the $(3 + y)$-th place is $\underbrace{1010 \dots 1}_{d - y - 3 \text{ digits}} \underbrace{100 \dots 0}_{y + 4 \text{ digits}}$.
    Thus, in order to change any of the first $d - y - 1$ digits, we would need
    \[
        |x| \geq \underbrace{1010 \dots 1}_{d - y - 3 \text{ digits}} \underbrace{100 \dots 0}_{y + 4 \text{ digits}} - \underbrace{1010 \dots 10}_{d + 1 \text{ digits}} = \underbrace{10 \dots 10110}_{y + 2 \text{ digits}},
    \]
    which contradicts the fact that $|x|$ has $y$ digits.
    Hence, the first $d - y - 1$ digits remain alternating ones and zeros.
    This also implies that there has to appear a one in the last $2 + y$ digits, since there is no way to remove it without altering the $(3 + y)$-th digit.
    Therefore we have at least $d - y - 1 = d - \lfloor \log(|x|) \rfloor - 2$ flips.

    Finally, consider the case that $y$ is odd and $x < 0$.
    The greatest number smaller than $\underbrace{1010 \dots 10}_{d + 1 \text{ digits}}$ with a different digit in the $(3 + y)$-th place is $\underbrace{1010 \dots 10}_{d - y - 2 \text{ digits}} \underbrace{011 \dots 1}_{y + 3 \text{ digits}}$.
    Thus, in order to change any of the first $d - y - 1$ digits, we would need
    \[
        |x| \geq \underbrace{1010 \dots 10}_{d + 1 \text{ digits}} - \underbrace{1010 \dots 10}_{d - y - 2 \text{ digits}} \underbrace{011 \dots 1}_{y + 3 \text{ digits}}
        = \underbrace{10 \dots 1011}_{y + 1 \text{ digits}},
    \]
    which contradicts the fact that $|x|$ has $y$ digits.
    Hence, the first $(d - y - 1)$ digits have to remain alternating ones and zeros.
    Moreover, the last $y + 2$ digits cannot all be zero, since that would imply $x = - \underbrace{10 \dots 10}_{y + 1 \text{ digits}}$, again contradicting the number of digits of $|x|$.
    Therefore we have at least $d - y - 1 = d - \lfloor \log(|x|) \rfloor - 2$ flips.
\end{proof}

\subsection{Omitted Proofs of Section \ref{ch:npHardness}: NP-Hardness}\label{app:npHardness}

In order to prove Proposition \ref{prop:timegadget} we first show that we can restrict our analysis to strategies where $K \cap F_{t} \neq \emptyset \Leftrightarrow K \subseteq F_t$ for certain $K \subseteq G$. With this lemma at hand, we show that there exists a $T$-winning $m$-strategy for $G$ if and only if there exists a $T$-winning $2m$-strategy for $\mathbb{G}$.

\lemCliques*
\begin{proof}
    We claim that $S' \coloneqq (F'_1, \dots, F'_T)$ with
    \[
        F'_i \coloneqq \begin{cases}
            F_i &\text{ if } K \subseteq F_i, \\
            F_i \setminus K &\text{ otherwise}
        \end{cases}
    \]
    for all $i \in [T]$ is a winning $m$-strategy on $G$.
    In particular, we claim that we have $B_i = B'_i$ for all $i \in [T]$, where the $B_i$ denote the burning sets appearing during strategy $S$ and the $B'_i$ denote the burning sets appearing during strategy $S'$.

    To prove this, we assume that this statement does not hold, and let $t \in [T]$ be the smallest possible such that $B_t \neq B'_t$.
    Note that due to our choice of $S'$, this implies $K \cap B'_i = K \cap B_i \in \{\emptyset, K\}$ for all $i < t$.
    From the choice of $t$, it follows that $F_t = F'_t \cup X$ for some $X \subsetneq K$, which implies $\tilde{B}_t = \tilde{B}'_t \setminus X$.
    If $K \cap B'_{t - 1} = \emptyset$, we also have $K \cap \tilde{B}'_t = \emptyset$ and therefore $\tilde{B}_t = \tilde{B}'_t$, contradicting $B_t \neq B'_t$.
    If $K \cap B'_{t - 1} = K$, then there exists a node $v \in \tilde{B}_t \cap \tilde{B}'_t$, which implies $K \cup N(K) \subseteq B_t \cap B'_t$.
    Since $\tilde{B}_t$ and $\tilde{B}'_t$ only differ in a subset of $K$, $B_t$ and $B'_t$ can only differ in a subset of $K \cup N(K)$.
    Therefore, we get $B_t = B'_t$, contradicting the definition of $t$.
\end{proof}

\begin{lemma}[$\mathbb{G}$]
\label{lem:shadowG}
    Let $G=(V,E)$ be an arbitrary graph. There exists a $T$-winning $m$-strategy for $G$ if and only if there exists a $T$-winning $2m$-strategy for $\mathbb{G}$.
\end{lemma}
\begin{proof}
    Let $S = (F_1, \dots, F_T)$ be a winning $m$-strategy on $G$. Then the $2m$-strategy $S' = (F_1', \ldots, F'_T)$ with $F'_i = F_i \cup \{v' \in \mathbb{G}: v \in F_i\}$ is a winning $2m$-strategy on $\mathbb{G}$.

    Conversely, let $S' = (F_1',\ldots, F_T')$ be a winning $2m$-strategy for $\mathbb{G}$. Due to Lemma~\ref{lem:cliques}, we can assume that if $F_i'$ contains a node $v$ (resp. $v'$) of $\mathbb{G}$, then it also contains $v'$ (resp. $v$). Therefore, $S = (F_1, \ldots, F_T)$ with $F_i = V(G) \cap F_i'$ is a winning $m$-strategy for $G$.
\end{proof}

\propTimeGadget*
\begin{proof}
    If there is a winning $m$-strategy $S'=(F'_1,\ldots,F'_T)$ for $G$, then the following strategy $F$ is a winning $4m$-strategy for $H(G,T)$. Hence, $\ffn(H(G,T))\leq 4m$.
    For a visualization of this strategy, see Figures \ref{fig:strategyStart}-\ref{fig:strategyEnd}.
    Recall that the blocks $A,B,X,Y,Z$ are $2m$-cliques and every $P_i^j$ is an $m$-clique.
    An edge between two blocks $\mathcal{B}_1$ and $\mathcal{B}_2$ in this image corresponds to connecting every node in $\mathcal{B}_1$ to every node in $\mathcal{B}_2$.
    Furthermore, in these figures, burning nodes are colored in red, extinguished nodes are colored in white and the current firefighter set is marked with thick borders around the respective nodes.
    The strategy $F$ is given by:
    \begin{enumerate}
        \item $F_{1+ T \cdot i +j } = \{A, B, P_{i + 1}^j, P_{i + 1}^{j+1} \}$ for $i \in [2T + 1]_0, j \in [T]_0$\\
        with $E_{2T^2+2T+1}=\mathcal{P}$ (see Figure \ref{fig:strategyStart})
        \item $F_{2 T^2 + 2T + 2} = \{A, B, X\}$\\
        with $E_{2T^2+2T+2}=A \cup \mathcal{P}$ (see Figure \ref{fig:strategy2})
        \item $F_{2 T^2 + 2T + 3} = \{X, Y\}$\\
        with $E_{2T^2+2T+3}=A \cup \left(\mathcal{P} \setminus  \bigcup_{i \in [2T + 2]} P_i^{T+1}\right)$ (see Figure \ref{fig:strategy3})
        \item $F_{2 T^2 + 2T + 3 + i} = \{Y, F'_i\}$ for $i \in [T]$\\
        with $E_{2 T^2 + 3T + 3}=\mathbb{G} \cup A \cup X$ (see Figure \ref{fig:strategy4})
        \item $F_{2 T^2 + 3T + 4} = \{Y, Z\}$\\
        with $E_{2 T^2 + 3T + 4}=\mathbb{G} \cup X \cup Y$ (see Figure \ref{fig:strategy5})
        \item $F_{2 T^2 + 3T + 5} = \{A, B, Z\}$\\
        with $E_{2 T^2 + 3T + 5}=\mathbb{G} \cup X \cup Y \cup Z$ (see Figure \ref{fig:strategy6})
        \item $F_{2 T^2 + 3T + 5 + T \cdot i +  j} = \{A, B, P_{i + 1}^j, P_{i + 1}^{j + 1} \}$ for $i \in [2T + 1]_0$, $j \in [T]_0$\\
        with $E_{4 T^2 + 5T + 5}=H(G,T)$ (see Figure \ref{fig:strategyEnd})
    \end{enumerate}

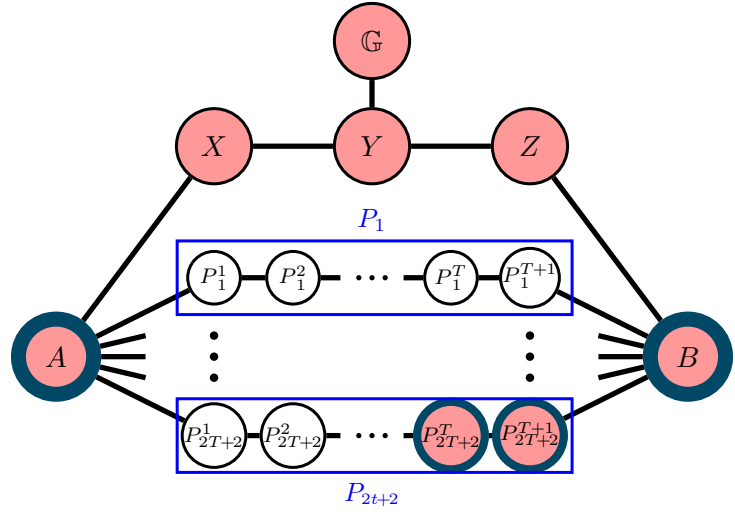
\begin{figure}[H]
\centering
\begin{tikzpicture}[scale=0.7, every node/.style={font=\scriptsize}]

\node[draw=dkteal, line width=2mm, circle, minimum size=1.0cm, inner sep=0pt,fill=red!40] (A) at (-6,0) {\normalsize $A$};
\node[draw=dkteal, line width=2mm, circle, minimum size=1.0cm, inner sep=0pt,fill=red!40] (B) at (6,0) {\normalsize $B$};
        \node[draw, line width=0.4mm, circle, minimum size=1.0cm, inner sep=0pt,fill=red!40] (X) at (-3,4) {\normalsize$X$};
        \node[draw, line width=0.4mm, circle, minimum size=1.0cm, inner sep=0pt,fill=red!40] (Y) at (0,4) {\normalsize$Y$};
        \node[draw, line width=0.4mm, circle, minimum size=1.0cm, inner sep=0pt,fill=red!40] (Z) at (3,4) {\normalsize$Z$};
        \node[draw, line width=0.4mm, circle, minimum size=1.0cm, inner sep=0pt,fill=red!40] (G) at (0,6) {\normalsize$\mathbb{G}$};
        \draw[line width=0.7mm, shorten >=-0.0cm, shorten <=-0.0cm] (A) -- (X);
        \draw[line width=0.7mm, shorten >=-0.0cm, shorten <=-0.0cm] (X) -- (Y);
        \draw[line width=0.7mm, shorten >=-0.0cm, shorten <=-0.0cm] (Y) -- (Z);
        \draw[line width=0.7mm, shorten >=-0.0cm, shorten <=-0.0cm] (Z) -- (B);
        \draw[line width=0.7mm, shorten >=-0.0cm, shorten <=-0.0cm] (Y) -- (G);
        % Paths P_1 - P_m
        % Path P_1
        \node[draw, line width=0.4mm, circle, minimum size=0.7cm, inner sep=0pt,fill=white] (P11) at (-3,1.5) {$P\mathrlap{}^{1}_{1}$};
        \node[draw, line width=0.4mm, circle, minimum size=0.7cm, inner sep=0pt,fill=white] (P12) at (-1.5,1.5) {$P\mathrlap{}^{2}_{1}$};
        
        \node[draw, line width=0.4mm, circle, minimum size=0.7cm, inner sep=0pt,fill=white] (P1t) at (1.5,1.5) {$P\mathrlap{}^{T}_{1}$};
        \node[draw, line width=0.4mm, circle, minimum size=0.7cm, inner sep=0pt,fill=white] (P1t+1) at (3,1.5) {$P\mathrlap{}^{\text{\tiny \smash{$T\mkern-5mu+\mkern-5mu1$}}}_{1}$};
        \draw[line width=0.7mm, shorten >=-0.0cm, shorten <=-0.0cm] (A) -- (P11);
        \draw[line width=0.7mm, shorten >=-0.0cm, shorten <=-0.0cm] (P11) -- (P12);
        \draw[line width=0.7mm, shorten >=-0.0cm, shorten <=-0.0cm] (P12) -- (-0.6,1.5);
        \fill (-0.25,1.5) circle (0.05cm);
        \fill (0,1.5) circle (0.05cm);
        \fill (0.25,1.5) circle (0.05cm);
        \draw[line width=0.7mm, shorten >=-0.0cm, shorten <=-0.0cm] (0.6,1.5) -- (P1t);
        \draw[line width=0.7mm, shorten >=-0.0cm, shorten <=-0.0cm] (P1t) -- (P1t+1);
        \draw[line width=0.7mm, shorten >=-0.0cm, shorten <=-0.0cm] (P1t+1) -- (B);
        \draw[draw=\pathcolor, line width=0.4mm] (-3.7,2.2) rectangle (3.8,0.8);

        \node at (0, 2.6) {\textcolor{\pathcolor}{\small$P_1$}};
        %Dots in between Paths
        \fill (-3,0.4) circle (0.08cm);
        \fill (-3,0) circle (0.08cm);
        \fill (-3,-0.4) circle (0.08cm);
        \draw[line width=0.7mm, shorten >=-0.0cm, shorten <=-0.0cm] (A) -- (-4.3,0.4);
        \draw[line width=0.7mm, shorten >=-0.0cm, shorten <=-0.0cm] (A) -- (-4.3,0.0);
        \draw[line width=0.7mm, shorten >=-0.0cm, shorten <=-0.0cm] (A) -- (-4.3,-0.4);
        
        \fill (3,0.4) circle (0.08cm);
        \fill (3,0) circle (0.08cm);
        \fill (3,-0.4) circle (0.08cm);

        \draw[line width=0.7mm, shorten >=-0.0cm, shorten <=-0.0cm] (B) -- (4.3,0.4);
        \draw[line width=0.7mm, shorten >=-0.0cm, shorten <=-0.0cm] (B) -- (4.3,0.0);
        \draw[line width=0.7mm, shorten >=-0.0cm, shorten <=-0.0cm] (B) -- (4.3,-0.4);

        %Path P_m
        \node[draw, line width=0.4mm, circle, minimum size=0.7cm, inner sep=0pt,fill=white] (Pm1) at (-3,-1.5) {$P^{1}_{\text{\tiny \smash{$2T\mkern-6mu+\mkern-6mu2$}}}$};
        \node[draw, line width=0.4mm, circle, minimum size=0.7cm, inner sep=0pt,fill=white] (Pm2) at (-1.5,-1.5) {$P^{2}_{\text{\tiny \smash{$2T\mkern-6mu+\mkern-6mu2$}}}$};

        \node[draw=dkteal, line width=1.1mm, circle, minimum size=0.9cm, inner sep=0pt,fill=red!40] (Pmt) at (1.5,-1.5) {$P^{T}_{\text{\tiny \smash{$2T\mkern-6mu+\mkern-6mu2$}}}$};
        \node[draw=dkteal, line width=1.1mm, circle, minimum size=0.9cm, inner sep=0pt,fill=red!40] (Pmt+1) at (3,-1.5) {$P^{\text{\tiny \smash{$T\mkern-5mu+\mkern-5mu1$}}}_{\text{\tiny \smash{$2T\mkern-6mu+\mkern-6mu2$}}}$};
        \draw[line width=0.7mm, shorten >=-0.0cm, shorten <=-0.0cm] (A) -- (Pm1);
        \draw[line width=0.7mm, shorten >=-0.0cm, shorten <=-0.0cm] (Pm1) -- (Pm2);
        \draw[line width=0.7mm, shorten >=-0.0cm, shorten <=-0.0cm] (Pm2) -- (-0.6,-1.5);
        \fill (-0.25,-1.5) circle (0.05cm);
        \fill (0,-1.5) circle (0.05cm);
        \fill (0.25,-1.5) circle (0.05cm);
        \draw[line width=0.7mm, shorten >=-0.0cm, shorten <=-0.0cm] (0.6,-1.5) -- (Pmt);
        \draw[line width=0.7mm, shorten >=-0.0cm, shorten <=-0.0cm] (Pmt) -- (Pmt+1);
        \draw[line width=0.7mm, shorten >=-0.0cm, shorten <=-0.0cm] (Pmt+1) -- (B);
        \draw[draw=\pathcolor, line width=0.4mm] (-3.7,-2.2) rectangle (3.8,-0.8);

        \node at (0, -2.6) {\textcolor{\pathcolor}{\small$P_{\text{\tiny \smash{$2t\mkern-6mu+\mkern-6mu2$}}}$}};
\end{tikzpicture}

\caption{Step 1 of a winning $4m$-strategy for $H(G, T)$ if $\mathbb{G}$ is $2m$-winning in time $T$.
Start by covering $A$ and $B$ in each time step, clearing the paths $P_i$ one by one. The last time step of this process is visualized here.}
\label{fig:strategyStart}
\end{figure}

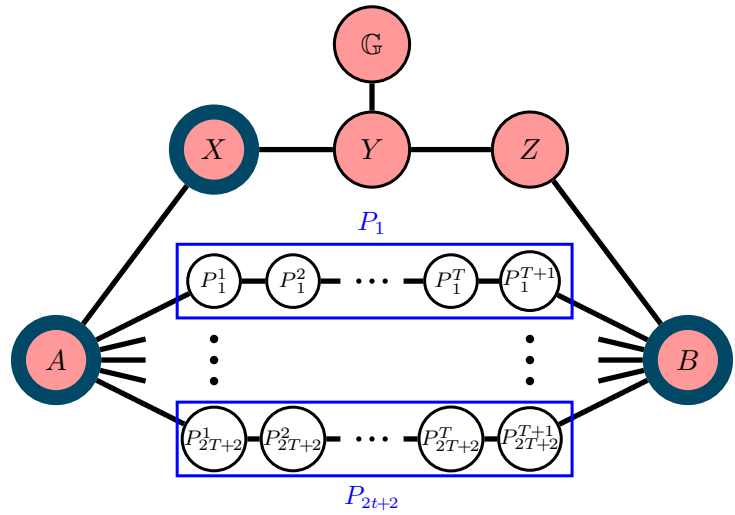
\begin{figure}[H]
\centering
\begin{tikzpicture}[scale=0.7, every node/.style={font=\scriptsize}]

\node[draw=dkteal, line width=2mm, circle, minimum size=1.0cm, inner sep=0pt,fill=red!40] (A) at (-6,0) {\normalsize $A$};
\node[draw=dkteal, line width=2mm, circle, minimum size=1.0cm, inner sep=0pt,fill=red!40] (B) at (6,0) {\normalsize $B$};
\node[draw=dkteal, line width=2mm, circle, minimum size=1.0cm, inner sep=0pt,fill=red!40] (X) at (-3,4) {\normalsize$X$};
        \node[draw, line width=0.4mm, circle, minimum size=1.0cm, inner sep=0pt,fill=red!40] (Y) at (0,4) {\normalsize$Y$};
        \node[draw, line width=0.4mm, circle, minimum size=1.0cm, inner sep=0pt,fill=red!40] (Z) at (3,4) {\normalsize$Z$};
        \node[draw, line width=0.4mm, circle, minimum size=1.0cm, inner sep=0pt,fill=red!40] (G) at (0,6) {\normalsize$\mathbb{G}$};
        \draw[line width=0.7mm, shorten >=-0.0cm, shorten <=-0.0cm] (A) -- (X);
        \draw[line width=0.7mm, shorten >=-0.0cm, shorten <=-0.0cm] (X) -- (Y);
        \draw[line width=0.7mm, shorten >=-0.0cm, shorten <=-0.0cm] (Y) -- (Z);
        \draw[line width=0.7mm, shorten >=-0.0cm, shorten <=-0.0cm] (Z) -- (B);
        \draw[line width=0.7mm, shorten >=-0.0cm, shorten <=-0.0cm] (Y) -- (G);
        % Paths P_1 - P_m
        % Path P_1
        \node[draw, line width=0.4mm, circle, minimum size=0.7cm, inner sep=0pt,fill=white] (P11) at (-3,1.5) {$P\mathrlap{}^{1}_{1}$};
        \node[draw, line width=0.4mm, circle, minimum size=0.7cm, inner sep=0pt,fill=white] (P12) at (-1.5,1.5) {$P\mathrlap{}^{2}_{1}$};
        
        \node[draw, line width=0.4mm, circle, minimum size=0.7cm, inner sep=0pt,fill=white] (P1t) at (1.5,1.5) {$P\mathrlap{}^{T}_{1}$};
        \node[draw, line width=0.4mm, circle, minimum size=0.7cm, inner sep=0pt,fill=white] (P1t+1) at (3,1.5) {$P\mathrlap{}^{\text{\tiny \smash{$T\mkern-5mu+\mkern-5mu1$}}}_{1}$};
        \draw[line width=0.7mm, shorten >=-0.0cm, shorten <=-0.0cm] (A) -- (P11);
        \draw[line width=0.7mm, shorten >=-0.0cm, shorten <=-0.0cm] (P11) -- (P12);
        \draw[line width=0.7mm, shorten >=-0.0cm, shorten <=-0.0cm] (P12) -- (-0.6,1.5);
        \fill (-0.25,1.5) circle (0.05cm);
        \fill (0,1.5) circle (0.05cm);
        \fill (0.25,1.5) circle (0.05cm);
        \draw[line width=0.7mm, shorten >=-0.0cm, shorten <=-0.0cm] (0.6,1.5) -- (P1t);
        \draw[line width=0.7mm, shorten >=-0.0cm, shorten <=-0.0cm] (P1t) -- (P1t+1);
        \draw[line width=0.7mm, shorten >=-0.0cm, shorten <=-0.0cm] (P1t+1) -- (B);
        \draw[draw=\pathcolor, line width=0.4mm] (-3.7,2.2) rectangle (3.8,0.8);

        \node at (0, 2.6) {\textcolor{\pathcolor}{\small$P_1$}};
        %Dots in between Paths
        \fill (-3,0.4) circle (0.08cm);
        \fill (-3,0) circle (0.08cm);
        \fill (-3,-0.4) circle (0.08cm);
        \draw[line width=0.7mm, shorten >=-0.0cm, shorten <=-0.0cm] (A) -- (-4.3,0.4);
        \draw[line width=0.7mm, shorten >=-0.0cm, shorten <=-0.0cm] (A) -- (-4.3,0.0);
        \draw[line width=0.7mm, shorten >=-0.0cm, shorten <=-0.0cm] (A) -- (-4.3,-0.4);
        
        \fill (3,0.4) circle (0.08cm);
        \fill (3,0) circle (0.08cm);
        \fill (3,-0.4) circle (0.08cm);

        \draw[line width=0.7mm, shorten >=-0.0cm, shorten <=-0.0cm] (B) -- (4.3,0.4);
        \draw[line width=0.7mm, shorten >=-0.0cm, shorten <=-0.0cm] (B) -- (4.3,0.0);
        \draw[line width=0.7mm, shorten >=-0.0cm, shorten <=-0.0cm] (B) -- (4.3,-0.4);

        %Path P_m
        \node[draw, line width=0.4mm, circle, minimum size=0.7cm, inner sep=0pt,fill=white] (Pm1) at (-3,-1.5) {$P^{1}_{\text{\tiny \smash{$2T\mkern-6mu+\mkern-6mu2$}}}$};
        \node[draw, line width=0.4mm, circle, minimum size=0.7cm, inner sep=0pt,fill=white] (Pm2) at (-1.5,-1.5) {$P^{2}_{\text{\tiny \smash{$2T\mkern-6mu+\mkern-6mu2$}}}$};

        \node[draw, line width=0.4mm, circle, minimum size=0.7cm, inner sep=0pt,fill=white] (Pmt) at (1.5,-1.5) {$P^{T}_{\text{\tiny \smash{$2T\mkern-6mu+\mkern-6mu2$}}}$};
        \node[draw, line width=0.4mm, circle, minimum size=0.7cm, inner sep=0pt,fill=white] (Pmt+1) at (3,-1.5) {$P^{\text{\tiny \smash{$T\mkern-5mu+\mkern-5mu1$}}}_{\text{\tiny \smash{$2T\mkern-6mu+\mkern-6mu2$}}}$};
        \draw[line width=0.7mm, shorten >=-0.0cm, shorten <=-0.0cm] (A) -- (Pm1);
        \draw[line width=0.7mm, shorten >=-0.0cm, shorten <=-0.0cm] (Pm1) -- (Pm2);
        \draw[line width=0.7mm, shorten >=-0.0cm, shorten <=-0.0cm] (Pm2) -- (-0.6,-1.5);
        \fill (-0.25,-1.5) circle (0.05cm);
        \fill (0,-1.5) circle (0.05cm);
        \fill (0.25,-1.5) circle (0.05cm);
        \draw[line width=0.7mm, shorten >=-0.0cm, shorten <=-0.0cm] (0.6,-1.5) -- (Pmt);
        \draw[line width=0.7mm, shorten >=-0.0cm, shorten <=-0.0cm] (Pmt) -- (Pmt+1);
        \draw[line width=0.7mm, shorten >=-0.0cm, shorten <=-0.0cm] (Pmt+1) -- (B);
        \draw[draw=\pathcolor, line width=0.4mm] (-3.7,-2.2) rectangle (3.8,-0.8);

        \node at (0, -2.6) {\textcolor{\pathcolor}{\small$P_{\text{\tiny \smash{$2t\mkern-6mu+\mkern-6mu2$}}}$}};
\end{tikzpicture}
\caption{Step 2: Covering $A,B$ and $X$.}
\label{fig:strategy2}
\end{figure}

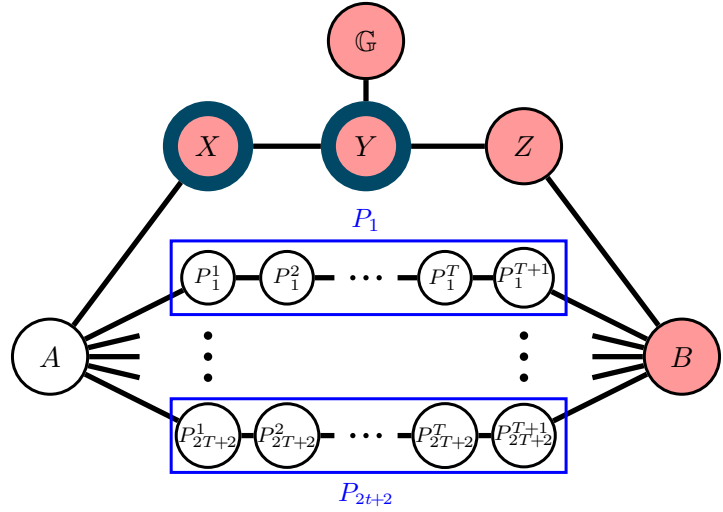
\begin{figure}[H]
\centering
\begin{tikzpicture}[scale=0.7, every node/.style={font=\scriptsize}]

    \node[draw, line width=0.4mm, circle, minimum size=1.0cm, inner sep=0pt,fill=white] (A) at (-6,0) {\normalsize $A$};
    \node[draw, line width=0.4mm, circle, minimum size=1.0cm, inner sep=0pt,fill=red!40] (B) at (6,0) {\normalsize $B$};
\node[draw=dkteal, line width=2mm, circle, minimum size=1.0cm, inner sep=0pt,fill=red!40] (X) at (-3,4) {\normalsize$X$};
\node[draw=dkteal, line width=2mm, circle, minimum size=1.0cm, inner sep=0pt,fill=red!40] (Y) at (0,4) {\normalsize$Y$};
        \node[draw, line width=0.4mm, circle, minimum size=1.0cm, inner sep=0pt,fill=red!40] (Z) at (3,4) {\normalsize$Z$};
        \node[draw, line width=0.4mm, circle, minimum size=1.0cm, inner sep=0pt,fill=red!40] (G) at (0,6) {\normalsize$\mathbb{G}$};
        \draw[line width=0.7mm, shorten >=-0.0cm, shorten <=-0.0cm] (A) -- (X);
        \draw[line width=0.7mm, shorten >=-0.0cm, shorten <=-0.0cm] (X) -- (Y);
        \draw[line width=0.7mm, shorten >=-0.0cm, shorten <=-0.0cm] (Y) -- (Z);
        \draw[line width=0.7mm, shorten >=-0.0cm, shorten <=-0.0cm] (Z) -- (B);
        \draw[line width=0.7mm, shorten >=-0.0cm, shorten <=-0.0cm] (Y) -- (G);
        % Paths P_1 - P_m
        % Path P_1
        \node[draw, line width=0.4mm, circle, minimum size=0.7cm, inner sep=0pt,fill=white] (P11) at (-3,1.5) {$P\mathrlap{}^{1}_{1}$};
        \node[draw, line width=0.4mm, circle, minimum size=0.7cm, inner sep=0pt,fill=white] (P12) at (-1.5,1.5) {$P\mathrlap{}^{2}_{1}$};
        
        \node[draw, line width=0.4mm, circle, minimum size=0.7cm, inner sep=0pt,fill=white] (P1t) at (1.5,1.5) {$P\mathrlap{}^{T}_{1}$};
        \node[draw, line width=0.4mm, circle, minimum size=0.7cm, inner sep=0pt,fill=white] (P1t+1) at (3,1.5) {$P\mathrlap{}^{\text{\tiny \smash{$T\mkern-5mu+\mkern-5mu1$}}}_{1}$};
        \draw[line width=0.7mm, shorten >=-0.0cm, shorten <=-0.0cm] (A) -- (P11);
        \draw[line width=0.7mm, shorten >=-0.0cm, shorten <=-0.0cm] (P11) -- (P12);
        \draw[line width=0.7mm, shorten >=-0.0cm, shorten <=-0.0cm] (P12) -- (-0.6,1.5);
        \fill (-0.25,1.5) circle (0.05cm);
        \fill (0,1.5) circle (0.05cm);
        \fill (0.25,1.5) circle (0.05cm);
        \draw[line width=0.7mm, shorten >=-0.0cm, shorten <=-0.0cm] (0.6,1.5) -- (P1t);
        \draw[line width=0.7mm, shorten >=-0.0cm, shorten <=-0.0cm] (P1t) -- (P1t+1);
        \draw[line width=0.7mm, shorten >=-0.0cm, shorten <=-0.0cm] (P1t+1) -- (B);
        \draw[draw=\pathcolor, line width=0.4mm] (-3.7,2.2) rectangle (3.8,0.8);

        \node at (0, 2.6) {\textcolor{\pathcolor}{\small$P_1$}};
        %Dots in between Paths
        \fill (-3,0.4) circle (0.08cm);
        \fill (-3,0) circle (0.08cm);
        \fill (-3,-0.4) circle (0.08cm);
        \draw[line width=0.7mm, shorten >=-0.0cm, shorten <=-0.0cm] (A) -- (-4.3,0.4);
        \draw[line width=0.7mm, shorten >=-0.0cm, shorten <=-0.0cm] (A) -- (-4.3,0.0);
        \draw[line width=0.7mm, shorten >=-0.0cm, shorten <=-0.0cm] (A) -- (-4.3,-0.4);
        
        \fill (3,0.4) circle (0.08cm);
        \fill (3,0) circle (0.08cm);
        \fill (3,-0.4) circle (0.08cm);

        \draw[line width=0.7mm, shorten >=-0.0cm, shorten <=-0.0cm] (B) -- (4.3,0.4);
        \draw[line width=0.7mm, shorten >=-0.0cm, shorten <=-0.0cm] (B) -- (4.3,0.0);
        \draw[line width=0.7mm, shorten >=-0.0cm, shorten <=-0.0cm] (B) -- (4.3,-0.4);

        %Path P_m
        \node[draw, line width=0.4mm, circle, minimum size=0.7cm, inner sep=0pt,fill=white] (Pm1) at (-3,-1.5) {$P^{1}_{\text{\tiny \smash{$2T\mkern-6mu+\mkern-6mu2$}}}$};
        \node[draw, line width=0.4mm, circle, minimum size=0.7cm, inner sep=0pt,fill=white] (Pm2) at (-1.5,-1.5) {$P^{2}_{\text{\tiny \smash{$2T\mkern-6mu+\mkern-6mu2$}}}$};

        \node[draw, line width=0.4mm, circle, minimum size=0.7cm, inner sep=0pt,fill=white] (Pmt) at (1.5,-1.5) {$P^{T}_{\text{\tiny \smash{$2T\mkern-6mu+\mkern-6mu2$}}}$};
        \node[draw, line width=0.4mm, circle, minimum size=0.7cm, inner sep=0pt,fill=white] (Pmt+1) at (3,-1.5) {$P^{\text{\tiny \smash{$T\mkern-5mu+\mkern-5mu1$}}}_{\text{\tiny \smash{$2T\mkern-6mu+\mkern-6mu2$}}}$};
        \draw[line width=0.7mm, shorten >=-0.0cm, shorten <=-0.0cm] (A) -- (Pm1);
        \draw[line width=0.7mm, shorten >=-0.0cm, shorten <=-0.0cm] (Pm1) -- (Pm2);
        \draw[line width=0.7mm, shorten >=-0.0cm, shorten <=-0.0cm] (Pm2) -- (-0.6,-1.5);
        \fill (-0.25,-1.5) circle (0.05cm);
        \fill (0,-1.5) circle (0.05cm);
        \fill (0.25,-1.5) circle (0.05cm);
        \draw[line width=0.7mm, shorten >=-0.0cm, shorten <=-0.0cm] (0.6,-1.5) -- (Pmt);
        \draw[line width=0.7mm, shorten >=-0.0cm, shorten <=-0.0cm] (Pmt) -- (Pmt+1);
        \draw[line width=0.7mm, shorten >=-0.0cm, shorten <=-0.0cm] (Pmt+1) -- (B);
        \draw[draw=\pathcolor, line width=0.4mm] (-3.7,-2.2) rectangle (3.8,-0.8);

        \node at (0, -2.6) {\textcolor{\pathcolor}{\small$P_{\text{\tiny \smash{$2t\mkern-6mu+\mkern-6mu2$}}}$}};
\end{tikzpicture}
\caption{Step 3: Covering $X$ and $Y$.}
\label{fig:strategy3}
\end{figure}

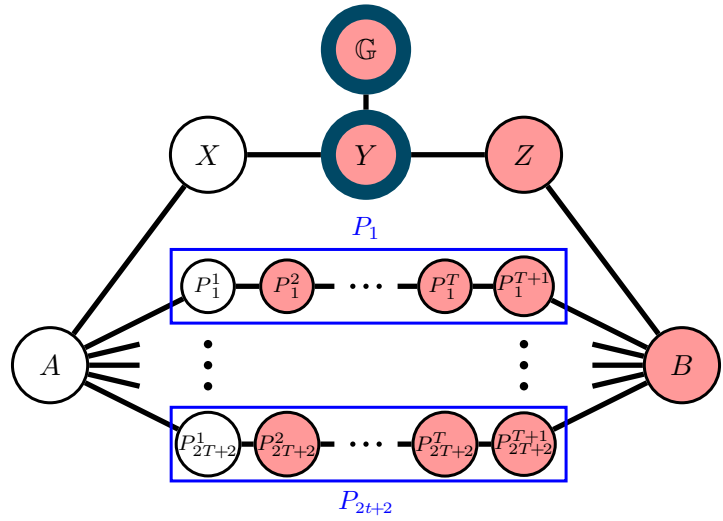
\begin{figure}[H]
\centering
\begin{tikzpicture}[scale=0.7, every node/.style={font=\scriptsize}]

    \node[draw, line width=0.4mm, circle, minimum size=1.0cm, inner sep=0pt,fill=white] (A) at (-6,0) {\normalsize $A$};
    \node[draw, line width=0.4mm, circle, minimum size=1.0cm, inner sep=0pt,fill=red!40] (B) at (6,0) {\normalsize $B$};
    \node[draw, line width=0.4mm, circle, minimum size=1.0cm, inner sep=0pt,fill=white] (X) at (-3,4) {\normalsize$X$};
\node[draw=dkteal, line width=2mm, circle, minimum size=1.0cm, inner sep=0pt,fill=red!40] (Y) at (0,4) {\normalsize$Y$};
        \node[draw, line width=0.4mm, circle, minimum size=1.0cm, inner sep=0pt,fill=red!40] (Z) at (3,4) {\normalsize$Z$};
\node[draw=dkteal, line width=2mm, circle, minimum size=1.0cm, inner sep=0pt,fill=red!40] (G) at (0,6) {\normalsize$\mathbb{G}$};
        \draw[line width=0.7mm, shorten >=-0.0cm, shorten <=-0.0cm] (A) -- (X);
        \draw[line width=0.7mm, shorten >=-0.0cm, shorten <=-0.0cm] (X) -- (Y);
        \draw[line width=0.7mm, shorten >=-0.0cm, shorten <=-0.0cm] (Y) -- (Z);
        \draw[line width=0.7mm, shorten >=-0.0cm, shorten <=-0.0cm] (Z) -- (B);
        \draw[line width=0.7mm, shorten >=-0.0cm, shorten <=-0.0cm] (Y) -- (G);
        % Paths P_1 - P_m
        % Path P_1
        \node[draw, line width=0.4mm, circle, minimum size=0.7cm, inner sep=0pt,fill=white] (P11) at (-3,1.5) {$P\mathrlap{}^{1}_{1}$};
        \node[draw, line width=0.4mm, circle, minimum size=0.7cm, inner sep=0pt,fill=red!40] (P12) at (-1.5,1.5) {$P\mathrlap{}^{2}_{1}$};
        
        \node[draw, line width=0.4mm, circle, minimum size=0.7cm, inner sep=0pt,fill=red!40] (P1t) at (1.5,1.5) {$P\mathrlap{}^{T}_{1}$};
        \node[draw, line width=0.4mm, circle, minimum size=0.7cm, inner sep=0pt,fill=red!40] (P1t+1) at (3,1.5) {$P\mathrlap{}^{\text{\tiny \smash{$T\mkern-5mu+\mkern-5mu1$}}}_{1}$};
        \draw[line width=0.7mm, shorten >=-0.0cm, shorten <=-0.0cm] (A) -- (P11);
        \draw[line width=0.7mm, shorten >=-0.0cm, shorten <=-0.0cm] (P11) -- (P12);
        \draw[line width=0.7mm, shorten >=-0.0cm, shorten <=-0.0cm] (P12) -- (-0.6,1.5);
        \fill (-0.25,1.5) circle (0.05cm);
        \fill (0,1.5) circle (0.05cm);
        \fill (0.25,1.5) circle (0.05cm);
        \draw[line width=0.7mm, shorten >=-0.0cm, shorten <=-0.0cm] (0.6,1.5) -- (P1t);
        \draw[line width=0.7mm, shorten >=-0.0cm, shorten <=-0.0cm] (P1t) -- (P1t+1);
        \draw[line width=0.7mm, shorten >=-0.0cm, shorten <=-0.0cm] (P1t+1) -- (B);
        \draw[draw=\pathcolor, line width=0.4mm] (-3.7,2.2) rectangle (3.8,0.8);

        \node at (0, 2.6) {\textcolor{\pathcolor}{\small$P_1$}};
        %Dots in between Paths
        \fill (-3,0.4) circle (0.08cm);
        \fill (-3,0) circle (0.08cm);
        \fill (-3,-0.4) circle (0.08cm);
        \draw[line width=0.7mm, shorten >=-0.0cm, shorten <=-0.0cm] (A) -- (-4.3,0.4);
        \draw[line width=0.7mm, shorten >=-0.0cm, shorten <=-0.0cm] (A) -- (-4.3,0.0);
        \draw[line width=0.7mm, shorten >=-0.0cm, shorten <=-0.0cm] (A) -- (-4.3,-0.4);
        
        \fill (3,0.4) circle (0.08cm);
        \fill (3,0) circle (0.08cm);
        \fill (3,-0.4) circle (0.08cm);

        \draw[line width=0.7mm, shorten >=-0.0cm, shorten <=-0.0cm] (B) -- (4.3,0.4);
        \draw[line width=0.7mm, shorten >=-0.0cm, shorten <=-0.0cm] (B) -- (4.3,0.0);
        \draw[line width=0.7mm, shorten >=-0.0cm, shorten <=-0.0cm] (B) -- (4.3,-0.4);

        %Path P_m
        \node[draw, line width=0.4mm, circle, minimum size=0.7cm, inner sep=0pt,fill=white] (Pm1) at (-3,-1.5) {$P^{1}_{\text{\tiny \smash{$2T\mkern-6mu+\mkern-6mu2$}}}$};
        \node[draw, line width=0.4mm, circle, minimum size=0.7cm, inner sep=0pt,fill=red!40] (Pm2) at (-1.5,-1.5) {$P^{2}_{\text{\tiny \smash{$2T\mkern-6mu+\mkern-6mu2$}}}$};

        \node[draw, line width=0.4mm, circle, minimum size=0.7cm, inner sep=0pt,fill=red!40] (Pmt) at (1.5,-1.5) {$P^{T}_{\text{\tiny \smash{$2T\mkern-6mu+\mkern-6mu2$}}}$};
        \node[draw, line width=0.4mm, circle, minimum size=0.7cm, inner sep=0pt,fill=red!40] (Pmt+1) at (3,-1.5) {$P^{\text{\tiny \smash{$T\mkern-5mu+\mkern-5mu1$}}}_{\text{\tiny \smash{$2T\mkern-6mu+\mkern-6mu2$}}}$};
        \draw[line width=0.7mm, shorten >=-0.0cm, shorten <=-0.0cm] (A) -- (Pm1);
        \draw[line width=0.7mm, shorten >=-0.0cm, shorten <=-0.0cm] (Pm1) -- (Pm2);
        \draw[line width=0.7mm, shorten >=-0.0cm, shorten <=-0.0cm] (Pm2) -- (-0.6,-1.5);
        \fill (-0.25,-1.5) circle (0.05cm);
        \fill (0,-1.5) circle (0.05cm);
        \fill (0.25,-1.5) circle (0.05cm);
        \draw[line width=0.7mm, shorten >=-0.0cm, shorten <=-0.0cm] (0.6,-1.5) -- (Pmt);
        \draw[line width=0.7mm, shorten >=-0.0cm, shorten <=-0.0cm] (Pmt) -- (Pmt+1);
        \draw[line width=0.7mm, shorten >=-0.0cm, shorten <=-0.0cm] (Pmt+1) -- (B);
        \draw[draw=\pathcolor, line width=0.4mm] (-3.7,-2.2) rectangle (3.8,-0.8);

        \node at (0, -2.6) {\textcolor{\pathcolor}{\small$P_{\text{\tiny \smash{$2t\mkern-6mu+\mkern-6mu2$}}}$}};
\end{tikzpicture}
\caption{Step 4: Covering $Y$ for $T$ timesteps, while clearing $\mathbb{G}$ with the remaining $2m$ firefighters.
In the meantime, the fire is spreading along the paths. The last of the $T$ time steps is visualized.}
\label{fig:strategy4}
\end{figure}

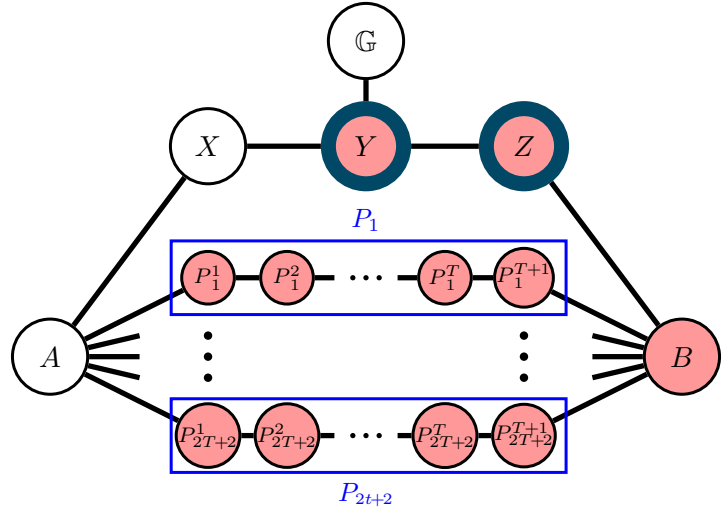
\begin{figure}[H]
\centering
\begin{tikzpicture}[scale=0.7, every node/.style={font=\scriptsize}]

    \node[draw, line width=0.4mm, circle, minimum size=1.0cm, inner sep=0pt,fill=white] (A) at (-6,0) {\normalsize $A$};
    \node[draw, line width=0.4mm, circle, minimum size=1.0cm, inner sep=0pt,fill=red!40] (B) at (6,0) {\normalsize $B$};
    \node[draw, line width=0.4mm, circle, minimum size=1.0cm, inner sep=0pt,fill=white] (X) at (-3,4) {\normalsize$X$};
\node[draw=dkteal, line width=2mm, circle, minimum size=1.0cm, inner sep=0pt,fill=red!40] (Y) at (0,4) {\normalsize$Y$};
\node[draw=dkteal, line width=2mm, circle, minimum size=1.0cm, inner sep=0pt,fill=red!40] (Z) at (3,4) {\normalsize$Z$};
        \node[draw, line width=0.4mm, circle, minimum size=1.0cm, inner sep=0pt,fill=white] (G) at (0,6) {\normalsize$\mathbb{G}$};
        \draw[line width=0.7mm, shorten >=-0.0cm, shorten <=-0.0cm] (A) -- (X);
        \draw[line width=0.7mm, shorten >=-0.0cm, shorten <=-0.0cm] (X) -- (Y);
        \draw[line width=0.7mm, shorten >=-0.0cm, shorten <=-0.0cm] (Y) -- (Z);
        \draw[line width=0.7mm, shorten >=-0.0cm, shorten <=-0.0cm] (Z) -- (B);
        \draw[line width=0.7mm, shorten >=-0.0cm, shorten <=-0.0cm] (Y) -- (G);
        % Paths P_1 - P_m
        % Path P_1
        \node[draw, line width=0.4mm, circle, minimum size=0.7cm, inner sep=0pt,fill=red!40] (P11) at (-3,1.5) {$P\mathrlap{}^{1}_{1}$};
        \node[draw, line width=0.4mm, circle, minimum size=0.7cm, inner sep=0pt,fill=red!40] (P12) at (-1.5,1.5) {$P\mathrlap{}^{2}_{1}$};
        
        \node[draw, line width=0.4mm, circle, minimum size=0.7cm, inner sep=0pt,fill=red!40] (P1t) at (1.5,1.5) {$P\mathrlap{}^{T}_{1}$};
        \node[draw, line width=0.4mm, circle, minimum size=0.7cm, inner sep=0pt,fill=red!40] (P1t+1) at (3,1.5) {$P\mathrlap{}^{\text{\tiny \smash{$T\mkern-5mu+\mkern-5mu1$}}}_{1}$};
        \draw[line width=0.7mm, shorten >=-0.0cm, shorten <=-0.0cm] (A) -- (P11);
        \draw[line width=0.7mm, shorten >=-0.0cm, shorten <=-0.0cm] (P11) -- (P12);
        \draw[line width=0.7mm, shorten >=-0.0cm, shorten <=-0.0cm] (P12) -- (-0.6,1.5);
        \fill (-0.25,1.5) circle (0.05cm);
        \fill (0,1.5) circle (0.05cm);
        \fill (0.25,1.5) circle (0.05cm);
        \draw[line width=0.7mm, shorten >=-0.0cm, shorten <=-0.0cm] (0.6,1.5) -- (P1t);
        \draw[line width=0.7mm, shorten >=-0.0cm, shorten <=-0.0cm] (P1t) -- (P1t+1);
        \draw[line width=0.7mm, shorten >=-0.0cm, shorten <=-0.0cm] (P1t+1) -- (B);
        \draw[draw=\pathcolor, line width=0.4mm] (-3.7,2.2) rectangle (3.8,0.8);

        \node at (0, 2.6) {\textcolor{\pathcolor}{\small$P_1$}};
        %Dots in between Paths
        \fill (-3,0.4) circle (0.08cm);
        \fill (-3,0) circle (0.08cm);
        \fill (-3,-0.4) circle (0.08cm);
        \draw[line width=0.7mm, shorten >=-0.0cm, shorten <=-0.0cm] (A) -- (-4.3,0.4);
        \draw[line width=0.7mm, shorten >=-0.0cm, shorten <=-0.0cm] (A) -- (-4.3,0.0);
        \draw[line width=0.7mm, shorten >=-0.0cm, shorten <=-0.0cm] (A) -- (-4.3,-0.4);
        
        \fill (3,0.4) circle (0.08cm);
        \fill (3,0) circle (0.08cm);
        \fill (3,-0.4) circle (0.08cm);

        \draw[line width=0.7mm, shorten >=-0.0cm, shorten <=-0.0cm] (B) -- (4.3,0.4);
        \draw[line width=0.7mm, shorten >=-0.0cm, shorten <=-0.0cm] (B) -- (4.3,0.0);
        \draw[line width=0.7mm, shorten >=-0.0cm, shorten <=-0.0cm] (B) -- (4.3,-0.4);

        %Path P_m
        \node[draw, line width=0.4mm, circle, minimum size=0.7cm, inner sep=0pt,fill=red!40] (Pm1) at (-3,-1.5) {$P^{1}_{\text{\tiny \smash{$2T\mkern-6mu+\mkern-6mu2$}}}$};
        \node[draw, line width=0.4mm, circle, minimum size=0.7cm, inner sep=0pt,fill=red!40] (Pm2) at (-1.5,-1.5) {$P^{2}_{\text{\tiny \smash{$2T\mkern-6mu+\mkern-6mu2$}}}$};

        \node[draw, line width=0.4mm, circle, minimum size=0.7cm, inner sep=0pt,fill=red!40] (Pmt) at (1.5,-1.5) {$P^{T}_{\text{\tiny \smash{$2T\mkern-6mu+\mkern-6mu2$}}}$};
        \node[draw, line width=0.4mm, circle, minimum size=0.7cm, inner sep=0pt,fill=red!40] (Pmt+1) at (3,-1.5) {$P^{\text{\tiny \smash{$T\mkern-5mu+\mkern-5mu1$}}}_{\text{\tiny \smash{$2T\mkern-6mu+\mkern-6mu2$}}}$};
        \draw[line width=0.7mm, shorten >=-0.0cm, shorten <=-0.0cm] (A) -- (Pm1);
        \draw[line width=0.7mm, shorten >=-0.0cm, shorten <=-0.0cm] (Pm1) -- (Pm2);
        \draw[line width=0.7mm, shorten >=-0.0cm, shorten <=-0.0cm] (Pm2) -- (-0.6,-1.5);
        \fill (-0.25,-1.5) circle (0.05cm);
        \fill (0,-1.5) circle (0.05cm);
        \fill (0.25,-1.5) circle (0.05cm);
        \draw[line width=0.7mm, shorten >=-0.0cm, shorten <=-0.0cm] (0.6,-1.5) -- (Pmt);
        \draw[line width=0.7mm, shorten >=-0.0cm, shorten <=-0.0cm] (Pmt) -- (Pmt+1);
        \draw[line width=0.7mm, shorten >=-0.0cm, shorten <=-0.0cm] (Pmt+1) -- (B);
        \draw[draw=\pathcolor, line width=0.4mm] (-3.7,-2.2) rectangle (3.8,-0.8);

        \node at (0, -2.6) {\textcolor{\pathcolor}{\small$P_{\text{\tiny \smash{$2t\mkern-6mu+\mkern-6mu2$}}}$}};
\end{tikzpicture}
\caption{Step 5: Covering $Y$ and $Z$.}
\label{fig:strategy5}
\end{figure}

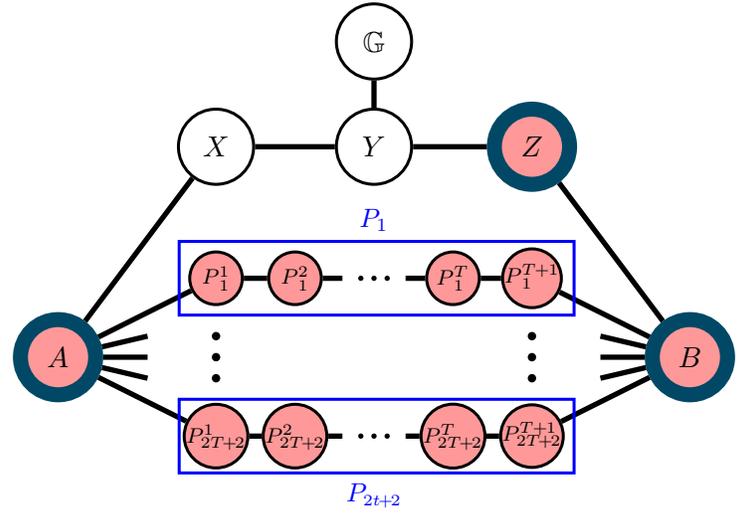
\begin{figure}[H]
\centering
\begin{tikzpicture}[scale=0.7, every node/.style={font=\scriptsize}]

\node[draw=dkteal, line width=2mm, circle, minimum size=1.0cm, inner sep=0pt,fill=red!40] (A) at (-6,0) {\normalsize $A$};
\node[draw=dkteal, line width=2mm, circle, minimum size=1.0cm, inner sep=0pt,fill=red!40] (B) at (6,0) {\normalsize $B$};
    \node[draw, line width=0.4mm, circle, minimum size=1.0cm, inner sep=0pt,fill=white] (X) at (-3,4) {\normalsize$X$};
    \node[draw, line width=0.4mm, circle, minimum size=1.0cm, inner sep=0pt,fill=white] (Y) at (0,4) {\normalsize$Y$};
\node[draw=dkteal, line width=2mm, circle, minimum size=1.0cm, inner sep=0pt,fill=red!40] (Z) at (3,4) {\normalsize$Z$};
        \node[draw, line width=0.4mm, circle, minimum size=1.0cm, inner sep=0pt,fill=white] (G) at (0,6) {\normalsize$\mathbb{G}$};
        \draw[line width=0.7mm, shorten >=-0.0cm, shorten <=-0.0cm] (A) -- (X);
        \draw[line width=0.7mm, shorten >=-0.0cm, shorten <=-0.0cm] (X) -- (Y);
        \draw[line width=0.7mm, shorten >=-0.0cm, shorten <=-0.0cm] (Y) -- (Z);
        \draw[line width=0.7mm, shorten >=-0.0cm, shorten <=-0.0cm] (Z) -- (B);
        \draw[line width=0.7mm, shorten >=-0.0cm, shorten <=-0.0cm] (Y) -- (G);
        % Paths P_1 - P_m
        % Path P_1
        \node[draw, line width=0.4mm, circle, minimum size=0.7cm, inner sep=0pt,fill=red!40] (P11) at (-3,1.5) {$P\mathrlap{}^{1}_{1}$};
        \node[draw, line width=0.4mm, circle, minimum size=0.7cm, inner sep=0pt,fill=red!40] (P12) at (-1.5,1.5) {$P\mathrlap{}^{2}_{1}$};
        
        \node[draw, line width=0.4mm, circle, minimum size=0.7cm, inner sep=0pt,fill=red!40] (P1t) at (1.5,1.5) {$P\mathrlap{}^{T}_{1}$};
        \node[draw, line width=0.4mm, circle, minimum size=0.7cm, inner sep=0pt,fill=red!40] (P1t+1) at (3,1.5) {$P\mathrlap{}^{\text{\tiny \smash{$T\mkern-5mu+\mkern-5mu1$}}}_{1}$};
        \draw[line width=0.7mm, shorten >=-0.0cm, shorten <=-0.0cm] (A) -- (P11);
        \draw[line width=0.7mm, shorten >=-0.0cm, shorten <=-0.0cm] (P11) -- (P12);
        \draw[line width=0.7mm, shorten >=-0.0cm, shorten <=-0.0cm] (P12) -- (-0.6,1.5);
        \fill (-0.25,1.5) circle (0.05cm);
        \fill (0,1.5) circle (0.05cm);
        \fill (0.25,1.5) circle (0.05cm);
        \draw[line width=0.7mm, shorten >=-0.0cm, shorten <=-0.0cm] (0.6,1.5) -- (P1t);
        \draw[line width=0.7mm, shorten >=-0.0cm, shorten <=-0.0cm] (P1t) -- (P1t+1);
        \draw[line width=0.7mm, shorten >=-0.0cm, shorten <=-0.0cm] (P1t+1) -- (B);
        \draw[draw=\pathcolor, line width=0.4mm] (-3.7,2.2) rectangle (3.8,0.8);

        \node at (0, 2.6) {\textcolor{\pathcolor}{\small$P_1$}};
        %Dots in between Paths
        \fill (-3,0.4) circle (0.08cm);
        \fill (-3,0) circle (0.08cm);
        \fill (-3,-0.4) circle (0.08cm);
        \draw[line width=0.7mm, shorten >=-0.0cm, shorten <=-0.0cm] (A) -- (-4.3,0.4);
        \draw[line width=0.7mm, shorten >=-0.0cm, shorten <=-0.0cm] (A) -- (-4.3,0.0);
        \draw[line width=0.7mm, shorten >=-0.0cm, shorten <=-0.0cm] (A) -- (-4.3,-0.4);
        
        \fill (3,0.4) circle (0.08cm);
        \fill (3,0) circle (0.08cm);
        \fill (3,-0.4) circle (0.08cm);

        \draw[line width=0.7mm, shorten >=-0.0cm, shorten <=-0.0cm] (B) -- (4.3,0.4);
        \draw[line width=0.7mm, shorten >=-0.0cm, shorten <=-0.0cm] (B) -- (4.3,0.0);
        \draw[line width=0.7mm, shorten >=-0.0cm, shorten <=-0.0cm] (B) -- (4.3,-0.4);

        %Path P_m
        \node[draw, line width=0.4mm, circle, minimum size=0.7cm, inner sep=0pt,fill=red!40] (Pm1) at (-3,-1.5) {$P^{1}_{\text{\tiny \smash{$2T\mkern-6mu+\mkern-6mu2$}}}$};
        \node[draw, line width=0.4mm, circle, minimum size=0.7cm, inner sep=0pt,fill=red!40] (Pm2) at (-1.5,-1.5) {$P^{2}_{\text{\tiny \smash{$2T\mkern-6mu+\mkern-6mu2$}}}$};

        \node[draw, line width=0.4mm, circle, minimum size=0.7cm, inner sep=0pt,fill=red!40] (Pmt) at (1.5,-1.5) {$P^{T}_{\text{\tiny \smash{$2T\mkern-6mu+\mkern-6mu2$}}}$};
        \node[draw, line width=0.4mm, circle, minimum size=0.7cm, inner sep=0pt,fill=red!40] (Pmt+1) at (3,-1.5) {$P^{\text{\tiny \smash{$T\mkern-5mu+\mkern-5mu1$}}}_{\text{\tiny \smash{$2T\mkern-6mu+\mkern-6mu2$}}}$};
        \draw[line width=0.7mm, shorten >=-0.0cm, shorten <=-0.0cm] (A) -- (Pm1);
        \draw[line width=0.7mm, shorten >=-0.0cm, shorten <=-0.0cm] (Pm1) -- (Pm2);
        \draw[line width=0.7mm, shorten >=-0.0cm, shorten <=-0.0cm] (Pm2) -- (-0.6,-1.5);
        \fill (-0.25,-1.5) circle (0.05cm);
        \fill (0,-1.5) circle (0.05cm);
        \fill (0.25,-1.5) circle (0.05cm);
        \draw[line width=0.7mm, shorten >=-0.0cm, shorten <=-0.0cm] (0.6,-1.5) -- (Pmt);
        \draw[line width=0.7mm, shorten >=-0.0cm, shorten <=-0.0cm] (Pmt) -- (Pmt+1);
        \draw[line width=0.7mm, shorten >=-0.0cm, shorten <=-0.0cm] (Pmt+1) -- (B);
        \draw[draw=\pathcolor, line width=0.4mm] (-3.7,-2.2) rectangle (3.8,-0.8);

        \node at (0, -2.6) {\textcolor{\pathcolor}{\small$P_{\text{\tiny \smash{$2t\mkern-6mu+\mkern-6mu2$}}}$}};
\end{tikzpicture}
\caption{Step 6: Covering $A,B$ and $Z$ right before the fire can spread again back to $X$.}
\label{fig:strategy6}
\end{figure}

\begin{figure}[H]
\centering
\begin{tikzpicture}[scale=0.7, every node/.style={font=\scriptsize}]

\node[draw=dkteal, line width=2mm, circle, minimum size=1.0cm, inner sep=0pt,fill=white] (A) at (-6,0) {\normalsize $A$};
\node[draw=dkteal, line width=2mm, circle, minimum size=1.0cm, inner sep=0pt,fill=red!40] (B) at (6,0) {\normalsize $B$};
    \node[draw, line width=0.4mm, circle, minimum size=1.0cm, inner sep=0pt,fill=white] (X) at (-3,4) {\normalsize$X$};
    \node[draw, line width=0.4mm, circle, minimum size=1.0cm, inner sep=0pt,fill=white] (Y) at (0,4) {\normalsize$Y$};
    \node[draw, line width=0.4mm, circle, minimum size=1.0cm, inner sep=0pt,fill=white] (Z) at (3,4) {\normalsize$Z$};
        \node[draw, line width=0.4mm, circle, minimum size=1.0cm, inner sep=0pt,fill=white] (G) at (0,6) {\normalsize$\mathbb{G}$};
        \draw[line width=0.7mm, shorten >=-0.0cm, shorten <=-0.0cm] (A) -- (X);
        \draw[line width=0.7mm, shorten >=-0.0cm, shorten <=-0.0cm] (X) -- (Y);
        \draw[line width=0.7mm, shorten >=-0.0cm, shorten <=-0.0cm] (Y) -- (Z);
        \draw[line width=0.7mm, shorten >=-0.0cm, shorten <=-0.0cm] (Z) -- (B);
        \draw[line width=0.7mm, shorten >=-0.0cm, shorten <=-0.0cm] (Y) -- (G);
        % Paths P_1 - P_m
        % Path P_1
        \node[draw, line width=0.4mm, circle, minimum size=0.7cm, inner sep=0pt,fill=white] (P11) at (-3,1.5) {$P\mathrlap{}^{1}_{1}$};
        \node[draw, line width=0.4mm, circle, minimum size=0.7cm, inner sep=0pt,fill=white] (P12) at (-1.5,1.5) {$P\mathrlap{}^{2}_{1}$};
        
        \node[draw, line width=0.4mm, circle, minimum size=0.7cm, inner sep=0pt,fill=white] (P1t) at (1.5,1.5) {$P\mathrlap{}^{T}_{1}$};
        \node[draw, line width=0.4mm, circle, minimum size=0.7cm, inner sep=0pt,fill=white] (P1t+1) at (3,1.5) {$P\mathrlap{}^{\text{\tiny \smash{$T\mkern-5mu+\mkern-5mu1$}}}_{1}$};
        \draw[line width=0.7mm, shorten >=-0.0cm, shorten <=-0.0cm] (A) -- (P11);
        \draw[line width=0.7mm, shorten >=-0.0cm, shorten <=-0.0cm] (P11) -- (P12);
        \draw[line width=0.7mm, shorten >=-0.0cm, shorten <=-0.0cm] (P12) -- (-0.6,1.5);
        \fill (-0.25,1.5) circle (0.05cm);
        \fill (0,1.5) circle (0.05cm);
        \fill (0.25,1.5) circle (0.05cm);
        \draw[line width=0.7mm, shorten >=-0.0cm, shorten <=-0.0cm] (0.6,1.5) -- (P1t);
        \draw[line width=0.7mm, shorten >=-0.0cm, shorten <=-0.0cm] (P1t) -- (P1t+1);
        \draw[line width=0.7mm, shorten >=-0.0cm, shorten <=-0.0cm] (P1t+1) -- (B);
        \draw[draw=\pathcolor, line width=0.4mm] (-3.7,2.2) rectangle (3.8,0.8);

        \node at (0, 2.6) {\textcolor{\pathcolor}{\small$P_1$}};
        %Dots in between Paths
        \fill (-3,0.4) circle (0.08cm);
        \fill (-3,0) circle (0.08cm);
        \fill (-3,-0.4) circle (0.08cm);
        \draw[line width=0.7mm, shorten >=-0.0cm, shorten <=-0.0cm] (A) -- (-4.3,0.4);
        \draw[line width=0.7mm, shorten >=-0.0cm, shorten <=-0.0cm] (A) -- (-4.3,0.0);
        \draw[line width=0.7mm, shorten >=-0.0cm, shorten <=-0.0cm] (A) -- (-4.3,-0.4);
        
        \fill (3,0.4) circle (0.08cm);
        \fill (3,0) circle (0.08cm);
        \fill (3,-0.4) circle (0.08cm);

        \draw[line width=0.7mm, shorten >=-0.0cm, shorten <=-0.0cm] (B) -- (4.3,0.4);
        \draw[line width=0.7mm, shorten >=-0.0cm, shorten <=-0.0cm] (B) -- (4.3,0.0);
        \draw[line width=0.7mm, shorten >=-0.0cm, shorten <=-0.0cm] (B) -- (4.3,-0.4);

        %Path P_m
        \node[draw, line width=0.4mm, circle, minimum size=0.7cm, inner sep=0pt,fill=white] (Pm1) at (-3,-1.5) {$P^{1}_{\text{\tiny \smash{$2T\mkern-6mu+\mkern-6mu2$}}}$};
        \node[draw, line width=0.4mm, circle, minimum size=0.7cm, inner sep=0pt,fill=white] (Pm2) at (-1.5,-1.5) {$P^{2}_{\text{\tiny \smash{$2T\mkern-6mu+\mkern-6mu2$}}}$};

        \node[draw=dkteal, line width=1.1mm, circle, minimum size=0.9cm, inner sep=0pt,fill=red!40] (Pmt) at (1.5,-1.5) {$P^{T}_{\text{\tiny \smash{$2T\mkern-6mu+\mkern-6mu2$}}}$};
        \node[draw=dkteal, line width=1.1mm, circle, minimum size=0.9cm, inner sep=0pt,fill=red!40] (Pmt+1) at (3,-1.5) {$P^{\text{\tiny \smash{$T\mkern-5mu+\mkern-5mu1$}}}_{\text{\tiny \smash{$2T\mkern-6mu+\mkern-6mu2$}}}$};
        \draw[line width=0.7mm, shorten >=-0.0cm, shorten <=-0.0cm] (A) -- (Pm1);
        \draw[line width=0.7mm, shorten >=-0.0cm, shorten <=-0.0cm] (Pm1) -- (Pm2);
        \draw[line width=0.7mm, shorten >=-0.0cm, shorten <=-0.0cm] (Pm2) -- (-0.6,-1.5);
        \fill (-0.25,-1.5) circle (0.05cm);
        \fill (0,-1.5) circle (0.05cm);
        \fill (0.25,-1.5) circle (0.05cm);
        \draw[line width=0.7mm, shorten >=-0.0cm, shorten <=-0.0cm] (0.6,-1.5) -- (Pmt);
        \draw[line width=0.7mm, shorten >=-0.0cm, shorten <=-0.0cm] (Pmt) -- (Pmt+1);
        \draw[line width=0.7mm, shorten >=-0.0cm, shorten <=-0.0cm] (Pmt+1) -- (B);
        \draw[draw=\pathcolor, line width=0.4mm] (-3.7,-2.2) rectangle (3.8,-0.8);

        \node at (0, -2.6) {\textcolor{\pathcolor}{\small$P_{\text{\tiny \smash{$2t\mkern-6mu+\mkern-6mu2$}}}$}};
\end{tikzpicture}
\caption{Step 7: Finally covering $A$ and $B$ while clearing all paths again. The last time step is visualized, after which all nodes are extinguished.}
\label{fig:strategyEnd}
\end{figure}
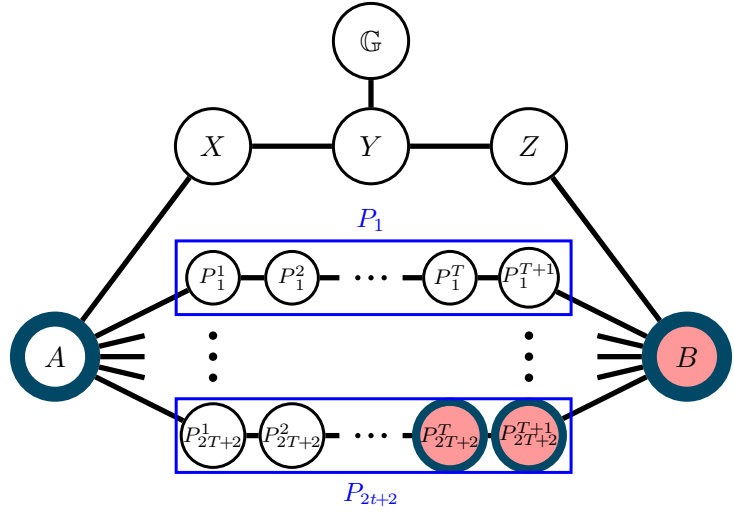

Next, we assume that there is no $m$-strategy for $G$ that wins in time $T$.
    We consider the following node subsets:
    $\Omega_1 = H(G,T) \setminus ( A \cup \mathcal{P} )$, \;
    $\Omega_2 = H(G,T) \setminus (\mathbb{G} \cup \mathcal{P})$, \;
    $\Omega_3 = \mathbb{G} \cup B \cup Y \cup Z \cup \bigcup_{i \in [2T + 2]} P_i^{T+1}$, \;
    $\Omega_4 = B \cup Y \cup Z \cup P_k \cup P_\ell \cup \{v\}$, \;
    $\Omega_5 = A \cup B \cup Y \cup Z \cup P_k \cup P_\ell$, \; $\Omega_6 = A \cup B \cup X \cup Z \cup P_k \cup P_\ell$, and $\Omega_7 = A \cup B \cup Y \cup Z \cup P_k \cup (P_\ell \setminus P_\ell^{1}) \cup \{v\}$ where $v$ is any node from $ \mathbb{G}$ and $k,\ell\in [2T+2]$ with $k\neq \ell$.
    We call a subset of burning nodes $\Omega_n$-blocked, if it contains $\Omega_n$ or one of its symmetric variants, regarding the following symmetries: Switching $A$ and $B$, $X$ and $Z$ as well as $P_i^j$ with $P_i^{T + 2 - j}$ for all $i \in [2T + 2], j \in [T + 1]$ (i.e., mirroring the graph as shown in Figure \ref{fig:TimeGadget} horizontally), switching the complete paths $\{P_1, \ldots, P_{2T + 2}\}$ according to any permutation, or replacing $v$ by any other node in $\mathbb{G}$.

    We prove in the following that for any $4m$-strategy and any subset of burning nodes that is $\Omega_n$-blocked for some $n \in [7]$, after finitely many steps, the subset of burning nodes will be $\Omega_{n'}$-blocked for some $n' \in [7]$ as visualized in Figure \ref{fig:stateGraph}. Since the initial state of a fully burning graph is $\Omega_1$-blocked, this means that there is no winning $4m$-strategy for $H(G,T)$, since the empty set is not $\Omega_{n}$-blocked for any $n \in [7]$.

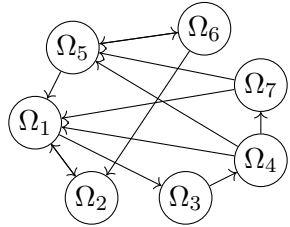
\begin{figure}[htb]
    \centering
	\begin{tikzpicture}[scale=0.5]
            % Draw nodes for states
            \node[draw, circle, minimum size=0.7cm, inner sep=0pt] (s1) at (1,1) {$\Omega_1$};
            \node[draw, circle, minimum size=0.7cm, inner sep=0pt] (s2) at (2.5,-1) {$\Omega_2$};
            \node[draw, circle, minimum size=0.7cm, inner sep=0pt] (s3) at (5,-1) {$\Omega_3$};
            \node[draw, circle, minimum size=0.7cm, inner sep=0pt] (s4) at (7,0) {$\Omega_4$};
            \node[draw, circle, minimum size=0.7cm, inner sep=0pt] (s5) at (2,3) {$\Omega_5$};
            \node[draw, circle, minimum size=0.7cm, inner sep=0pt] (s6) at (5.5,3.5) {$\Omega_6$};
            \node[draw, circle, minimum size=0.7cm, inner sep=0pt] (s7) at (7,2) {$\Omega_7$};
            % \draw[->] (s1) edge[out= -90,in= 180,looseness=5] (s1);
            \draw[->] (s1) to (s2);
            \draw[->] (s1) to (s3);

            \draw[->] (s2) to (s1);

            \draw[->] (s3) to (s4);

            \draw[->] (s4) to (s1);
            \draw[->] (s4) to (s5);
            \draw[->] (s4) to (s7);

            \draw[->] (s5) to (s1);
            \draw[->] (s5) to (s6);

            \draw[->] (s6) to (s2);
            \draw[->] (s6) to (s5);
            
            \draw[->] (s7) to (s1);
            \draw[->] (s7) to (s5);
            \end{tikzpicture}
    \caption{For any subset of burning nodes that is $\Omega_n$-blocked for some $n \in [7]$, the subset of burning nodes will be $\Omega_{n'}$-blocked after finitely many steps, for some $n' \in [7]$ with an edge $(\Omega_{n},\Omega_{n'})$ or $n' = n$.}
    \label{fig:stateGraph}
\end{figure}

    \noindent\textbf{Case 1:} Let $B_t$ be $\Omega_1$-blocked with $\Omega_1 = H(G,T) \setminus ( A \cup \mathcal{P})$.
    Since $Y$ (resp. $Z$) has at least $4m$ (resp. $3m$) burning neighbours in addition to its own size of $2m$ (resp. $2m$), $E_{t + 1}$ cannot contain a node from $Y$ (resp. $Z$).

    \textit{Case 1.1:} $E_{t + 1}$ contains a node $v$ from $\mathbb{G}$.
    Then $Y \cup \{v\} \subseteq F_{t+1}$. The only helpful thing one can do with the remaining $2m-1$ firefighters is to guard $B$. Thus, $B_{t + 1} \supseteq H(G,T) \setminus \mathcal{P}$ and hence is $\Omega_2$-blocked.

    \textit{Case 1.2:} $E_{t + 1}$ contains a node from $B$.
    Then $B \cup Z \subseteq F_{t+1}$. The only useful thing to do with the remaining $m$ firefighters is to guard $A$. Thus, $B_{t + 1} \supseteq H(G,T) \setminus ( B \cup \mathcal{P})$ and hence is $\Omega_1$-blocked.
    
    \textit{Case 1.3:} $E_{t + 1}$ contains a node from $X$.
    Then $F_{t+1} = X \cup Y$. Thus $B_{t + 1} \supseteq \mathbb{G} \cup B \cup Y \cup Z \cup \bigcup_{i \in [2T + 2]} P_i^{T+1}$ and hence is $\Omega_3$-blocked.

    \medskip
    \noindent\textbf{Case 2:} Let $B_t$ be $\Omega_2$-blocked with $\Omega_2 = H(G,T) \setminus (\mathbb{G} \cup \mathcal{P})$.
    Since $Y$ (resp. $X,Z$) has at least $4m$ (resp. $3m$) burning neighbours in addition to its own size of $2m$ (resp. $2m$), $E_{t + 1}$ cannot contain a node from $Y$ (resp. $X,Z$).

    If $E_{t + 1}$ contains a node from $A$ (resp. $B$), then $A \cup X \subseteq F_{t+1}$ (resp. $B \cup Z \subseteq F_{t+1}$). The only helpful thing to do with the $m$ remaining firefighters is to guard $B$ (resp. $A$). Thus $B_{t + 1} \supseteq H(G,T) \setminus ( B \cup \mathcal{P} )$ (resp. $H(G,T) \setminus ( A \cup \mathcal{P} )$) and hence is $\Omega_1$-blocked.

    \medskip
    \noindent\textbf{Case 3:} Let $B_t$ be $\Omega_3$-blocked with $\Omega_3 = \mathbb{G} \cup B \cup Y \cup Z \cup \bigcup_{i \in [2T + 2]} P_i^{T+1}$.
    If $Y \cap F_{t+1} = \emptyset$, we have $B_{t + 1} \supseteq \mathbb{G} \cup B \cup X \cup Y \cup Z$, which is $\Omega_1$-blocked.
    Otherwise, let $t' \in \{2, \ldots, T\}$ be the first time such that $Y \cap F_{t + t'}$ (if it exists). Then there were at most $2m$ firefighters used in $\mathbb{G}$ in $F_{t+1}, \ldots, F_{t + t' - 1}$, so by Lemma \ref{lem:shadowG}, we have $B_{t + t' - 1} \cap \mathbb{G} \neq \emptyset$. Furthermore, we have $B_{t + t' - 2} \cap \mathbb{G} \neq \emptyset$, which implies $Y \subseteq B_{t + t' - 1}$. Thus, we have $B_{t + t'} \supseteq \mathbb{G} \cup X \cup Y \cup Z$.
    Since there are $2T + 2$ paths connecting $A$ and $B$ and we need to use at least $m$ firefighters to influence the spreading of the fire along one path, there are at least $2$ paths $P_k, P_\ell$ such that $(P_k \cup P_\ell) \cap F_{\tilde{t}+1} = \emptyset$ for all $\tilde{t} \in \{t, \ldots, t + t' - 2\}$. Thus, $B$ has at least $2m$ neighbours in $B_{\tilde{t}}$ for all $\tilde{t} \in \{t, \ldots, t + t' - 1\}$, which implies $B_{t + t'} \supseteq B$. Therefore, $B_{t + t'}$ is $\Omega_1$-blocked.
    Now, assume that $Y \in F_{t + \tilde{t}}$ for all $\tilde{t} \in [T]$. Then there were at most $2m$ firefighters used in $\mathbb{G}$ in $F_{t+1}, \ldots, F_{t + T}$, so by Lemma \ref{lem:shadowG}, we have $B_{t + T} \cap \mathbb{G} \neq \emptyset$. Additionally, this implies $Y \subseteq B_{t + T}$. Since there are $2T + 2$ paths connecting $A$ and $B$ and we need to use at least $m$ firefighters to influence the spreading of the fire along one path, there are at least $2$ paths $P_k, P_\ell$ such that $(P_k \cup P_\ell) \cap F_{t + \tilde{t}+1} = \emptyset$ for all $\tilde{t} \in [T - 1]_0$. Hence, we have $B_{t + T} \supseteq P_k \cup P_\ell$. In particular, this implies $P_k^1 \cup P_\ell^1 \subseteq B_{t + \tilde{t}}$ for all $\tilde{t} \in [T]_0$. Therefore, $B$ has at least $2m$ burning neighbours the whole time, and thus $B_{t + \tilde{t}} \supseteq B$ for all $\tilde{t} \in [T]_0$, which further implies $Z \subseteq B_{t + T}$.
    This finally shows that $B_{t + T}$ is $\Omega_4$-blocked.

    \medskip
    \noindent\textbf{Case 4:} Let $B_t$ be $\Omega_4$-blocked with $\Omega_4 = B \cup Y \cup Z \cup P_k \cup P_\ell \cup \{v\}$ where $v$ is any node from $\mathbb{G}$.
    Since $B$ (resp. $Z$) has at least $4m$ (resp. $3m$) burning neighbours in addition to its own size of $m$ (resp. $2m$), $E_{t + 1}$ cannot contain a node from $B$ (resp. $Z$). In addition, $Y$ has at least $2m+1$ burning neighbours in addition to its own size of $2m$ and hence $E_{t + 1}$ cannot contain a node from $Y$.

    \textit{Case 4.1:} $E_{t + 1}$ contains $v$.
    Then $\{v\} \cup Y \subseteq F_{t+1}$. Since any node in $A, P_k, P_\ell, B$ and $Z$ has at least $2m - 1$ neighbours and the nodes in $Y$ have $2m$ neighbours in $Z$, the remaining $2m - 1$ firefighters cannot stop the fire from spreading to those groups. Thus $B_{t + 1} \supseteq A \cup B \cup Y \cup Z \cup P_k \cup P_\ell$, and hence is $\Omega_5$-blocked.
    
    \textit{Case 4.2:} $E_{t + 1}$ contains a node from $P_k$ and a node from $P_\ell$.
    Then we have $F_{t+1} \subseteq P_k \cup P_\ell$. Thus $B_{t + 1} \supseteq \mathbb{G} \cup B \cup X \cup Y \cup Z$, and hence is $\Omega_1$-blocked.

    \textit{Case 4.3:} $E_{t + 1}$ contains a node from $P_k$ but no node from $P_\ell$ (or, symmetrically, a node from $P_\ell$ but no node from $P_k$).
    Then we have $|F_{t+1} \cap P_k| \geq 2m$. The final $2m$ firefighters can only extinguish more nodes in $P_k$ or stop the fire from spreading to $X$ and $\mathbb{G}$ by guarding $Y$. Thus $B_{t + 1} \supseteq A \cup B \cup Y \cup Z \cup P_\ell$, and hence is $\Omega_7$-blocked.

    \medskip
    \noindent\textbf{Case 5:} Let $B_t$ be $\Omega_5$-blocked with $\Omega_5 = A \cup B \cup Y \cup Z \cup P_k \cup P_\ell$.
    Since $B$ (resp. $Z$) has at least $4m$ (resp. $3m$) burning neighbours in addition to its own size of $m$ (resp. $2m$), $E_{t + 1}$ cannot contain a node from $B$ (resp. $Z$).

    \textit{Case 5.1:} $E_{t + 1}$ contains a node from $A$.
    Then $F_{t+1} = A \cup P_k^{1} \cup P_\ell^{1}$. Thus $B_{t + 1} \supseteq \mathbb{G} \cup B \cup X \cup Y \cup Z \cup P_k \cup P_\ell$ and hence is $\Omega_1$-blocked.

    \textit{Case 5.2:} $E_{t + 1}$ contains a node from $Y$.
    Then $F_{t+1} = Y \cup Z$. Thus $B_{t + 1} \supseteq A \cup B \cup X \cup Z \cup P_k \cup P_\ell$ and hence is $\Omega_6$-blocked.

    \textit{Case 5.3:} $E_{t + 1}$ contains a node from $P_k$ (or, symmetrically, $P_\ell$).
    In order to achieve this, we must have $|F_{t+1} \cap (P_k \cup A \cup B)| \geq 3m$. The final $m$ firefighters can only extinguish further nodes in $P_k$, since all other nodes have at least $m$ neighbours in $B_t \setminus (P_k \cup A \cup B)$. Then, we have $B_{t + 1} \supseteq \mathbb{G} \cup A \cup B \cup X \cup Y \cup Z$, and hence is $\Omega_1$-blocked.

    \medskip
    \noindent\textbf{Case 6:} Let $B_t$ be $\Omega_6$-blocked with $\Omega_6 = A \cup B \cup X \cup Z \cup P_k \cup P_\ell$.
    Since $A$ and $B$ both have at least $4m$ burning neighbours in addition to their own size of $m$, $E_{t + 1}$ cannot contain a node from $A$ or $B$.

    \textit{Case 6.1:} $E_{t + 1}$ contains a node from $X$ (or, symmetrically, $Z$).
    Then $X \cup A \subseteq F_{t+1}$. Since any other node has at least $m$ neighbours in $B_t \setminus (X \cup A)$, the position of the last $m$ firefighters does not matter.
    Then, we have $B_{t + 1} \supseteq A \cup B \cup Y \cup Z \cup P_k \cup P_\ell$ and hence is $\Omega_5$-blocked.
    
    \textit{Case 6.2:} $E_{t + 1}$ contains a node from $P_k$ (or, symmetrically, $P_\ell$).
    In order to achieve this, we must have $|F_{t+1} \cap (P_k \cup A \cup B)| \geq 3m$. The final $m$ firefighters can only extinguish further nodes in $P_k$, since all other nodes have at least $m$ neighbours in $B_t \setminus (P_k \cup A \cup B)$. Then, we have $B_{t + 1} \supseteq A \cup B \cup X \cup Y \cup Z \cup P_\ell$, and hence is $\Omega_2$-blocked.

    \medskip
    \noindent\textbf{Case 7:} Let $B_t$ be $\Omega_7$-blocked with $\Omega_7 = A \cup B \cup Y \cup Z \cup P_k \cup (P_\ell \setminus P_\ell^{1}) \cup \{v\}$ where $v$ is any node from $ \mathbb{G}$.
    Since $B$ (resp. $Y$/$Z$) has at least $4m$ (resp. $2m + 1$ / $3m$) burning neighbours in addition to its own size of $m$ (resp. $2m$ / $2m$), $E_{t + 1}$ cannot contain a node from $B$ (resp. $Y$/$Z$).

    \textit{Case 7.1:} $E_{t + 1}$ contains a node from $A$.
    Then $A \cup P_k^{1} \subseteq F_{t+1}$. With the last $2m$ firefighters there are multiple possibilities.
    Either we have $Y \subseteq F_{t+1}$, then we have $B_{t+1} \supseteq B \cup Y \cup Z \cup P_k \cup P_\ell \cup \{v\}$ and hence is $\Omega_4$-blocked.
    Otherwise, we have $B_{t + 1} \supseteq \mathbb{G} \cup B \cup X \cup Y \cup Z$ which is $\Omega_1$-blocked.

    \textit{Case 7.2:} $E_{t + 1}$ contains $v$.
    Then $Y \cup \{v\} \subseteq F_{t+1}$. Independent of the positions of the remaining firefighters, we have $B_{t + 1} \supseteq A \cup B \cup X \cup Y \cup Z \cup P_k \cup P_\ell$ which is $\Omega_5$-blocked.

    \textit{Case 7.3:} $E_{t + 1}$ contains a node of $P_k$ or $P_\ell$.
    Then we must have $|F_{t+1} \cap (P_k \cup P_\ell)| \geq 2m$.
    If we have $Y \subseteq F_{t+1}$, then we have $B_{t + 1} \supseteq A \cup B \cup X \cup Y \cup Z$ which is $\Omega_2$-blocked.
    Otherwise, $B_{t + 1} \supseteq \mathbb{G} \cup B \cup X \cup Y \cup Z$ which is $\Omega_1$-blocked.
\end{proof}

\noindent To finish this section, we prove a result on the hardness of $\mfft$ on trees.
\thmMffinTIMEHard*
\begin{proof}
    We prove this via a reduction from \textsc{3-Partition}. which is \texttt{strongly NP-hard}, see \cite{10.5555/574848}.
    For some $k \in \N_{>0}$, let $a_1, \dots, a_{3k} \in \N_{>0}$ be the positive integer numbers in a given instance of \textsc{3-Partition}, and set $s = \sum_{i = 1}^{3k} a_i$.
    Without loss of generality, we may assume $\tfrac{s}{k} \in \N_{>0}$. Otherwise, a 3-partition of the numbers trivially cannot exist.

    We now construct a graph $G$ which we claim is $(\tfrac{s}{k} + 3s + 1)$-winning in time $k$ if and only if there exists a 3-partition of $a_1, \dots, a_{3k}$.
    To this end, let $T_i$ be an arbitrary tree with $a_i + s$ nodes for each $i \in [3k]$.
    Then the graph $G$ arises by adding a new node $c$ and, for each $i \in [3k]$, adding an edge between $c$ and an arbitrary node from $T_i$ as visualized in Figure \ref{fig:TreeConstruction}.
    Note that $|G| = s + 3sk + 1$.
    By choosing a star graph (resp. a path graph) for each $T_i$ and attaching $c$ to the internal node (resp. to an end of the path), we get the result for trees with diameter $\leq 4$ (resp. for spiders).
\begin{figure}[H]
    \centering
\begin{tikzpicture}[scale=0.6, every node/.style={font=\scriptsize}]

        %First node, X and some edges
        \node[draw, line width=0.4mm, circle, minimum size=0.2cm, inner sep=0pt, fill=white] (v1) at (0,-0.5) {};

        \draw[] (v1) -- (-0.25,0.2);
        \draw[] (v1) -- (0,0.2);
        \draw[] (v1) -- (0.25,0.2);

        %Complete Graphs
        \node[draw, line width=0.4mm, circle, minimum size=0.7cm, inner sep=0pt] (K1) at (-2,1) {$T_1$};
        \node[draw, line width=0.4mm, circle, minimum size=0.7cm, inner sep=0pt] (K4) at (2,1) {$T_{3k}$};
        \fill (-0.6,1) circle (0.1cm);
        \fill (0.0,1) circle (0.1cm);
        \fill (0.6,1) circle (0.1cm);
        \draw[] (v1) -- (K1);
        \draw[] (v1) -- (K4);

\end{tikzpicture}
    \caption{Construction of $G$. Every $T_i$ is an arbitrary tree with $a_i+s$ nodes.} 
    \label{fig:TreeConstruction}
\end{figure}
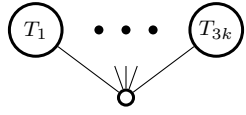

    First, if there exists a 3-partition $(i_1, i_1', i_1''), \dots, (i_k, i_k', i_k'')$ of $a_1, \dots, a_{3k}$ such that $i_j, i_j', i_j'' \in [3k]$, then $S = (T_{i_1} \cup T_{i_1'} \cup T_{i_1''} \cup \{c\}, \dots, T_{i_k} \cup T_{i_k'} \cup T_{i_k''} \cup \{c\})$ is a winning $(\tfrac{s}{k} + 3s + 1)$-strategy of length $k$.

    Next, assume that the graph $G$ is $(\tfrac{s}{k} + 3s + 1)$-winning in time $k$.
    $G$ has $s + 3sk + 1$ nodes, which is greater than $(k - 1) \cdot (\tfrac{s}{k} + 3s + 1) = s + 3sk + k - \tfrac{s}{k} - 3s - 1$, since $s \geq 3k$. Hence, any winning $(\tfrac{s}{k} + 3s + 1)$-strategy has to have a length of at least $k$, since otherwise, at least one node of $G$ would never even appear in a single firefighter set.

    Therefore we have $\tilde{B}_t \neq \emptyset$ for any $t \in [k - 1]$, which implies $|B_t| \geq |\tilde{B}_{t}| + 1$, since $G$ is connected.
    If there exists a $t'$ with $|B_{t'}| \geq |\tilde{B}_{t'}| + 2$, we would have
    \begin{align*}
        |\tilde{B}_k| &\geq |B_{k - 1}| - |F_{k}|
        \geq |B_{k - 1}| - (\tfrac{m}{k} + 3m + 1) \\
        &\geq |\tilde{B}_{k - 1}| + 1 - (\tfrac{s}{k} + 3s + 1) \\
        &\geq \dots \\
        &\geq |B_{t'}| + (k - 1 - t') - (k - t') \cdot (\tfrac{s}{k} + 3s + 1) \\
        &\geq |\tilde{B}_{t'}| + (k + 1 - t') - (k - t') \cdot (\tfrac{s}{k} + 3s + 1) \\
        &\geq \dots \\
        &\geq |\tilde{B}_1| + k - (k - 1) \cdot (\tfrac{s}{k} + 3s+ 1) \\
        &\geq |B_0| + k - k \cdot (\tfrac{s}{k} + 3s + 1) \\
        &= (s + 3sk + 1) + k - s - 3sk - k = 1,
    \end{align*}
    which means that the used strategy is not winning in time $k$.
    We therefore must have
    \begin{equation}
    \label{eq:burnsetdiff}
        |B_t| = |\tilde{B}_t| + 1 \text{ for all } t \in [k - 1].
    \end{equation}
    Note that in this case, the above inequality yields $|\tilde{B}_k| \geq 0$, which implies that all the inequalities must actually be equalities in order for the strategy to be winning in time $k$.
    Hence, we have
    \begin{equation}
    \label{eq:ffsetsize}
        |F_t| = \frac{s}{k} + 3s + 1 \text{ for each } t \in [k].
    \end{equation}
    Therefore, each $F_t$ contains nodes of at least three different $T_i$, and fully contains at most three different $T_i$.

    Now assume that there is a $\tilde{t} \in [k]$ such that $F_{\tilde{t}}$ is not equal to $\{c\} \cup T_i \cup T_{i'} \cup T_{i''}$ for some pairwise different $i, i', i'' \in [3k]$.
    Without loss of generality, we assume $\tilde{t}$ to be as small as possible, which implies $B_{\tilde{t} - 1} = \{c\} \cup \bigcup_{i \in I} T_i$ for some nonempty $I \subseteq [3k]$.
    Note that this implies $\tilde{t} < k$, since we must have $F_k = B_{k - 1}$.
    By our choice of $\tilde{t}$, $F_{\tilde{t}}$ has to fulfill $c \notin F_{\tilde{t}}$ or there is an $i \in [3k]$ such that $\emptyset \neq F_{\tilde{t}} \cap T_i \subsetneq T_i$.

    Let $i, i', i'' \in [3k]$ be such that $F_{\tilde{t}}$ contains nodes of $T_i$, $T_{i'}$ and $T_{i''}$.
    If $c \notin F_{\tilde{t}}$, then $\tilde{E}_{\tilde{t}}$ contains at least one node from $T_i$, $T_{i'}$ and $T_{i''}$ each that is adjacent to a node in $\tilde{B}_{\tilde{t}}$. This implies $|B_{\tilde{t}}| \geq |\tilde{B}_{\tilde{t}}| + 3$, which contradicts (\ref{eq:burnsetdiff}).
    
    Now let $c \in F_{\tilde{t}}$. If at least two subtrees $T_i$ and $T_{i'}$ are only partially contained in $F_{\tilde{t}}$, there are at least two nodes in $\tilde{E}_{\tilde{t}}$ that are adjacent to nodes in $\tilde{B}_{\tilde{t}}$, which again contradicts (\ref{eq:burnsetdiff}). So let us assume now that exactly one subtree $T_i$ is only partially contained in $F_{\tilde{t}}$.
    
    If there exists a subtree $T_j \subseteq B_{\tilde{t} - 1}$ that does not intersect $F_{\tilde{t}}$, it follows that there are at least two nodes in $\tilde{E}_{\tilde{t}}$ that are adjacent to nodes in $\tilde{B}_{\tilde{t}}$, namely $c$ and at least one node in $T_i$.
    This again contradicts (\ref{eq:burnsetdiff}).
    If such a subtree $T_j$ does not exist, it follows that $B_{\tilde{t}} \subseteq \{c \} \cup T_i$.
    Then, we could successfully finish the strategy with $F_{\tilde{t} + 1} = B_{\tilde{t}}$, which contradicts (\ref{eq:ffsetsize}).

    Thus, we have shown that for any $t \in [k]$, we have $F_t = \{c \} \cup T_{i_t} \cup T_{i'_t} \cup T_{i''_t}$ for some $i_t, i'_t, i''_t \in [3k]$.
    We need to visit each subtree at least once, therefore $\bigcup_{t \in k} T_{i_t} \cup T_{i'_t} \cup T_{i''_t}$ is a disjoint union of disjoint groups of three pairwise different subtrees each, as there are $3k$ subtrees and only $k$ steps.
    By (\ref{eq:ffsetsize}), we have $|T_{i_t} \cup T_{i'_t} \cup T_{i''_t}| = \frac{s}{k} + 3s$ for each $t \in [k]$, which implies  $(|T_{i_t}| - s) + (|T_{i'_t}| - s) + (|T_{i''_t}| -s) = \frac{s}{k}$.
    Therefore, $(i_1, i_1', i_1''), \dots, (i_k, i_k', i_k'')$ yields a 3-partition of $a_1, \dots, a_{3k}$.
\end{proof}

\subsection{Omitted Proofs of Section \ref{sec:shortStrat}: Graphs with Long Shortest Strategies}\label{app:shortStrats}
We start this section by showing a useful property for the auxiliary graph $H_m$.
Note that Definitions \ref{def:aux} and \ref{def:GkX} depend on the parameters $\alpha, \beta \in \N_{>0}$, which have fixed values (in particular, they do not depend on $m$ or $X$).
The only necessary property of these parameters is the inequality $2\beta + 2 \geq \alpha \geq \beta + 3$. Actually, we only need $2 \beta + 2 \geq \alpha \geq \beta + 2$, which would allow $\beta = 1, \alpha = 3$.
However, having $\alpha \geq \beta + 3$ makes some proofs slightly less convoluted. The smallest possible choice would therefore be $\beta = 1, \alpha = 4$.
We do not use explicit values for the parameters in the upcoming constructions and proofs, since that would not lead to relevant simplifications.
\begin{lemma}[Properties of $H_m$]
\label{lem:Hkstrategy}
    Let $S = (F_1, \dots, F_T)$ be a $m$-winning strategy for $H_m$ with $B_t \notin \{\emptyset, H_m\}$ for all $t \in [T - 1]$.
    Then we have $T \geq \alpha$, and each firefighter set uses at least $m - 1$ firefighters.
    Furthermore, at least $\alpha$ firefighter sets need to use all $m$ firefighters.
    In particular, we have $|F_1| = m$.
\end{lemma}
\begin{proof}
    Denote by $K_{m - 1}$ the $(m - 1)$-clique contained in $H_m$.
    From the assumption that $B_t \neq \emptyset$ for any $t \in [T - 1]$ and the fact that any node in $K_{m - 1}$ is adjacent to any other node of $H_m$, it follows that $K_{m - 1} \subseteq B_{t - 1}$ for any $t \in [T]$.
    Hence, if $K_{m - 1} \not\subseteq F_t$ for some $t \in [T]$, we have $B_t = H_m$, a contradiction.
    Therefore, we have $K_{m - 1} \subseteq F_t$ for any $t \in [T]$, which leaves at most one firefighter in $H_m \setminus K_{m - 1}$ in each turn.
    Since any node of a graph must be part of at least one firefighter set in a winning strategy and $|H_m \setminus K_{m - 1}| = \alpha$, it follows that $T \geq \alpha$ and $|F_t| = m$ for at least $\alpha$ different $t \in [T]$.
    Finally, we must have $|F_1| = m$, since otherwise $F_1 = K_{m - 1}$ which would imply $B_1 = H_m$.
\end{proof}

\noindent In order to calculate $\ffn(G(m, X))$, we use the following two lemmata describing parts of a winning $m$-strategy.

\begin{lemma}[Extinguishing of $H_m$]
\label{lem:Hkextinguishable}
    For any $i \in [m]$, one can extinguish $H^i_m \setminus \{u_i\}$ in $\alpha$ steps, independent of the state of the rest of the graph $G(m, X)$.
\end{lemma}

\begin{proof}
    Let $K_{m - 1}$ denote the $(m - 1)$-clique contained in $H^i_m$ and $\{w_1, \dots, w_\alpha\} = H^i_m \setminus K_{m - 1}$.
    Then the strategy $(F_1, \dots, F_\alpha)$ with $F_j = K_{m - 1} \cup \{w_j\}$ for $j \in [\alpha]$ achieves what was claimed to be possible in the statement of the lemma, since $u_i$ is the only node connected to the rest of the graph $G(m, X)$ and $u_i$ is contained in every $F_j$.
\end{proof}

\begin{lemma}[Property of $H_m$]
\label{lem:spiderextinguishable}
     If at some time $t$, we have some $I \subseteq [m]$ with $1 \leq |I| < m$ such that $(H_m^i \setminus \{u_i\}) \subseteq E_t$ for all $i \in I$, then there is a strategy that achieves $B_{t'} \subseteq \bigcup_{i \in [m] \setminus I} H_m^i$ for some $t' \geq t$.
\end{lemma}
\begin{proof}
    For better clarity, we will describe such a strategy in words instead of explicitly stating the firefighter sets.

    Without loss of generality, we assume $I = [j]$ for some $j < m$.
    Since $|I| < m$, we can position one firefighter at each $u_i$ for all $i \in I$ and still have at least one additional firefighter left.
    Using this additional firefighter as well as the one positioned at $u_1$, we can extinguish the path from $u_1$ to $c$ without letting the fire spread to $H_m^1$.
    For the rest of the strategy, we permanently station one of those two firefighters at $c$, which ensures that the fire cannot spread to $H_m^1$.
    This again leaves at least one firefighter free, allowing us to extinguish the path from $u_2$ to $c$, and afterwards removing the need to keep a firefighter stationed in $u_2$.
    We repeat this until all paths connecting $c$ to some $H^i_m$ for $i \in I$ are extinguished.

    Next, we keep one firefighter positioned in $c$ and use the remaining $m - 1$ firefighters to extinguish $X$, which is possible since $\ffn(X) = m - 1$ and the only node of $G(m, X) \setminus X$ connected to $X$ is $c$.

    After that, we still keep one firefighter in $c$ and start extinguishing the paths connecting $c$ to $u_i$ for $i \in [m] \setminus I$ one by one.
    Whenever we finished extinguishing one of these paths, we keep one firefighter in the corresponding $u_i$ to prevent the path from being reignited.
    Since $|[m] \setminus I| \leq m - 1$, we still have at least $2$ firefighters available to extinguish each path, until we reach the final path.
    When we start extinguishing the final path, we don't need to keep a firefighter in $c$ anymore, which again guarantees that we have at least $2$ firefighters available, and hence that this task is possible.
    After extinguishing the final path, we have restricted the burning set as desired in the lemma.
\end{proof}

\noindent We are now ready to determine the firefighter number of $G(m, X)$.
The general idea for a winning $m$-strategy is to alternate between extinguishing a new $H^i_m$ and extinguishing $G(m, X) \setminus \bigcup_{i \in [m]} H^i_m$ without letting the previously extinguished $H^i_m$ reignite.

\lemFfnofGkX*
\begin{proof}
    By construction, the graph $G(m, X)$ contains the subgraph $H_m$, which in turn contains a $m$-clique. This shows $\ffn(G(m, X)) \geq m$ due to Lemma \ref{lem:lowerbounds}.3.

    To show $\ffn(G(m, X)) \leq m$, let us give a winning $m$-strategy.
    By Lemma \ref{lem:Hkextinguishable}, we can start by extinguishing $H_m^1 \setminus \{u_1\}$.
    From Lemma \ref{lem:spiderextinguishable}, it follows that we can then reach the state where any burning nodes are contained in $\bigcup_{i \in [m] \setminus [1]} H_m^i$.
    For any $i \neq j$, $i, j \in [m]$, it takes $2 \beta + 2$ steps for the fire to reach $u_i$ if it starts at $u_j$ and no firefighters are used to stop or slow the spread.
    Since $\alpha \leq 2 \beta + 2$, this means that we can now, by Lemma \ref{lem:Hkextinguishable}, extinguish $H^2_m \setminus \{u_2\}$ in $\alpha$ steps so that afterwards, both $H^1_m \setminus \{u_1\}$ and $H^2_m \setminus \{u_2\}$ are extinguished.
    Again using Lemma \ref{lem:spiderextinguishable}, we can reach the state where any burning nodes are contained in $\bigcup_{i \in [m] \setminus [2]} H_m^i$.
    
    By repeating this argument $m - 1$ times, we can achieve $B_t = \bigcup_{i \in [m] \setminus [m - 1]} H_m^i$ $= H_m^m$ and then extinguish $H_m^m$ using the strategy given in the proof of Lemma \ref{lem:Hkextinguishable}, which does not let the fire spread back to $G(m, X) \setminus H_m^m$ and thus completely extinguishes the graph.
\end{proof}

\noindent After determining the firefighter number of $G(m, X)$, we shall now find a lower bound to the number of steps needed to extinguish $G(m, X)$ with that number of firefighters.
Our strategy for this is the following:
In Lemma \ref{lem:Hkordering}, we prove that the different $H_m$ contained in $G(m, X)$ need to be extinguished one after the other, roughly speaking.
The Lemmata \ref{lem:Xkwillburn} and \ref{lem:Xwasextinguished} then show that between extinguishing two different $H_m$, we always need to extinguish $X$, which needs to be done with just $m - 1$ firefighters due to Lemma \ref{lem:alwaysaffinc}.

\begin{lemma}[Properties of Shortest Strategies of $G(m, X)$]
\label{lem:Hkordering}
    Let $S = (F_1, \dots, F_T)$ be a shortest winning $m$-strategy for $G(m, X)$.
    For each $i \in [m]$, set $t_i \coloneqq \max \{t \in [T]: H_m^i \subseteq B_t\}$ and $t_i' \coloneqq \min \{t \in [T]: t > t_i, H_m^i \subseteq E_t\}$.
    Then, we have
    \begin{enumerate}[a)]
        \item $H_m^i \subseteq B_{t_i}$,
        \item $H_m^i \subseteq E_{t'_i}$,
        \item $t_i' \geq t_i + \alpha$ and
        \item $|F_t \cap H_m^i| \geq m - 1$ for all $t_i < t \leq t_i'$
    \end{enumerate}
    for all $i \in [m]$.
    Furthermore, with an appropriate reordering of the $H_m^i$, we have $t_{i + 1} \geq t_i'$ for all $i \in [m - 1]$.
\end{lemma}

\begin{proof}
    Since the graph is fully burning at time $0$ and fully extinguished at time $T$, the numbers $t_i$ and $t_i'$ are well-defined.

    Statements a) and b) hold by the definitions of $t_i$ and $t_i'$.
    Furthermore, for every $t \in \{t_i + 1, \dots, t_i' - 1\}$, we have $B_t \cap H_m^i \neq \emptyset$ and $E_t \cap H_m^i \neq \emptyset$, since otherwise, a contradiction to the maximality and minimality of $t_i$ and $t_i'$ would arise.
    Together with Lemma \ref{lem:Hkstrategy}, this implies statements c) and d).

    Since the $H_m^i$ are pairwise disjoint and $|F_t| = m < 2m - 2$ for all $t \in [T]$, d) implies that the sets $\{t_i + 1, \dots, t_i'\}$ are pairwise disjoint for $i \in [m]$.
    Therefore we may, without loss of generality, assume $t_1 < t_1' \leq \dots \leq t_m < t_m'$ by appropriately reordering the indices of $H_m^1, \dots, H^m_m$.
\end{proof}

\begin{lemma}[$X$ is Burning]
\label{lem:Xkwillburn}
    $X \subseteq B_{t_i'}$ for all $i \in [m - 1]$.
\end{lemma}

\begin{proof}
    Let $i \in [m - 1]$ be arbitrary.
    By the properties of $t_1, t_1', \dots, t_m, t_m'$ as given by Lemma \ref{lem:Hkordering}, we know that at time $t_i$, the remaining strategy $(S_i, \dots, S_T)$ needs to fully extinguish both $H_m^i$ and $H^m_m$ from a fully burning state, which, according to Lemma \ref{lem:Hkstrategy}, takes at least $2 \alpha$ steps.
    Therefore, there must be at least one node $v \in B_{t_i} \setminus H^m_i$, since one could otherwise simply extinguish $H^m_i$ in $\alpha$ steps without letting the fire spread, creating a shorter strategy.
    
    Now, let $w$ be an arbitrary node in $X$.
    Then there exists a path $P$ containing $j \leq 4 + \beta$ nodes that connects $v$ and $w$ and does not intersect $H_m^i$.
    Whenever we have $|F_t \cap P| = 0$ and $1 \leq |P \cap B_{t - 1}| \leq |P| - 1$, it follows that $|P \cap B_t| \geq |P \cap B_{t - 1}| + 1$.
    Moreover, if $|F_t \cap P| = 1$ and $|P \cap B_{t - 1}| \geq 2$, we have $|P \cap B_t| \geq |P \cap B_{t - 1}|$.
    Together with Lemma \ref{lem:Hkstrategy} and the fact that $|P| \leq 4 + \beta$ is not greater than $\alpha + 1$, this proves that $|P \cap B_{t'_i}| = |P|$, and in particular $w \in B_{t'_i}$.
\end{proof}

\begin{lemma}[$X$ is Extinguished]
\label{lem:Xwasextinguished}
    $X \subseteq E_{t_i}$ for all $i \in [m] \setminus \{1\}$.
\end{lemma}

\begin{proof}
    Assume that there exists a node $v \in X \cap B_{t_i}$ for some $i \in [m] \setminus \{1\}$.
    Furthermore, let $w$ be an arbitrary node in $H_m^1$.
    By repeating the same argument as in the proof of Lemma \ref{lem:Xkwillburn}, i. e., using that there is a path between $v$ and $w$ that does not intersect $H^1_m$ and contains at most $4 + \beta \leq \alpha + 1$ nodes, we see that we have $w \in B_{t'_i}$, and thus $H_m^1 \subseteq B_{t'_i}$.
    This contradicts the definition of $t_1$, since $t'_i > t_1$.
\end{proof}

\begin{lemma}[Firefighters Cover $c$]
\label{lem:alwaysaffinc}
    Let $t', t'' \in [T]$ be such that $X \subseteq B_{t'} \cap E_{t''}$ and for any $t \in \{t' + 1, t'' - 1\}$, we have $X \not\subseteq B_t$ and $X \not\subseteq E_t$.
    Then $c \in F_t$ for all $t \in \{t' + 1, \dots, t''\}$.
\end{lemma}

\begin{proof}
    Let $t \in \{t' + 1, \dots, t''\}$ be arbitrary.
    By the definition of $t'$ and $t''$, there exists a node $v \in B_{t - 1} \cap X$.
    If $t = 1$, we immediately get $c \in B_{t - 1}$.
    Otherwise, there must be a neighbour of $v$ contained in $\tilde{B}_{t - 1}$.
    As any neighbour of $v$ is either $c$ itself or a neighbour of $c$, this again implies $c \in B_{t - 1}$.
    
    Since every node of $X$ shares an edge with $c$, $c \in \tilde{B}_t$ implies $X \subseteq B_t$, which would contradict the definition of $t'$ and $t''$.
    Therefore, we must have $c \in F_t$ for all $t \in \{t' + 1, \dots, t''\}$.
\end{proof}

\noindent Combining the previous lemmata gives us a lower bound to the time that is required in order to extinguish $G(m, X)$ with $m$ firefighters.
\lemGkXtakeslong*
\begin{proof}
    From the Lemmata \ref{lem:Hkordering}, \ref{lem:Xkwillburn} and \ref{lem:Xwasextinguished}, it follows that during any shortest winning $m$-strategy for $G(m, X)$, the subgraph $X$ changes its state from fully burning to fully extinguished at least $m - 1$ times.
    Furthermore, due to Lemma \ref{lem:alwaysaffinc}, at most $m - 1$ firefighters can be used to extinguish $X$ these $m - 1$ times, which proves that the strategy takes at least $(m - 1) \cdot T(X)$ steps.
\end{proof}
\end{document}